\documentclass[10pt]{article}
\usepackage{amsmath}    % need for subequations
\usepackage{cases}
\usepackage{amsfonts}
\usepackage{amssymb}  % needed for mathbb  OK
\usepackage{needspace}
\usepackage{bigints}
\usepackage{wrapfig}
\usepackage{graphicx}   % need for figures
\usepackage[table]{xcolor}
\usepackage{subfig}
\usepackage{verbatim}   % useful for program listings
\usepackage{color}      % use if color is used in text
\usepackage{parskip}
\usepackage{float}
\usepackage{courier}
\usepackage{exercise}
\usepackage{sistyle}
\usepackage{booktabs}
\usepackage{changepage}
\usepackage{tabularx, array}

\usepackage{tikz}
\newcommand*\circled[2][gray!20]{\tikz[baseline=(char.base)]{
    \node[shape=circle, draw, fill=#1, inner sep=1.5pt] (char) {#2};}}
\newcommand*\rectangled[2][gray!20]{\tikz[baseline=(char.base)]{
    \node[shape=rectangle, draw, fill=#1, inner sep=1.5pt] (char) {#2};}}
\SIthousandsep{,}
\usepackage{makeidx}
\makeindex
\usepackage[nottoc]{tocbibind}

\usepackage[colorlinks = true,
          linktocpage=true,
            pagebackref=true, % add back references to bibliography
            linkcolor = red,
            urlcolor  = blue,
            citecolor = red,
            anchorcolor = blue]{hyperref}
\definecolor{dkgreen}{rgb}{0,0.6,0}
\definecolor{gray}{rgb}{0.5,0.5,0.5}
\definecolor{lightgray}{rgb}{0.9, 0.9, 0.9}
\definecolor{mauve}{rgb}{0.58,0,0.82}
\definecolor{index}{rgb}{0.88,0.32,0}
\definecolor{green2}{rgb}{0.6,1,0.6}
\definecolor{orange2}{rgb}{1,0.8, 0.5}  % {1,0.8,0.8}
\definecolor{red2}{rgb}{1,0.75, 0.7}  % {1,0.8,0.8}
\definecolor{turquoise}{rgb}{0,1.0,1.0}
\definecolor{yellow1}{rgb}{1,1.0, 0.06}  %{06,1.0,0.0.6}
\definecolor{black2}{rgb}{0.95,0.95, 0.95}
\definecolor{cyan2}{rgb}{1,0.95, 1}
\definecolor{yellow2}{rgb}{1,1, 0.9}
\definecolor{blue2}{rgb}{0.95,1.0, 1}
\definecolor{green3}{rgb}{0.9,1, 0.9}

\usepackage{listings}
\usepackage{blindtext}
\usepackage{geometry}
\renewcommand{\arraystretch}{1.4} %%%
\newtheorem{theorem}{Theorem}[section]
\newtheorem{lemma}[theorem]{Lemma}

\newtheorem{statement}{Statement}[section]

\newtheorem{conjecture}{Conjecture}[section]

\newcommand{\qed}{\nobreak \ifvmode \relax \else
      \ifdim\lastskip<1.5em \hskip-\lastskip
      \hskip1.5em plus0em minus0.5em \fi \nobreak
      \vrule height0.75em width0.5em depth0.25em\fi}

\begin{document}

\hypersetup{linkcolor=blue}
%inserting a glossary entry in gloss: \gls{gls:keyword1} \\

\begin{center}
{\Large \bf{Piercing Gilbreath's Conjecture: From Deep Number Theory\vspace{0.5ex}\\Insights to Fintech  and Cybersecurity} 
%\quad \\ \addvspace{1ex} 
%with Spectacular Applications
}  \\
% \addvspace{1ex}
%Stochastic Processes and Simulations -- Volume 1
\addvspace{5ex}
\end{center}
\begin{center}
Vincent Granville, Ph.D. $|$  CAIO $|$ vincent@BondingAI.io\\
 \href{https://bondingAI.io/}{BondingAI.io}, version 1.0, July 2026  \\ 
\addvspace{5ex}
\end{center}

\hypersetup{linkcolor=red}

%\tableofcontents

%\vspace{6ex}

\begin{abstract}
I propose a new methodology to attack the fascinating Gilbreath's conjecture about prime numbers, 
first posted in 1878 and unsolved to this day. The problem statement is rudimentary: kids can understand it. 
However, despite decades of research, almost
no progress has been made. This paper changes the game~by
 presenting a new approach based on sieving, a number of new results with proof, a precise 
path to the solution, and solid references. It also introduces the concept of reverse sieving, along with
 applications to testing randomness, pattern and fraud detection, cybersecurity, synthetic data, sequence categorization and normalization, or to detect and quantify a new type of chaos in time series including Brownian motions. Magic primes, forbidden prime number constellations, cellular automata,  and reduction via classes of equivalent sequences, are some of the innovative and 
promising topics discussed in the paper.  
\end{abstract}

\tableofcontents

\section{Introduction}

The Gilbreath conjecture is a celebrated mathematical problem tied to prime numbers, first raised 
by François Proth in 1878 and unsolved to this day, despite numerous attempts and 
computational verification up to primes larger than $10^{14}$ in 2025~\cite{pouf25}. 
In 2023, Zachary Chase proved an analog for random integers with very small growth (much smaller than prime numbers), see~\cite{zachch23}. But the problem for the prime numbers sequence remains fully open.   
However, Chase's proof generated renewed interest and fruitful discussions, notably by 
well-known mathematician Juan Arias de Reyna, see 
\href{https://institucional.us.es/blogimus/en/2020/07/gilbreaths-conjecture/}{here}. 
In an earlier blog posted \href{https://11011110.github.io/blog/2011/02/19/gilbreath-made-practical.html}{here} in 2011, computer scientist David Eppstein shares a different version of the conjecture, applicable to 
\textcolor{index}{practical numbers},
 a type of integers derived from \textcolor{index}{Egyptian fractions} with a distribution similar to that of primes. See also recent paper by Leila Muney~\cite{llmu26}. 
\vspace{1ex}

\begin{table}[H]
\centering
\renewcommand{\arraystretch}{1.0}
\small
\setlength{\tabcolsep}{1.pt}
\parbox{.45\linewidth}{
\centering

%\begin{tabular}{ccccccccccccccccccccccc} 
%\begin{tabular}{|*{23}{>{\centering\arraybackslash}p{0.25cm}|}}
\begin{tabular}{*{23}{>{\centering\arraybackslash}p{0.25cm}}}
 %\hline
 % &  &  &        \\  [-2.5ex]
 % 1--grams & 2--grams & 3--grams & 4--grams \\%[0.5ex] 
% \hline 
\hline 
   &  &   &        \\  [-2.2ex]
%\hline
 & 2 & & 3 & & 5 & & 7 & & 11 & & 13 & & 17 & & 19 & & 23 & & 29 & & 31\\
\hline
  &  &  &  & & & & & & & & & &  & & & & & &   \\  [-2ex]
%10 & & 1&  & 2 & \\
6 & &  1 & &  2 & &  2 & &  4 & &  2 & &  4 & &  2 & &  4 & & 6 &  & 2\\  %
4 & & & 1 & & 0 & & 2 & & 2 & & 2 & & 2 & & 2 & & 2 & & 4 \\ %
2 & & & & 1 & & 2 & & 0 & & 0 & & 0 & & 0 & & 0 & &2 \\%
2 & & & & & 1 & & 2 & & 0 & & 0 & & 0 & & 0 & & 2\\
2 & & & & & & 1 & & 2 & & 0 & & 0 & & 0 & & 2\\
2 & & & & & & & 1 & & 2 & & 0 & & 0 & & 2\\
2 & & & & & & & & 1 & & 2 & & 0 & & 2\\
2 & & & & & & & & & 1 & & 2 & & 2\\ %
1 & & & & & & & & & & 1 & & 0\\
1 & & & & & & & & & & & \circled[green2]{1}\\
 %\hline
\end{tabular}
%\caption{Backend table {\fontfamily{ptm}\selectfont hash stem} (extract)}
\captionsetup{justification=centering}
\caption{Prime numbers, one of many\\sequences verifying the conjecture}
\label{table:gilb1}
}
%\hfill
\quad\quad
\parbox{.45\linewidth}{
\centering

%\begin{tabular}{ccccccccccccccccccccccc} 
%\begin{tabular}{*{23}{p{0.25cm}}}
%\begin{tabular}{|*{23}{>{\centering\arraybackslash}p{0.25cm}|}}
\begin{tabular}{*{23}{>{\centering\arraybackslash}p{0.25cm}}}
 %\hline
 % &  &  &        \\  [-2.5ex]
 % 1--grams & 2--grams & 3--grams & 4--grams \\%[0.5ex] 
% \hline 
\hline 
   &  &   &        \\  [-2.2ex]
%\hline
 & 2 & & 3 & & 5 & & 7 & & 11 & & 13 & & 17 & & 23 & & \textcolor{red}{31} & &\textcolor{red}{41} & & 43\\
\hline
  &  &  &  & & & & & & & & & &  & & & & & &   \\  [-2ex]
%10 & & 1&  & 2 & \\
10 & &  1 & &  2 & &  2 & &  4 & &  2 & &  4 & &  6 & &  8 & & \textcolor{red}{10} &  & 2\\
8 & & & 1 & & 0 & & 2 & & 2 & & 2 & & 2 & & 2 & & 2 & & \circled[yellow1]{\textcolor{red}{8}} \\
6 & & & & 1 & & 2 & & 0 & & 0 & & 0 & & 0 & & 0 & &\textcolor{red}{6} \\
6 & & & & & 1 & & 2 & & 0 & & 0 & & 0 & & 0 & & \textcolor{red}{6}\\
6 & & & & & & 1 & & 2 & & 0 & & 0 & & 0 & & \textcolor{red}{6}\\
6 & & & & & & & 1 & & 2 & & 0 & & 0 & &\textcolor{red}{6}\\
6 & & & & & & & & 1 & & 2 & & 0 & & \textcolor{red}{6}\\
6 & & & & & & & & & 1 & & 2 & & \textcolor{red}{6}\\
4 & & & & & & & & & & 1 & & \textcolor{red}{4}\\
3 & & & & & & & & & & & \circled[black2]{\textcolor{red}{3}}\\  % orange2
 %\hline
\end{tabular}
%\caption{Backend table {\fontfamily{ptm}\selectfont hash stem} (extract)}
\captionsetup{justification=centering}
\caption{Removing a few primes from\\ the sequence invalidates the result}
\label{table:gilb2}
}
\end{table}

\noindent Gilbreath's conjecture is illustrated in table~\ref{table:gilb1}: the successive differences of adjacent numbers,
 in absolute value, all start with 1 when the original sequence at the top consists of the prime numbers. Whether you look at the first 20 primes, the first $10^{14}$, or all of them. Remove a few primes in the list, and you get the counter-example
featured in table~\ref{table:gilb2}. The leftmost column (left to the triangle) represents the maximum value
 at each level. It contains all the critical information towards understanding the mechanics behind the scene. 

Surprisingly, the conjecture has nothing to do with prime numbers nor even their distribution. It works for plenty of sequences,
 some growing much faster, some much slower than primes, and even with oscillations or full chaos 
(Brownian motions, see figure~\ref{fig:aajj5gb17}), even with
 negative numbers. 
Yet, the proof sketch outlined in this article heavily relies on properties of prime numbers. It is as if the prime numbers are the ideal candidate, leading to simplifications not attainable with arbitrary numbers that also satisfy the conjecture. 
In the end, the core result is that in order to verify the conjecture, the sequence can be chaotic in certain specific ways (I called it
 constrained chaos), or smooth but again with specific constraints. The prime numbers happen to meet those constraints. 
Other curious sequences are listed in OEIS entry A358691, \href{https://oeis.org/A358691}{here}, including Fibonacci numbers leading
to periodicity in the left diagonal of the triangle instead of 1's.

Understanding what is going on in tables~\ref{table:gilb1} and~\ref{table:gilb2}, and what broke the nice pattern in the latter, is critical to build a proof of the conjecture. In short, in the bottom table, we skipped the primes 19, 29, and 37. This accelerated the growth 
and introduced higher values in the first differences, including a 10. But that alone did not kill the pattern. 
What killed it is the prime 43, not because it was too large compared to 41 (the classic~way to crash), but because it was too small. This resulted in a very high value in the second differences, given its short length: the number 8 circled in yellow. 
If you replace 43 by 45, it will work. 

In layman's term, you can visualize the problem as follows. You are a pilot about to take off. The runway in front of you is the triangle.
It ends in the bottom row: your either fly if the bottom number is 1, otherwise you crash.
The yellow light flashing on your dashboard tells you something is not right. The column left to the triangle tells your speed: the higher the value, the lower the speed. Your dashboard tells you that you are not accelerating fast enough. And at that low acceleration rate, you don't have enough runway to take off.

\section{Proof for sifted sequences: standard and reverse sieving}\label{siebel}

If the original sequence (the row above the triangle) has $n$ elements, then the corresponding triangle has $n-1$ rows called 
\textcolor{index}{levels},
labeled from 1 at the top to $n-1$ at the bottom. The maximum value at level $k$ 
is denoted as $\gamma(k)$ and called the \textcolor{index}{depth} at level $k$, with $1\leq k \leq n-1$. 
In my implementation, by design I stop growing the sequence (adding new elements to the right), and thus the size of the triangle, when
one of the following two conditions is met:
\vspace{1ex}
\begin{itemize}
\item We arrived at the bottom row: either with a 1 (called \textcolor{index}{success}), or with a number $>1$ (called
\textcolor{index}{failure}). 
\item At any level, the leftmost element is not 1 (called failure). In short, we stop at the first level where this problem occurs,
 if it ever does.
\end{itemize}
\vspace{1ex}
 
\noindent As a result, my triangles always have a 1 on the left at each level except possibly for the last level consisting of
 one element. For infinite sequences, I look at all finite subsequences of increasing length. To be marked as success, all
finite subsequences must be marked as success. I now state a few fundamental yet trivial results. 
\vspace{1ex}

\begin{theorem}\label{th1}
All integers in the triangle are even, except the first one on the left at each level, which is odd. 
Also, the depth function $\gamma(k)$ is a non-increasing function.  Finally, if we reach $\gamma(k_0)=2$ at some level
 $k_0$, we can stop: we know for certain that the end result will be a success. 
\end{theorem}
\vspace{1ex}
{\bf Proof}: For two positive integers $x, y$ we always have $|x-y|\leq \min(x, y)$.  
Thus from one level to the next, the values can not increase, insuring that the maximum value $\gamma(k+1)$
 at level $k+1$ cannot be larger than $\gamma(k)$. 
Now, if $\gamma(k_0)=2$ at some level $k_0$, then $\gamma(k)\leq 2$ for all $k>k_0$.
Since the length of each row in the triangle decreases by 1 each time we increase $k$, 
we eventually reach a level with 2 elements: the one on the left equal to 1, and the one the right equal to 0 or 2. 
The next level is the final one, and consists of one element equal to either $|0-1|$ or $|2-1|$. Either way,
 it is equal to 1, that is, success. $\square$
\vspace{2ex}\\
\noindent The minimum $k_0$ satisfying $\gamma(k_0)=2$, if it exists, is called the 
\textcolor{index}{regularity coefficient} of the sequence.
If there is no  such $k_0$, the sequence is said to be \textcolor{index}{irregular}. 
In case of success, typically (but not always), $\gamma(k)$ decays exponentially fast, but it can stay at $\gamma(k)=2$ for a large number of $k$ before ending on 1. Thus, stopping~the computations as soon as we find $k_0$ with $\gamma(k_0)=2$ can save
much time for very large sequences, or when testing a large number of small sequences. 
\vspace{1ex}
\begin{theorem}\label{abc}
Modifying, removing, or adding new elements to a sequence does not modify the sub-triangle on the left, 
 attached to the unchanged terms on the left  to where the changes start occurring.   
\end{theorem}
\vspace{1ex}
{\bf Proof}:  Compare tables~\ref{table:gilb1} and~\ref{table:gilb2}. 
 The first seven elements 2, 3, 5, 7, 11, 13, 17 are common to both sequences. The corresponding sub-triangles are
identical. Generalization is straightforward. This is true whether the~numbers are prime or not, and whether the sequence is finite or infinite. $\square$
\vspace{2ex}\\
\noindent Theorem~\ref{abc} is useful in the following context. Let's $p_n$ be the last element in a finite sequence
of length $n$ that is marked as successful. You want to add one element $p_{n+1}$. What are the possible values
 for $p_{n+1}$ so that the augmented sequence is also successful? In many cases, any odd $p_{n+1}$
 in $[\,p_n, 2p_n]$ will work. 
In even more cases, any odd $p_{n+1}$ in $[\,p_n, 1.2 \, p_n]$ works. 
This is interesting, since \textcolor{index}{Nagura’s theorem}, published in 1952, states that there is always a prime number
in the interval $[\, p_n, 1.2\, p_n]$ if $p_n\geq 25$. One would think that this proves Gilbreath's conjecture, as it guarantees that you can always find the next prime within the prescribed range to grow the sequence indefinitely, eventually covering all primes,  while preserving success at all times. However, the counter-example in table~\ref{table:gilb2} is a huge blow to this strategy. 

Examples like table~\ref{table:gilb2} are exotic and rare. 
Whether there are finitely or infinitely many of them is unknown. And while
you would expect such exceptions to occur only in short sequences that are otherwise well behaved, 
no one knows if it could also happen in sequences of arbitrary length. 
The idea might be worth pursing, but it does not appear as an easy path to prove the conjecture.
Before stating the next theorem, let's call an infinite sequence \textcolor{index}{$\Delta$-periodic} if its first-order differences 
eventually become periodic. 
\vspace{1ex}
\begin{theorem}\label{ddew}
If the sequence is $\Delta$-periodic, looking at the numbers up to the 
first full period cycle, is enough to decide about its success: it turns an
infinite problem into a finite one.  In particular, starting with the sequence of all integers $> 1$, 
then removing even numbers except $p_1 =2$, then removing multiples of 3 except $p_2=3$, 
then removing multiples of 5 except $p_3=5$, 
%then removing even numbers and adding back 2, then removing multiples of 3 and adding~back 3, 
%then removing multiples of 5 and adding back 5, 
yields a successful $\Delta$-periodic sequence denoted as $S_3$. 
This remains true if we continue the iterations with $p_4, p_5$ and so on up to $p_8=19$, leading to successful $S_8$. 
%2 3 5 7 11 13 17 19
%1 2 3 4 5  6   7   8
\end{theorem}
\vspace{1ex}
{\bf Proof}: All the differences (first, second order and so on) must have the same period. 
Thus all values within~each level will keep repeating themselves with the same period. You only need to look at the first cycle and pre-period values. This proves the
first part of the theorem.  Now solving the second part.
\begin{itemize}
\item Removing all even numbers and adding back 2 yields $\gamma(1)=2$. So, we are done at the first level thanks to theorem~\ref{th1}. 
The period has length 1 and starts at position 2 in the first order differences.  
\item Removing all multiples of 3 and adding back 3, after the previous step, yields $\gamma(1)=4$ and $\gamma(2)=2$. So, we are done at the second level. Now, the period has length 2 and starts at position 3 in the first order differences.  
\item
Removing all multiples of 5 and adding back 5, after the two previous steps, yields 
$\gamma(1)=6$, $\gamma(2)=4$~and $\gamma(3)=2$. So, we are done at the third level. 
This time the period is
$[4, 2, 4, 2, 4, 6, 2, 6]$, has length 8 and starts at position 4 in the first order differences. 
 %4 2 4 2 4 6 2 6
\item Beyond $p_3=5$, I verified the result on a computer up to $p_8=19$. A finite number of
steps is required for each $p_i$, see~the Python code in section~\ref{pythor}.  
\end{itemize}
\noindent This concludes the proof. Thanks to theorem~\ref{abc}, keeping 2, 3 and 5 before the period kicks in
does not change the success status. Note how fast the \textcolor{index}{depth function} decays. 
$\square$
\vspace{2ex}\\
\noindent Gilbreath's conjecture states that theorem~\ref{ddew}  is true no matter how long we go with prime sieving.  
By contrast to Plouffe who shows success for the finite sequence
consisting of the first $10^{14}$ primes~\cite{pouf25}, we proved success~for infinite sequences containing all primes but
 blended with a proportion $\rho_k$ of all compound integers, where
\begin{equation}
\rho_\kappa = \prod_{i=1}^\kappa \Big(1 - \frac{1}{p_i}\Big)\sim \frac{e^{-\gamma}}{\log \kappa}
\end{equation}
according to \textcolor{index}{Meterns' third theorem}. So far, theorem~\ref{ddew} is proved for $\kappa\leq 8$, and $\rho_8\approx 17\%$. 
The next theorem, while stating nothing new, starts our descent into deep number theory. The short proof is worth reading.
\vspace{-1ex}\\
\begin{theorem}\label{rotule}
Starting with all integers $>1$, let $S_\kappa$ be the sequence obtained by iteratively removing all multiples of $p_i$ except $p_i$, for $i=1,\dots\kappa$, where $p_1, p_2$ and so on are the prime numbers starting with
$p_1 = 2$. Then $S_\kappa$ is $\Delta$-periodic with period length
\begin{equation}
L(S_\kappa) = \phi\big(p_\kappa^{\#}\big)=\phi(p_1\cdot p_2 \cdots p_\kappa) = \prod_{i=1}^\kappa (p_i - 1).\label{machin}
\end{equation}
The period kicks in in the first order differences after the first $\kappa$ elements. Furthermore,
in the periodic cycle,~the depth at level 1 (maximum gap between successive values in $S_\kappa$) satisfies
\begin{equation}
\gamma_\kappa(1) = O\big(\kappa^2 \log^2 \kappa \big).\label{butor}
\end{equation}
%\vspace{1ex}
Here $\gamma_\kappa(\cdot)$ denotes the depth function attached to the sequence $S_\kappa$
 and $\gamma_\kappa(1)$ stands for the first level. 
\end{theorem}
\vspace{1ex}
{\bf Proof}: This is an application 
of  the \textcolor{index}{sieve of Eratosthenes}. The period length is known: 
$\phi(\cdot)$ is the \textcolor{index}{Eulter~totient function}.
The asymptotic result~(\ref{butor}) is a consequence of \textcolor{index}{Iwaniec's bound} established in 1978, 
and $\gamma_k(1)$ is the \textcolor{index}{primorial Jacobsthal function} with argument $\kappa$.  
The notation $p_\kappa^{\#}$ represents \textcolor{index}{primorial} numbers. $\square$
\vspace{2ex}\\
\noindent The sequence $S_\kappa$ consists of the integers coprime
to the primorial $p_\kappa^{\#}=p_1\cdots p_\kappa$. This is a topic of active~research in \textcolor{index}{Sieve theory}, 
including about its first and second order differences, whose absolute values match the entries in our triangle. 
Theorem~\ref{rotule} gives us the exact, finite number of terms we need to look at in $S_k$ and its differences of any order in absolute value, allowing us to ignore all the other terms.  
The period length $L(S_\kappa)$, as a function of $\kappa$, is listed on OEIS as sequence A005867,
see \href{https://oeis.org/A005867}{here}.
The primorial Jacobsthal function (largest gap in $S_\kappa$) is listed as sequence A048670, see \href{https://oeis.org/A048670}{here}. Its first few values are
\vspace{1ex}
\begin{adjustwidth}{1cm}{1cm}
2, 4, 6, 10, 14, 22, 26, 34, 40, 46, 58, 66, 74, 90, 100, 106, 118, 132, 152, 174, 190, 200, 216, 234, 258, 264, 282, 300, 312, 330, 354, 378, 388, 414, 432, 450, 476, 492, 510, 538, 550, 574, 600, 616, 642, 660, 686, 718, 742, 762, 798, 810, 834, 858, 876, 908, 926, 954
\end{adjustwidth}
\vspace{1ex}
\noindent The sieve in theorem~\ref{rotule} can be performed in random order: 
 the primes can be rearranged in any order, as~long as you cover all of them in
the long run. I call it \textcolor{index}{random prime sieving}, and  
I use the notation $S_2', S_3'$ and so on for the generated sequences produced this way. For instance, 
If $p_1 = 11, p_2=47$ and $p_3=3$, then $S'_3$ is based on these three primes.  
Sieving the primes in random order should not be confused with sieving random integers sequentially~\cite{ajs23}. 
There is no research papers about the former; the latter has applications in cryptography. 
For more about the patterns linked to sieving, see~\cite{cpsgg19}.  
First order differences are discussed in~\cite{amprinc}, focusing on large gaps between
consecutive primes, and in~\cite{mziller20x} focusing 
on consecutive numbers coprime to primorials (our main topic). 
The \textcolor{index}{Cramér-Shanks conjecture}~\cite{grvillex} also quantifies how fast these gaps may grow.
See also~\cite{tt14rf}.
Finally, theorem~\ref{rotule} can be generalized to random prime sieving, as follows. 
\vspace{2ex}

\begin{theorem}\label{rotx} 
Random prime sieving using any combination of $\kappa$ distinct primes $p_1,\dots, p_k$ in any order,
 leads to increasingly thinner successful $\Delta$-periodic sequences $S_1',\dots,S_\kappa'$ encompassing all primes, 
for any $\kappa \leq \kappa_0$. Here $\kappa_0$ 
 is a constant yet to be determined. The period length is
\begin{equation}
L(S'_\kappa) =\phi(p_1\cdot p_2 \cdots p_\kappa) = \prod_{i=1}^\kappa (p_i - 1).\label{machin2}
\end{equation}
Again, $\phi(\cdot)$ is the Euler totient function. The period applies to $S_\kappa'$ differences of any order, in absolute value.
\end{theorem}
\vspace{1ex}
{\bf Proof}: Take $\kappa_0=0$ and the theorem is proved! For $\kappa_0 = 1$, the proof should be quite easy:
the difficult step~is to prove that the sequence succeeds regardless of $p_1$ (the first prime being sieved out).  

\noindent Periodicity and length of the period is well known for all
$\kappa$. 
Note that formulas~(\ref{machin}) and~(\ref{machin2}) are identical except that now, the product $p_1\cdots p_\kappa$
is not a primorial in general.  
Proving success for all $\kappa \leq \kappa_0$ might not be easy 
for larger values of $\kappa_0$. I conjecture success for any $\kappa_0$ including $\kappa_0=\infty$. 
 $\square$
\vspace{2ex}\\
\noindent My conjecture stated in the proof of theorem~\ref{rotx} is stronger than Gilbreath. 
You don't need success for all $S_\kappa'$ but only for at least one infinite subsequence
 and
for at least one permutation of all primes (out of the infinitely many possibilities) when applying the random prime sieve.

\begin{table}[H]
\centering
\renewcommand{\arraystretch}{1.2}
\small
\parbox{.45\linewidth}{
\centering
\begin{tabular}{c|ccccccc}
\hline
$p$ & 3  & 5   & 7  & 11     & 13      & 17     & 19  \\ 
\hline
\hline
 $\lambda$ & 2  & 3 & 15  & 16     & 22      & 37     & 41  \\
$\gamma$ & 4  & 6 & 10  & 14     & 22      & 26     & 34  \\
$\tau$ & 4  & 9 & 50  & 30   & 1818   & 39,205   & 10,283\\ 
$\nu$ & 7 & 23 & 199 & 113  & 9439  & 217,127   & 60,043\\
\hline
\end{tabular}
\caption{Stats for standard sieve ($\kappa=8$)}\label{oxy11}
}
%\hfill
\quad\quad
\parbox{.45\linewidth}{
\centering
\begin{tabular}{c|ccccccc}
\hline
$p$ & 19  & 17  & 13     & 11      &7    &  5   &    3\\
\hline
\hline
$\lambda$  & 1  & 17   & 21     & 34     &31    & 31   &  41\\
$\gamma$  & 2   & 6    & 8     & 10     &12     & 18     & 34\\
$\tau$  & 2 & 282  & 543    & 964   & 828  & 17,248 &  10,283 \\
$\nu$  & 3 & 625 & 1307   & 2559   &2559  & 67,191  & 60,043\\
\hline
\end{tabular}
\caption{Stats for reverse sieve ($\kappa=8$)}\label{oxy12}
}
\end{table}

\noindent Tables~\ref{oxy11} and~\ref{oxy12} illustrate respectively theorem~\ref{rotule} and~\ref{rotx}.
For the latter, I sieved in reverse order (called \textcolor{index}{reverse prime sieving}), starting from the largest prime
under consideration,  $p_8=19$, down to the
smallest one $p_1=2$. All the generated sequences succeed. As you move from left to right in either table, each
sequence (summarized in the corresponding column) is a subsequence of the previous one. The rows are as follows:
\vspace{1ex}
\begin{itemize}
\item $\lambda$ is the smallest index $k$ for which $\gamma(k)=2$. 
\item $\gamma$ is a shortcut for $\gamma(1)$, the maximum value in the first order differences. 
\item $\tau$ is the index in the first order differences, for the first occurrence of the record $\gamma(1)$. 
In short, this is the index of the first value in the sequence followed by a the maximum gap. 
\item  $\nu$ is the left value in the sequence, leading to the first occurrence of maximum gap $\gamma(1)$ in the first order differences.
Not to be confused with OEIS \href{https://oeis.org/A002386}{A002386}. 
\end{itemize}
\vspace{1ex}

\noindent Row $\gamma$ in table~\ref{oxy11}
corresponds to entry \href{https://oeis.org/A048670}{A048670} in the OEIS encyclopedia  
(Jacobsthal function), when extended to all primes. The other rows do not have an entry in OEIS.
And row $\tau$ may not be a good candidate because it~depends on how you compute it. In my case, 
 after sieving all multiples of a prime $p$, I add $p$ back into the sequence, creating
 an offset for $\tau$, that increases with the number of sieved primes, as you move to the right
in table~\ref{oxy11}. Entries in table~\ref{oxy12} are not good OEIS
candidates, as they depend on which prime you start the reverse sieving. 

\section{Full list of short sequences and sequence categorization}\label{beurred}

I now shift perspective, working with the sequence of primes and other sequences with similar patterns.
In the end, what could validate Gilbreath's conjecture is not the random character of primes, but the opposite:
 the lack of randomness on some key metrics. More specifically, the unusual
 rarity of long runs in first order differences. Likewise, what would make an arbitrary sequence
 not to succeed despite meeting most requirements, is caused by specific non-random patterns that are 
incompatible with the non-random patterns found in primes, such as very long runs early on in first order differences followed by a
 steep cliff. This is the exact opposite of prime number behavior. It makes primes
 unlikely to fail Gilthreat's conjecture, or in case of failure, \textcolor{index}{failure points} are incredibly rare and the first
one must occur very late in the sequence, well beyond the first $10^{14}$ primes.

\subsection{Testing an exhaustive infinite list of admissible short sequences}

A sequence $S_n = (q_1,\dots, q_n)$ of positive integers is \textcolor{index}{admissible} if $q_1 = 2, q_2=3$, all subsequent elements are odd,
 and all differences in absolute value start with 1, except possibly the last one (the bottom of the triangle). 
The \textcolor{index}{seed} $\sigma$ is the second element in the row just above the bottom value.
Obviously, a sequence succeeds if and only if $\sigma\in \{0, 2\}$. Otherwise, $\sigma$ is an even number $>2$ 
and the sequence fails, unable to reach 1 at the bottom.  
Let $N(\sigma, n)$ be the number of \textcolor{index}{admissible sequences} $S_n$ with seed $\sigma$. 
Finally, by definition, a \textcolor{index}{valid sequence} is one that is admissible but also satisfies $q_k < q_{k+1}$ for all $k$. 
For finite sequences of any length $n$, we have:
\vspace{1ex}
\begin{theorem}\label{thxbar} For any fixed $n>2$, the number $N(\sigma, n)$ of admissible sequences $S_n$ with $n$ elements and seed~$\sigma$, is high when $\sigma=0$, highest when $\sigma=2$, then sharply decreases as $\sigma$ increases. The decrease is rather fast at first, then loses momentum and eventually stops. 
That is, for each $n$, there is a value $\sigma_0(n)$, such that $N(\sigma, n)$~is constant (not depending on $n$)
 if $\sigma\geq\sigma_0(n)$. In addition, for any fixed $n$,
%\begin{quote}
\vspace{1ex}
\begin{enumerate}
    \item[(1)] The number $N^+(n)$ of successful admissible sequences $S_n$ grows faster than $n!$ by orders of magnitude, while 
    the number $N^-(n)$ of failed ones is infinite.
   \item[(2)] The last element $q_n$ of any admissible sequence $S_n$, successful or not, 
satisfies $\sigma+3\leq q_n \leq \sigma + 2^n -1$.~Both the lower and upper bounds are attained for all $n,\sigma$,
and cannot be sharpened.  
  \item[(3)] By construction, admissible sequences satisfy $q_{k+1}\geq q_k$ while valid ones satisfy $q_{k+1} \geq q_k +2$. 
For any $n$ and $\sigma$, there are of course more
admissible sequences than valid ones. The ratio $\rho(\sigma, n)$ between these two numbers 
 oscillates around 0.30 for small values of $n$.
\end{enumerate}
%\end{quote}
\end{theorem}
\vspace{1ex}
{\bf Proof}: For any fixed $\sigma$, the first order differences of any admissible sequence $S_n$ stays within these two extremes
%\begin{itemize}
%\item Global minimum differences: $(1, 0, 0, 0, \dots, 0, 0, \sigma)$
%\item Global maximum differences: $(1, 2, 4, 8,\dots, 2^{n-2}, \sigma')$
%\end{itemize}
instances: $(1, 0, 0, 0, \dots, 0, 0, \sigma)$ and $(1, 2, 4, 8,\dots, 2^{n-2}, \sigma')$,
where $\sigma' = 2^{n-1} +\sigma -2$. Thus $q_n \geq \sigma+3$ and the bound cannot be increased.
Similarly, the upper bound for $q_n$ is
$$
q_n \leq  2 + \Bigg(\sum_{k=0}^{n-2} 2^k\Bigg) + (2^{n-1} + \sigma - 2) = \sigma + 2^n - 1.
$$
This upper bound is also always attained and cannot be lowered. Thus:
$
\sigma+ 3 \leq q_n \leq \sigma + 2^n - 1. 
$
Since $q_1=2$, $q_2=3$ and $q_3\leq q_4\leq \dots \leq q_n$, the
 number $N(\sigma, n)$ is always finite. 

\noindent For a fixed $n$, at equilibrium, when $N(\sigma, n)$ is constant 
regardless of $\sigma>\sigma_0(n)$,
 the following happens. If~you remove 
the last (rightmost) element 
in each sequence
 in two different sets corresponding to two distinct $\sigma, \sigma'>\sigma_0(n)$, the
resulting two sets
 of truncated sequences are perfectly identical in all respects. 
In short, the sequence generator is stuck: it has no more leeway to create new sequences,
 other than changing the last element $q_n$ in each, as a function of the seed.
At that point, $q_n$ largely dominates $q_1, q_2$ and so on. It means that the number of sequences is now 
the same for each new $\sigma>\sigma_0(n)$. Therefore, the number of failed sequences at any given $n$,
 computed as the aggregate over $\sigma = 4, 6, 8$ and so on, is infinite. 
The number of successful ones, aggregated over $\sigma=0,2$, is finite.

\noindent I now work on the remaining points in the theorem.
\vspace{1ex}
\begin{itemize}
\item The fact that $N(0, n) < N(1, n)$ is because $\sigma=1$ has more leeway than $\sigma=0$ to move up from the very bottom
of the triangle. This early advantage back-propagates from bottom to top at each level $n$ in the triangle.
\item For a fixed $\sigma$, as $n$ increases, the growth for $N(\sigma, n)$ appears to be
of the order $A^n B^{n^2}$, where $A, B$ are two constants barely depending on $\sigma$ (if at all), with $B\approx 1.45$.  
See table~\ref{table:tabuli1} with explanations.
\item For $\rho(\sigma, n)$ values, see table~\ref{table:tabuli12} with explanations.
It is not clear if this ratio stays around $0.30$ as $n$ gets larger, or if it decreases, possibly converging to 0
 while being about the same across all $\sigma$, at each level $n$. 
\end{itemize}
\vspace{1ex}
Of particular interest are the values $\sigma_0(n)$. They tell you
at what $\sigma$ the equilibrium is reached, given $n$. I~haven't had time to look at this yet. 
$\square$
\vspace{2ex}\\
\noindent Table~\ref{table:tabuli1} shows the ratios $N(\sigma, n)/N(\sigma, n-1)$. The ratio of successive ratios
 (the multiplicative equivalent of~first order differences), within each column, is about 1.45 and shows little variations.
This is where the number 1.45 comes from in the proof of theorem~\ref{thxbar}.  
Table~\ref{table:tabuli12} shows the proportion of valid sequences among admissible~ones, for various $\sigma$ and $n$.
That proportion seems to not depend on $\sigma$, and possibly not even on $n$.

%------------------------------------------------------------------
\begin{table}[H]
\centering
\renewcommand{\arraystretch}{1.0}
\small
\parbox{.45\linewidth}{
\centering

\begin{tabular}{c|rrrrr} 
\hline 
   &  &   &        \\  [-2.2ex]
%\hline
$n$  & $\sigma = 0$ & $\sigma = 1$ & $\sigma = 2$ & $\sigma = 3$  & $\sigma = 4$ \\
\hline
\hline
  &  &  &    \\  [-2ex]
5&	3.50&	3.33&	2.50&	2.50&	2.50\\
6&	5.14&	4.60&	3.60&	3.40&	3.40\\
7&	7.06&	6.93&	5.33&	5.06&	4.82\\
8&	10.57&	9.89&	8.39&	7.64&	7.35\\
9&	15.06&	14.74&	12.60&	11.94	&11.33\\
10&	21.96&	21.57&	19.11&	18.16&	17.56\\
11&	31.81&	31.70&	28.29&	27.72	&26.80\\
\hline
\end{tabular}
%\caption{Backend table {\fontfamily{ptm}\selectfont hash stem} (extract)}
\caption{Growth ratios $N(\sigma,n)/N(\sigma,n-1)$}
\label{table:tabuli1}

}
%\hfill
\hspace{-0.5cm}
\parbox{.45\linewidth}{
\centering

\begin{tabular}{c|rrrrr} 
\hline 
   &  &   &        \\  [-2.2ex]
%\hline
$n$  & $\sigma = 0$ & $\sigma = 1$ & $\sigma = 2$ & $\sigma = 3$  & $\sigma = 4$ \\
\hline
\hline
  &  &  &    \\  [-2ex]
5&	0.428&	0.300&	0.400&	0.400&	0.400\\
6&	0.333&	0.326&	0.333&	0.352&	0.352\\
7&	0.342&	0.291&	0.343&	0.314&	0.329\\
8&	0.308&	0.303&	0.303&	0.313&	0.313\\
9&	0.301&	0.298&	0.299&	0.298&	0.303\\
10&	0.296&	0.295&	0.293&	0.296&	0.296\\
11&	0.296&	0.288&	0.294&	0.285&	0.296\\
\hline
\end{tabular}
%\caption{Backend table {\fontfamily{ptm}\selectfont hash stem} (extract)}
\caption{Prop. $\rho(\sigma, n)$ of valid sequences}
\label{table:tabuli12}
}
\end{table}
%-----------------------------------

\subsection{Efficient sequence corridors}\label{effco}

I now introduce the concept of sequence corridor. 
An \textcolor{index}{efficient sequence corridor} is a narrow, bounded set of 
admissible sequences containing the one of particular interest, 
and where the density of failing sequences if low. The corridor is said to be \textcolor{index}{optimum}
 if it contains no failing sequences.
In our case, optimum means that for each triangle in the corridor, all the rows start with 1, including the bottom one. 
Failure, for a specific sequence or triangle, is when the value in the bottom row is an odd integer $> 1$. 
In our setting, all rows start with 1 except possibly the bottom one, regardless of success status.

I now build an optimum corridor that contains the sequence consisting of the first $n$ primes and no other integers, 
for $n=11$. Because $n$ if finite, this corridor contains only a finite number of sequences. It can easily be extended
 to $n=\infty$. Proving optimality at $n=\infty$ would imply that Gilbreath's conjecture is true. Sequences
in the corridor are called \textcolor{index}{corridor sequences}, and built as follows:
\vspace{1ex}
\begin{itemize}
\item A corridor sequence $S_n = (q_1, \dots, q_n)$ starts with $q_1=2$ and $q_2=3$. Then, followed by strictly increasing odd numbers. 
The rows in the corresponding triangle, except possibly the bottom one, must start with 1. That way, the sequence is both admissible and valid. 
\item The first $m$ terms $q_1,\dots, q_m$ with $m<n$ constitute the warm-up period. No extra conditions are imposed upon them. 
The parameter $m$ is called the \textcolor{index}{offset}. The corridor is governed by a
 smooth function $f$ called the \textcolor{index}{corridor function}, indicating the general trend of the sequence.
Here, $f(x) = x\log x$.
\item For all $k>m$, the values must satisfy
\begin{equation}
   \alpha_1 f(k) \leq  q_k \leq \alpha_2 f(k), \quad \quad \beta_1 q_{k-1} \leq q_k \leq \beta_2 q_{k-1}\label{eureke}
\end{equation}
The positive coefficients $\alpha_1,\alpha_2,\beta_1, \beta_2$ along with $m$ are called the \textcolor{index}{corridor parameters}. 
Also, in our case, $\alpha_1 = \beta_1 = 1$, so I focus on $\alpha_2$ and $\beta_2$ only.
\end{itemize}
\vspace{1ex}
The parameters must be chosen so that most sequences in the corridor succeed, and that the target sequence
of particular interest (here, the prime
 number sequence) is also in the corridor. By contrast, table~\ref{table:rabhyg43x} features a sequence outside the corridor for multiple reasons: not a valid one because not strictly increasing, also growing too fast, and finally the success at the bottom of the triangle is preceded by 3 failures.
Now I can state have~the following result. 
\vspace{1ex}
\begin{theorem}\label{thxcorr} The corridor
with offset $m=5$, parameters $\alpha_1 = 1.0,\, \alpha_2 = 1.3,\, \beta_1 = 1.0, \,\beta_2 = 1.5$,
 and governed by the function $f(x)=x \log(x)$,  contains the prime number sequence. Also,
 for all sequences of length $n=11$~in the corridor, only one is failing; it is featured
in table~\ref{table:gilb3}.
\end{theorem}
\vspace{1ex}
{\bf Proof}: To show that the prime number sequence is in the corridor, replace
$q_k, q_{k-1}$ by the successive primes $p_k, p_{k-1}$ in~(\ref{eureke}) 
and check that the formula holds when $k>m$ with the prescribed set of parameters. 
The~rightmost inequality in~(\ref{eureke}) holds thanks to \textcolor{index}{Nagura's theorem} (year 1952),
that guarantees the existence of a prime 
between $p$ and $1.2 \, p$ for any prime $p>23$. 
Sharper bounds have been found more recently, including Schoenfeld (1976),
Dusart~\cite{pdx18} and Axler~\cite{colgate22}. 
I also checked that the prime sequence 
satisfies the corridor requirements for
the first $10^5$ primes. See the function \texttt{test\_primes} in
the Python code.

\noindent To prove that the sequence in table~\ref{table:gilb3} is the only exception with $n=11$ in the efficient corridor, 
I generated all admissible sequences of length $n$, with $\sigma\in \{4, 6, 8, 10\}$. 
There are about $10^7$ of them. Besides the example in table~\ref{table:gilb3}  which has $\sigma=4$, I found
no other exceptions in the corridor. 
Then, by digging into the proof of theorem~\ref{thxbar}, I concluded that exceptions are not possible 
if $\sigma>10$ (assuming $n=11$), under the corridor requirements. 

\noindent Note that success corresponds to
$\sigma\in\{0, 2\}$. By skipping these two values in my test as I was looking for failures only, and stopping at $\sigma=10$, 
it reduced compute time by a factor at least $10$.  $\square$
\vspace{2ex}\\
\noindent  Clearly, the main cause of failure is unexpected spikes in the sequence, especially early on. 
Long sub-sequences with identical values is another major cause of failure, but they are ruled out as non valid. 
Large spikes are eliminated thanks to the corridor requirements. For the prime numbers, spikes are
called \textcolor{index}{prime gaps} and are benign compared to other sequences in the corridor. 
The first occurrence of successive prime gap records are listed in OEIS \href{https://oeis.org/A002386}{A002386}. The primes concerned are shown in the first list below.
\vspace{1ex}
\begin{adjustwidth}{0.5cm}{0.5cm}
2, 3, 7, 23, 89, 113, 523, 887, 1129, 1327, 9551, 15683, 19609, 31397, 155921, 360653, 370261, 492113, 1349533, 1357201, 2010733, 4652353, 17051707, 20831323, 47326693, 122164747, 189695659, 191912783, 387096133, 436273009, 1294268491, $\dots$
\end{adjustwidth}
\vspace{1ex}
\begin{adjustwidth}{0.5cm}{0.5cm}
1, 2, 4, 6, 8, 14, 18, 20, 22, 34, 36, 44, 52, 72, 86, 96, 112, 114, 118, 132, 148, 154, 180, 210, 220, 222, 234, 248, 250, 282, 288, 292, 320, 336, 354, 382, 384, 394, 456, 464, 468, 474, 486, 490, 500, 514, 516, 532, 534, 540, 582, 588, 602, 652,$\dots$
\end{adjustwidth}
\vspace{1ex}
They are well spaced out, and the record values increase slowly, see 
OEIS \href{https://oeis.org/A005250}{A005250} with extract just above.
These two facts combined make the gaps too weak to cause
anything %other than 
but local disturbances quickly absorbed in the triangle, in a reasonable
number of levels well before we reach the bottom. Of cause, large gaps followed by tiny prime differences (such as twin primes) 
take more time to dissipate, and are worth investigating. For instance, between 20,831,323 and 20,831,533, there is a massive gap of 210 composite numbers. Yet, directly following 20,831,533, the next prime is 20,831,537, creating an incredibly small gap of just 4.
More on this~topic on the Prime Gap Project website, \href{https://primegap-list-project.github.io/}{here}. 
The Polymath Project proved 
that there are infinitely many prime~gaps smaller than 246.

However, a real concern is the presence of very long runs (called \textcolor{index}{plateaus}) in the first order differences (the prime gaps), followed by a steep cliff (a spike), followed by more plateaus and cliffs with the worst possible distribution of elevations. Especially
when happening early on. The only failing sequence with $n=11$ in the corridor  (table~\ref{table:gilb3}) exhibits such a pattern.
Yet, prime numbers are notorious for lacking such patterns due to
congruential demands. For instance, there is only one run of the form
$p, p+2, p+4$ where all three numbers
 are prime, because one of them must be a multiple of 3. The only
possibility is $3, 5, 7$ because $3$ is the only prime divisible by 3. 
The topic is related to primes in
arithmetic progressions and discussed in \cite{trb76d,napel9}. 
Prime runs are listed in OEIS \href{https://oeis.org/A333254}{A067090}. I summarize all these observations in the next statement. 
\vspace{1ex}
\begin{statement}\label{thxch656z} 
If the prime number sequence  fails (that is, if Gilbreath's conjecture is not true), it would not be caused by unusually large gaps in prime numbers, nor by the asymptotic distribution of primes, nor by the extreme rarity of long runs in prime gaps (the maximum length of arithmetic progressions consisting of primes only), nor by how early
 such runs occur. 
\end{statement}

\subsection{Hidden patterns causing failures}

I now discuss a very unusual sequence  in the corridor: the only one that fails when $n=11$. I haven't tested larger values of $n$ due to
computational complexity. Despite reassurances provided by statement~\ref{thxch656z}, there is~a possibility that some unknown insanely rare pattern 
could manifest after a colossal number of terms in the prime number sequence.

The unfortunate, exotic, and extremely rare sequence in table~\ref{table:gilb3} is
of critical importance. It is the only valid one of length $n=11$ out of infinitely many, that fails
while meeting the requirements to be in the efficient corridor. 
One may argue that the first differences have a long run incompatible with primes,
 and a gap larger than any in the first 11 primes. 
%Discarding this sequence based
%on that (by adding no long runs early on in the corridor requirements) is like
%burying your head in the sand. 
However the \textcolor{index}{failure point} happens at 29, not early on. 
Replacing 29 by either 27 or 31 leads to success, despite
the long run and large gap still there. So, something else
is causing the failure. But what?
Think about it before reading on. 

In his paper on Gilbreath's conjecture~\cite{zachch23}, Chase uses a narrow corridor with slow growth. The growth---too slow to encompass the prime number sequence---guarantees that large spikes take a very long time to show up, increasing the proportion of succeeding
sequences in his corridor. He then proves that his set of failing sequences has zero density in his corridor.
By contrast, my corridor grows at the speed of primes. Yet, large gaps are not the issue. 
After all, the valid sequence 2, 3, 5, 9, 17, \dots with $q_{k+1} = 2q_{k}-1$ succeeds, but is not in my corridor. 
As a side note, this example is the most extreme of all valid successful sequences, in terms of growth. 

To dig further into the problem, look at table~\ref{table:gilb3}. What creates the failure is not just the number 29,
 even though prior to reaching it, the sequence was successful despite a major issue. 
The large gap of 8 between 17 and 25, combined with a run of identical tiny gaps right after, is an issue. 
You can fix it by replacing 29 either with 27 or 31. 
But~there is more to it. What compounds the problem is  the segment 11, 13, 15, 17, resulting in a long run of 
2 (smallest possible gap) followed by the large gap of 8. The first issue causes 6 (a high value this early) to
trickle down a few levels, more than expected. The second issue allows 6 to propagate down  
much further, not leaving enough levels to fully recover from it, thus crashing at the bottom.

Actually, even the segment 3, 5 contributes to the demise, introducing a 0 near the bottom and next to~1. This, combined with a 4 
next to it, itself arising from the above issues, kills the deal. It is not possible to change 3, 5
to get rid of that problem, without introducing a new fatal issue, unless you also change some other values. 
In the end, not only large values in the triangle are a liability, but also the 0's. 
Especially when they pop up too close to the left at a level too close to the bottom, with many consecutive 0's followed by a large value, followed by more 0's, followed by another large but different value in-between 0 and the previous high, and so on. 
\vspace{1ex}

\begin{comment}
\begin{table}[H]
\centering
\renewcommand{\arraystretch}{1.0}
\small
\setlength{\tabcolsep}{1.5pt}
%\begin{tabular}{ccccccccccccccccccccccc} 
\begin{tabular}{*{23}{>{\centering\arraybackslash}p{0.25cm}}}
 %\hline
 % &  &  &        \\  [-2.5ex]
 % 1--grams & 2--grams & 3--grams & 4--grams \\%[0.5ex] 
% \hline 
\hline 
   &  &   &        \\  [-2.2ex]
%\hline
 & 2 & & 3 & & 5 & & 9 & & 11 & & 13 & & 15 & & \textcolor{red}{17} & & \textcolor{red}{25} & & 27 & & 29\\
\hline
  &  &  &  & & & & & & & & & &  & & & & & &   \\  [-2ex]
%10 & & 1&  & 2 & \\
8 & &  1 & &  2 & &  4 & &  2 & &  2 & &  2 & &  2 & &  \textcolor{red}{8} & & 2 &  & 2\\   
6 & & & 1 & & 2 & & 2 & & 0 & & 0 & & 0 & & \textcolor{red}{6} & & 6 & & 0 \\  
6 & & & & 1 & & 0 & & 2 & & 0 & & 0 & & \textcolor{red}{6} & & 0 & &6 \\
6 & & & & & 1 & & 2 & & 2 & & 0 & & \textcolor{red}{6} & & 6 & & 6\\
6 & & & & & & 1 & & 0 & & 2 & & \textcolor{red}{6} & & 0 & & 0\\
6 & & & & & & & 1 & & 2 & & 4 & & \textcolor{red}{6} & & 0\\
6 & & & & & & & & 1 & & 2 & & 2 & & \textcolor{red}{6}\\
4 & & & & & & & & & 1 & & 0 & & \textcolor{red}{4}\\
4 & & & & & & & & & & 1 & & \textcolor{red}{4}\\
3 & & & & & & & & & & & \circled[black2]{\textcolor{red}{3}}\\  % orange2
 %\hline
\end{tabular}
%\caption{Backend table {\fontfamily{ptm}\selectfont hash stem} (extract)}
\captionsetup{justification=centering}
\caption{Unusually rare instance of\\failure within the efficient corridor}
\label{table:gilb3}
\end{table}
\end{comment}

\begin{table}[H]
\centering
\renewcommand{\arraystretch}{1.0}
\small
\setlength{\tabcolsep}{0.8pt}
\parbox{.45\linewidth}{
\centering

\begin{tabular}{*{23}{>{\centering\arraybackslash}p{0.25cm}}}
 %\hline
 % &  &  &        \\  [-2.5ex]
 % 1--grams & 2--grams & 3--grams & 4--grams \\%[0.5ex] 
% \hline 
\hline 
   &  &   &        \\  [-2.2ex]
%\hline
 & 2 & & 3 & & 5 & & 9 & & 11 & & 13 & & 15 & & \textcolor{red}{17} & & \textcolor{red}{25} & & 27 & & 29\\
\hline
  &  &  &  & & & & & & & & & &  & & & & & &   \\  [-2ex]
%10 & & 1&  & 2 & \\
8 & &  1 & &  2 & &  4 & &  2 & &  2 & &  2 & &  2 & &  \textcolor{red}{8} & & 2 &  & 2\\   
6 & & & 1 & & 2 & & 2 & & 0 & & 0 & & 0 & & \textcolor{red}{6} & & 6 & & 0 \\  
6 & & & & 1 & & 0 & & 2 & & 0 & & 0 & & \textcolor{red}{6} & & 0 & &6 \\
6 & & & & & 1 & & 2 & & 2 & & 0 & & \textcolor{red}{6} & & 6 & & 6\\
6 & & & & & & 1 & & 0 & & 2 & & \textcolor{red}{6} & & 0 & & 0\\
6 & & & & & & & 1 & & 2 & & 4 & & \textcolor{red}{6} & & 0\\
6 & & & & & & & & 1 & & 2 & & 2 & & \textcolor{red}{6}\\
4 & & & & & & & & & 1 & & 0 & & \textcolor{red}{4}\\
4 & & & & & & & & & & 1 & & \textcolor{red}{4}\\
3 & & & & & & & & & & & \circled[black2]{\textcolor{red}{3}}\\  % orange2
 %\hline
\end{tabular}
%\caption{Backend table {\fontfamily{ptm}\selectfont hash stem} (extract)}
\captionsetup{justification=centering}
\caption{Unusually rare instance of\\failure within the efficient corridor}
\label{table:gilb3}
}
%\hfill
\quad\quad
\parbox{.45\linewidth}{
\centering

\begin{tabular}{*{23}{>{\centering\arraybackslash}p{0.25cm}}}
 %\hline
 % &  &  &        \\  [-2.5ex]
 % 1--grams & 2--grams & 3--grams & 4--grams \\%[0.5ex] 
% \hline 
\hline 
   &  &   &        \\  [-2.2ex]
%\hline
 & 2 & & 3 & & 5 & & 7 & & 9 & & \textcolor{red}{15} & & \textcolor{red}{41} & & 43 & & 47 & & 43 & & 27\\
\hline
  &  &  &  & & & & & & & & & &  & & & & & &   \\  [-2ex]
%10 & & 1&  & 2 & \\
26 & &  1 & &  2 & &  2 & &  2 & &  6 & &  \textcolor{red}{26} & &  2 & &  4 & & 4 &  & 16\\  %
24 & & & 1 & & 0 & & 0 & & 4 & & \textcolor{red}{20} & & 24 & & 2 & & 0 & & 12 \\ %
22 & & & & 1 & & 0 & & 4 & & \textcolor{red}{16} & & 4 & & 22 & & 2 & &12 \\%
20 & & & & & 1 & & 4 & & \textcolor{red}{12} & & 12 & & 18 & & 20 & & 10\\ %
10 & & & & & & \circled[black2]{\textcolor{red}{3}} & & \textcolor{red}{8} & & 0 & & 6 & & 2 & & 10\\ %
8 & & & & & & & \circled[black2]{\textcolor{red}{5}} & & \textcolor{red}{8} & & 6 & & 4 & & 8\\ %
4 & & & & & & & & \circled[black2]{\textcolor{red}{3}} & & 2 & & 2 & & 4\\
2 & & & & & & & & & 1 & & 0 & & 2\\ 
2 & & & & & & & & & & 1 & & 2\\
1 & & & & & & & & & & & \circled[green2]{1}\\
 %\hline
\end{tabular}
%\caption{Backend table {\fontfamily{ptm}\selectfont hash stem} (extract)}
\captionsetup{justification=centering}
%\caption{Prime numbers, one of many\\sequences verifying the conjecture}
\caption{4 successes, then 3 failures then 3\\successes (sequence outside the corridor)}
\label{table:rabhyg43x}

}
\end{table}
%\end{comment} %-------------------------------------------------------

large integer. Each digit is the equivalent of a row in the triangle. Usually, the
impact is local. But in rare instances, it propagates back to the first digit, changing its value. 
In the same way that adding 29 to the right to the sequence  2, 3, 5, 9, 11, 13, 15, 17, 25, 27
changes the bottom value in the triangle from 1 to 3, while adding 27 or 31 preserves the 1.
I conclude with the following statement. 
\vspace{1ex}
\begin{statement}\label{thxuuu6vb} 
Removing from the corridor all the sequences with the type of plateau and cliff combinations just described, would reduce
the already low density of failing sequences by an additional order of magnitude. If these
\noindent A good analogy is when adding a small integer to a large number. The carryover operations modify the digits in the 
combinations are the only non-trivial issue, then all sequences in the upgraded corridor would succeed, including the prime numbers
which would still meet the additional requirement to stay in the new corridor. 
\end{statement}

\begin{table}[ht]
\centering
\renewcommand{\arraystretch}{1.1}
\small
\begin{tabular}{l cccc cccc} % Total 7 underlying columns
\toprule[1pt] 
% First Row: Broadest category spanning all data columns
%Method & \multicolumn{6}{c}{Experimental Evaluation} \\ 
%\cmidrule(r){2-7}

% Second Row: Multiple distinct groups of multiple columns
 & \multicolumn{3}{c}{Successful Sequences} & \multicolumn{3}{c}{Failed Sequences} \\
\cmidrule(lr){2-4} \cmidrule(l){5-7}

% Third Row: Individual column headers
$n$ & Admissible & Valid & Corridor & Admissible & Valid & Corridor \\
\midrule

% Data Rows
6  & 82 & 27 & 4 & $\infty$ & $\infty$ & 0 \\
7  & 573 &  180  & 7  &  $\infty$ & $\infty$ & 0 \\
8  & 5,839 & 1,786 & 18  & $\infty$ & $\infty$ & 0 \\
9  & 86,921  & 26,094 & 46  & $\infty$ & $\infty$ &  1\\
10  & 1,890,317 &  559,127 & 121 & $\infty$ & $\infty$ &  0\\
11  & 60,013,894  & 17,535,396 & 345 & $\infty$ & $\infty$ &  1\\
\bottomrule[1pt]
\end{tabular}
\caption{Sequence count depending on type and length $n$}
\label{tab:multigrpo8}
\end{table}

\noindent Table~\ref{tab:multigrpo8} shows how rare failing sequences are in the corridor, at least up to length $n=11$. 
Empirical evidence suggests that failure rate decreases when $n$ increases, as the failure point usually occurs early in failing sequences 
in the corridor. In short, while the number of failing sequences may increase with $n$, even exponentially fast,
the number of successful ones increases at a much faster pace, by orders of magnitude. I summarize this
intuition in the following conjecture. 
\vspace{1ex}
\begin{conjecture}\label{thxuuu6vb} 
Let $g(n)$ and $h(n)$ be the number of sequences of length $n$ respectively succeeding and failing~in our corridor. 
Then both $g(n)$ and $h(n)$ grow indefinitely, with $h(n)/g(n)\rightarrow 0$. 
\end{conjecture}
\vspace{1ex}
The above claim depends heavily on the corridor and its threshold parameters. A different corridor
configuration or parameter set may lead to very different conclusions. In some cases, $h(n)$ and even $g(n)$ may stay finite at all times.
I found a similar behavior in a different corridor with growth $n^{3/2}$ instead of $n\log n$, but again
the choice of corridor parameters had a big impact. Of course, the lower the proportion of failures, the better. 
With one caveat: if failures are extremely rare, it may be difficult to figure out what could cause them, besides trivial patterns.

\section{Synthetic sequences with Poisson gaps to mimic prime numbers}\label{synthetia}

I discuss the most interesting insights first, before detailing my expanded search for failures and the implications for the prime number sequence. From the mathematician's vantage point, it looks like a war is taking place~in the corridor, with prime numbers using every
possible trick to win it. Not with deep theoretical machinery, but rather, simple, deterministic patterns that make the primes eminently non random. In particular:
\vspace{1ex} 
\begin{itemize}  
\item Besides the obvious, the feature most likely to kill a sequence is the pattern 2 4 2 4 2 4 and so on, 
preceded or followed by a moderately large value, in the first order differences. Or buried deeper in lower levels in the triangle. It is almost a requirement. 
\item In prime numbers, the first order differences are known as prime gaps, controlled by congruential rules. 
Infinitely many patterns are banned or can happen only once at the very beginning. The best example is arguably 2 4 2 4 2 4 and so on. 
So the prime number sequence will never face the issue just described. 
\item If the pattern 2 4 2 4 2 and so on appears at a lower level in the triangle, surrounded by a big enough number on either side to guarantee failure at the bottom level, the end result upstream (at the sequence level) is a much bigger number or too fast of a growth that
kicks the sequence outside the corridor. Thus, it cannot be found either at lower levels in prime numbers, whose sequence 
is inside the corridor. 
\end{itemize}   
\vspace{1ex}

\noindent What determines failure or not is the length of the pattern in question, how close to the left it starts, at which level, and 
 how big the two surrounding numbers are (left and right). Of course, the sequence may still succeed if the values
to the left are such that they can dilute the generated impulse fast enough before it propagates to the bottom.   
A good example is the sequence consisting of the primes $2, 3,\dots, 47$ followed by 49,  53,  55,  65. 
 It fails not because the gap between 55 and 65 is too large. Replace (55, 65) by (53, 67) and it succeeds despite increasing
 that gap and creating a duplicate of 53.
Yet the gap 55--65, while modest, is large enough that when combined with the preceding numbers in
the original sequence, it leads to failure. The suggested replacement breaks a combination of bad cycles. Barely, but
 just enough to move from failure to success. The good looking sequence fails, its ugly sister succeeds. See table~\ref{table:gilb4}.

\begin{table}[H]
\centering
\renewcommand{\arraystretch}{1.0}
\small
%\footnotesize
%\scriptsize
\setlength{\tabcolsep}{2pt}
%\begin{tabular}{ccccccccccccccccccccccc cccccccc ccccccc} 
\begin{tabular}{*{38}{>{\centering\arraybackslash}p{0.25cm}}}
 %\hline
 % &  &  &        \\  [-2.5ex]
 % 1--grams & 2--grams & 3--grams & 4--grams \\%[0.5ex] 
% \hline 
\hline 
   &  &   &        \\  [-2.2ex]
%\hline
% & 2 & & 3 & & 5 & & 9 & & 11 & & 13 & & 15 & & 17 & & 25 & & 27 & & 29\\
& 2 && 3 && 5 && 7 && 11&& 13&& 17&& 19&& 23&& 29&& 31&& 37&& 41&& 43&& 47&& 49&& 53&& 55&& 65\\
\hline
  &  &  &  & & & & & & & & & &  & & & & & &   \\  [-2ex]
%10 & & 1&  & 2 & \\
10 & &  1 & &  2 & &  2 & &  4 & &  2 & &  4 & &  2 & &  4 & & 6 &  & 2 & &  6  && 4  && 2  && 4  && 2  && 4  && 2 && 10\\  %
8 & & & 1 & & 0 & & 2 & & 2 & & 2 & & 2 & & 2 & & 2 & & 4 && 4 && 2 && 2 && 2 && 2 && 2 && 2 && 8\\ 
6 & & & & 1 & & 2 & & 0 & & 0 & & 0 & & 0 & & 0 & &2 && 0 && 2 && 0 && 0 && 0 && 0 && 0 && 6\\%
6 & & & & & 1 & & 2 & & 0 & & 0 & & 0 & & 0 & & 2 && 2 && 2 && 2 && 0 && 0 && 0 && 0 && 6\\
6 & & & & & & 1 & & 2 & & 0 & & 0 & & 0 & & 2 && 0 && 0 && 0 && 2 && 0 && 0 && 0 && 6\\
6 & & & & & & & 1 & & 2 & & 0 & & 0 & & 2 && 2 && 0 && 0 && 2 && 2 && 0 && 0 && 6\\
6 & & & & & & & & 1 & & 2 & & 0 & & 2 &&  0 && 2 && 0 && 2 && 0 && 2 && 0 && 6\\
6 & & & & & & & & & 1 & & 2 & & 2 & & 2 & & 2 & & 2 & & 2 & & 2 & & 2 & & 2 & & 4\\ %
4 & & & & & & & & & & 1 & & 0 & & 0 & & 0 & & 0 & & 0 & & 0 & & 0 & & 0 & & 4\\
4 & & & & & & & & & & & 1 & & 0 & & 0 & & 0 & & 0 & & 0 & & 0 & & 0 & & 4\\
4 & & & & & & & & & & & & 1 & & 0 & & 0 & & 0 & & 0 & & 0 & & 0 & &4\\
4 & & & & & & & & & & & & & 1 & & 0 & & 0 & & 0 & & 0 & & 0 && 4\\
4 & & & & & & & & & & & & & & 1 & & 0 & & 0 & & 0 & & 0 & & 4\\
4 & & & & & & & & & & & & & & & 1 & & 0 & & 0 & & 0 & & 4\\
4 & & & & & & & & & & & & & & & & 1 & & 0 & & 0 & & 4\\
4 & & & & & & & & & & & & & & & & &1 & & 0 & & 4\\
4 & & & & & & & & & & & & & & & & & &1 & & 4\\
3 & & & & & & & & & & & & && & & & & & \circled[black2]{\textcolor{red}{3}}\\
 %\hline
\end{tabular}
%\caption{Backend table {\fontfamily{ptm}\selectfont hash stem} (extract)}
\caption{Rare failing sequence; replace 53, 55, 65 by 53, 53, 67 to make it succeed!}
\label{table:gilb4}
\end{table}

\noindent To summarize, the new failure causes discovered are: (1) the flip side of a basin with one-time spike at the~border,
 that is, a high elevation plateau surrounded by a one-time abyss on either side, (2) the pattern 2 4 2 4 2 4 with a rather large value on either end, and variants not found in prime gaps, and (3) the same at a lower level.  
The last case results in larger values at the top level, kicking the sequence out of the corridor. 
If this can be precisely quantified and proved, it would solve Gilbreath's conjecture. The work needed does not seem
insurmountable. Prime numbers are actually a very good candidate for success due to
 congruential constraints on prime gaps, unlike more random sequences.  But the first step is to check whether or not I covered all failure types.  

The concepts of cliff, plateau and abyss also apply to sequences growing faster than the primes, with similar conclusions. For instance,
the sequence $2, 3, 5, 9, 15, \,23,\, 33, \,45, \dots$ with $q_1=2$ and $q_n = n^2-3n+5$ for $n>1$ 
obviously succeeds no matter how long it is. After $q_8=45$, one expect $q_9 = 59$. Without even computing~the triangle, you can tell that adding $q_9\in\{59,61\}$ succeeds while $q_9\geq 63$ fails. It gives you the
minimum height for the killing cliff.  
With little efforts, you find that adding $q_9\in\{57, 55, 53 \}$ succeeds while adding $q_9\leq51$~fails.
And this determines the minimum depth of the deadly abyss.
In this sequence, the values between $q_2=3$ and $q_8=45$ play the role of the plateau, by constituting a very smooth albeit fast growing segment, with 2nd order differences all zero as the degree of the polynomial is 2. Again, 
we are talking about a very narrow cliff of abyss, just one \textcolor{index}{outlier value}. Otherwise, future values 
at the sequence level could restore success, as explained earlier.

\subsection{High performance computing to discover rare failure patterns}\label{gsimuls}

So far, researchers focused on large prime gaps as the main danger, using the best bounds available
 to shrink~the corridor by imposing conditions stricter than~(\ref{eureke}). This is the wrong approach:
the non trivial causes for failures are different, not visible to the naked eye, and hard to find in simulations within a narrow corridor. Actually, more chaotic sequences with bigger gaps have
 a high success rate, see figure~\ref{fig:aajj5gb17} also based on Poisson increments.  

My first attempt to build long sequences to discover new patterns of failure in the corridor,
 produced mixed results. The first failure always happened very early on in the simulations, and always due to a 
moderately large gap preceded by a small one -- which also happens in succeeding sequences including the primes. 
I did not discover other patterns because I run too few simulations: I generated $10^6$ random sequences
in the corridor. Other patterns exist, see section~\ref{pordel}. You might need more than $10^{50}$ simulations to find one,
or use a different approach. The following simulation problem is a good analogy: you generate bit streams of length $n$. Success is defined as a stream with above $10\%$ of `1'. Failures are easy to find for small $n$, but with
 $n=100$, the chance of failure for a random bit stream is $1.107 \times 10^{-17}$. For $n =1000$, it goes
 down to $7.428 \times 10^{-163}$. Yet, the number of failures over all bit streams increases with $n$ and becomes infinite,
 while the proportion of failing streams tends to zero very fast.

In my $10^6$ sequences, each with 1000 terms, 2883 failed. The sequences were all distinct, but there is~little leeway for the first 20 terms or so, especially since all sequences start with 2, 3, 5, 7, 11, \dots, 31, 37, 41, 43. 
As a result, I found many duplicates among the short paths that lead to first failure. 
After duplicates removal, the number 2883 is reduced down to 81 distinct failure cases.
They are listed in table~\ref{tablt5rib5} where the columns are as follows:
\vspace{1ex}
\begin{itemize}
\item $n$ is the index of the term where first failure occurs. Since the first 14 term are the same in each sequence, for instance $n=17$ means that failure
occurred just after adding 3 new terms. 
\item  $\sigma$ has the same meaning as earlier; $\sigma - 1$ is the value at the bottom of the triangle. It is the first one $\neq 1$ in the left diagonal.
In case of success, $\sigma\in\{0, 2\}$. Otherwise $\sigma$ is an even integer $\geq$ 4.  
\item $\gamma$ is the maximum value in the first order differences when stopping at the first failure point.  
\item $\delta_n = q_n - q_{n-1}$ where $q_n$ is the value
at first failure point in the sequence. With few exceptions,
the large $|\delta_n-\delta_{n-1}|$ is causing the failure. 
Usually, the maximum $\gamma$ occurs at the right end, that is, $\gamma=\delta_{n}$. 
\item The ``forbidden" column is the most striking feature, see section~\ref{forbidden}.
Note the entry 2, 4, 2, 4, 2~mentioned previously and linked to potential early failure. 
\end{itemize}
\vspace{1ex}
\begin{statement} \label{chius}
The vast majority of failing corridor sequences contain a \textcolor{index}{forbidden prime constellation} pattern early on. 
\hspace{-1.5ex} Primes do not have such patterns. \hspace{-1.5ex} It also explains why the prime number sequence
does not fail early on, where the chance of failure is highest. 
\end{statement}
\vspace{1ex}
\noindent Statement~\ref{chius} is based on table~\ref{tablt5rib5}. I now describe my simulations. All the failures occurring
late in a  sequence were preceded by a failure early on within the same sequence, thus defeating the purpose. But they did reveal 
many banned patterns in the first order differences: the forbidden constellations. The algorithm starts
with~the sequence consisting of the first 14 primes. At iteration $n$, the new added term 
is $q_{n+1}= q_n + 2 + 2 u_{n+1}$, where the $u_n$ are i.i.d deviates from a Poisson($\lambda$) distribution. 
If the augmented sequence with $q_{n+1}$ is no longer in the corridor, I repeat this step until I find a new $u_{n+1}$ candidate that keeps us in the corridor.  
%----

Thus, I use \textcolor{index}{rejection sampling} to mimic the gaps $q_{n+1}-q_n$. For a better fit, use the prime gap \textcolor{index}{empirical distribution}\index{empirical distribution}, a \textcolor{index}{generalized Poisson}~\cite{gpdch26}, 
or an adjusted Poisson  that depends on $n$, with the \textcolor{index}{Hardy-Littlewood $k$-tuples correction}\index{Hardy-Littlewood $k$-tuples correction}. 
Tao and others~\cite{gbctao26} use a geometric distribution. More on this topic in~\cite{lpgtao23,ajpc26, dres24}. These models assume independence between successive gaps. This assumption is wrong due to the forbidden patterns. %### no independence  ... tao's paper
The small $\lambda = 2.5$ produces few rejections in rejection sampling (speeding up the
algorithm) and generates the largest number of distinct sequence failures, yet $<0.01\%$ out of $10^6$ generated sequences
 with $n=1000$ terms. Almost all failures were caused by a too large gap early on, after the first 14 terms.

%-------

As of today, the best asymptotic bound for the maximum gap between two primes, proved unconditionally (without assuming that the Riemann Hypothesis is true), is  
\begin{equation}
p_{n+1} - p_n \lesssim  O\big(p_n^{525}\big),\label{carbo}  %%%%%%% notation for "asymptotically <=
\end{equation}
established in 2001 
by Baker~\cite{pq111}, see also~\cite{pq222}. 
In our context, a weaker bound suffices. According to~conjecture~\ref{conjhk}, we only need $p_{n+1} - p_n \lesssim n$. Not meeting this bound
would guarantee failure of the  prime number~sequence. Of course, the sharper bound you use, the easier to prove Gilbreath's conjecture. 
Assuming RH, $O\big(p_n^{0.525}\big)$ can be reduced to $O\big(\sqrt{p_n} \log p_n\big)$, a small
 improvement.~Cramér conjectured in 1936 that it can be lowered to $O(\log^2 p_n)$ but his heuristic has been questioned.
In 2025, Wang~\cite{pq333} claimed to have proved Cramér's conjecture. 
As you may know, $p_n\sim n\log n$, thus the {\em average} prime gap is asymptotically $\log n$.

\subsection{Prime gaps and forbidden prime constellations}\label{forbidden}

The ``forbidden" column in table~\ref{tablt5rib5} lists the first instance of a \textcolor{index}{forbidden prime constellation} 
found beyond the first 14 elements in all but one of the 81 failed sequences from the previous experiment
(see section~\ref{gsimuls}). But what is exactly?

Take the quintuplet 2, 4, 2, 4, 2 in the first order differences. If present in the primes, it means that 
there is at least one prime $p$ for which $p, p+2, p+6, p+8, p+12, p+14$ are all primes.
But no matter $p$, at least one of these 6 numbers is always a multiple of 5. 
So this quintuplet can not be found in prime gaps, except possibly just once, at the
beginning. 
Actually 5, 7, 11, 13, 17, 19 works, but this is the only exception. Thus,~it~is~called~a forbidden pattern.  
This quintuplet is found in failed sequences, but not in successful ones, in my small sample. 

Likewise, the forbidden patterns 4, 4 and 8, 8 happens frequently in successful sequences, but not in the failed ones (again, small sample that only indicates a probability of occurrence). Of course all banned patterns occur in successful sequences, but their distribution is much different from that in failed ones, especially for longer patterns. To find them in a given sequence, I looked at successive pairs in the first order differences after the first 14 terms and tested each pair for admissibility in primes (that is, satisfying or not the congruential constraints). I then did the same with triplets, quadruplets, and so on. If a pair is banned, all triplets containing it are also banned. In particular the most fundamental pair 2, 2 is banned. Thus, so is 4, 2, 2 or 4, 2, 2, 6, 2 and so on. Same with triplets, quadruplets and all tuples.  It makes the prime numbers anything but random. At the same time, it gives them some protection against failure: a long chain of 2 is dangerous if followed or worse, preceded by a spike. Primes are protected against that and other dangerous patterns such as 2, 4, 2, 4, 2 (all banned). 

I did not find any reference about forbidden patterns, but there is a lot of ongoing research about admissible ones, sometimes
called permissible prime constellations. It started with the prime $k$-tuple conjecture and the  
Hardy–Littlewood asymptotic constants. See the Wolfram entry on the topic, \href{https://mathworld.wolfram.com/PrimeConstellation.html}{here}, 
and Richard Mathar's table of prime gap constellations, \href{https://oeis.org/A022004/a022004_2.pdf?b%27%27}{here}. 
For more more recent references on the topic, see~\cite{qq357, qq356}.
\vspace{1ex}

\begin{table}[H]
\centering
\renewcommand{\arraystretch}{1.0}
%\small
%\footnotesize
%\scriptsize
%\setlength{\tabcolsep}{3pt}
\begin{tabular}{|cccccc|cccccc|cccccc|} 
\hline 
   &  &   &   & &   & & & & & & & & & &  & &\\  [-2.2ex]
$n$	& $\sigma$	&$\gamma$	& $\delta_{n-1}$ & $\delta_n$ & Forbidden	 & $n$& $\sigma$	&$\gamma$	& $\delta_{n-1}$ & $\delta_n$ & Forbidden	 &$n$	&  $\sigma$	&$\gamma$	& $\delta_{n-1}$ & $\delta_n$ & Forbidden\\		
\hline
  &  & &  & &  & & & & & & & & & & &  & \\  [-2ex]
16	&	4	&	12	&	2	&	12	&	--	&	18	&	4	&	14	&	2	&	14	&	2, 2 	&	22	&	4	&	12	&	2	&	12	&	2, 6, 2 	\\
18	&	4	&	10	&	2	&	10	&	4, 6, 4 	&	20	&	6	&	14	&	2	&	14	&	2, 8 	&	22	&	4	&	12	&	2	&	12	&	2, 2 	\\
16	&	4	&	14	&	2	&	14	&	2, 2 	&	22	&	6	&	14	&	2	&	14	&	2, 8 	&	22	&	4	&	12	&	2	&	12	&	4, 10 	\\
22	&	4	&	12	&	2	&	12	&	2, 8 	&	20	&	4	&	12	&	2	&	12	&	2, 8 	&	22	&	4	&	12	&	2	&	12	&	4, 6, 10 	\\
18	&	6	&	12	&	2	&	12	&	4, 6, 4 	&	22	&	4	&	12	&	2	&	12	&	8, 6, 2 	&	18	&	6	&	16	&	2	&	16	&	2, 2 	\\
17	&	4	&	14	&	4	&	14	&	6, 4, 2, 4, 2, 4 	&	22	&	4	&	12	&	2	&	12	&	4, 10 	&	22	&	6	&	14	&	2	&	14	&	2, 8 	\\
18	&	4	&	10	&	2	&	10	&	2, 4, 2, 4, 2 	&	22	&	4	&	12	&	2	&	12	&	2, 8 	&	22	&	4	&	12	&	2	&	12	&	4, 6, 4 	\\
18	&	6	&	12	&	2	&	12	&	2, 4, 2, 4, 2 	&	22	&	4	&	12	&	2	&	12	&	4, 10 	&	32	&	4	&	12	&	2	&	12	&	10, 4 	\\
20	&	4	&	12	&	6	&	4	&	2, 2 	&	20	&	8	&	16	&	2	&	16	&	2, 8 	&	22	&	6	&	14	&	2	&	14	&	2, 4, 6, 6, 6 	\\
18	&	4	&	14	&	2	&	14	&	2, 2 	&	26	&	4	&	10	&	2	&	10	&	2, 8 	&	24	&	4	&	12	&	6	&	8	&	2, 8 	\\
18	&	8	&	14	&	2	&	14	&	2, 4, 2, 4, 2 	&	22	&	6	&	16	&	2	&	16	&	2, 2 	&	20	&	4	&	12	&	2	&	12	&	2, 2 	\\
26	&	6	&	12	&	2	&	12	&	2, 8 	&	22	&	4	&	12	&	2	&	12	&	4, 2, 6, 6 	&	22	&	4	&	12	&	2	&	12	&	4, 10 	\\
18	&	4	&	14	&	2	&	14	&	2, 6, 2 	&	22	&	6	&	14	&	2	&	14	&	2, 8 	&	20	&	4	&	8	&	2	&	8	&	2, 8 	\\
22	&	6	&	14	&	2	&	14	&	2, 2 	&	22	&	4	&	12	&	6	&	4	&	4, 10 	&	22	&	4	&	16	&	2	&	16	&	2, 2 	\\
18	&	10	&	16	&	2	&	16	&	2, 4, 2, 4, 2 	&	22	&	4	&	12	&	2	&	12	&	2, 8 	&	22	&	4	&	12	&	2	&	12	&	2, 2 	\\
25	&	4	&	14	&	4	&	14	&	2, 8 	&	22	&	4	&	12	&	2	&	12	&	2, 8 	&	24	&	4	&	12	&	2	&	12	&	2, 8 	\\
22	&	6	&	14	&	2	&	14	&	2, 8 	&	22	&	4	&	12	&	2	&	12	&	4, 2, 6, 6 	&	24	&	4	&	12	&	6	&	4	&	2, 8 	\\
20	&	4	&	16	&	2	&	16	&	2, 2 	&	22	&	4	&	12	&	2	&	12	&	2, 8 	&	22	&	4	&	18	&	2	&	18	&	2, 2 	\\
20	&	4	&	12	&	2	&	12	&	2, 8 	&	22	&	4	&	12	&	2	&	12	&	2, 8 	&	24	&	4	&	12	&	6	&	4	&	4, 10 	\\
26	&	6	&	14	&	2	&	14	&	2, 2 	&	22	&	6	&	14	&	2	&	14	&	4, 2, 6, 6 	&	24	&	4	&	12	&	6	&	4	&	2, 2 	\\
22	&	4	&	12	&	2	&	12	&	2, 8 	&	24	&	4	&	12	&	2	&	12	&	2, 8 	&	20	&	4	&	20	&	2	&	20	&	4, 4 	\\
22	&	4	&	12	&	2	&	12	&	4, 10 	&	22	&	4	&	12	&	2	&	12	&	2, 2 	&	22	&	6	&	14	&	2	&	14	&	4, 2, 6, 6 	\\
25	&	4	&	14	&	4	&	6	&	2, 8 	&	25	&	4	&	14	&	4	&	6	&	8, 6, 2 	&	22	&	4	&	12	&	2	&	12	&	8, 2 	\\
20	&	4	&	12	&	2	&	12	&	2, 2 	&	28	&	4	&	16	&	6	&	4	&	16, 2, 4, 2 	&	20	&	4	&	16	&	2	&	16	&	2, 8 	\\
20	&	4	&	12	&	2	&	12	&	4, 10 	&	24	&	4	&	12	&	6	&	8	&	4, 6, 4 	&	24	&	4	&	12	&	6	&	8	&	4, 10 	\\
22	&	4	&	14	&	2	&	8	&	2, 8 	&	21	&	4	&	16	&	4	&	6	&	2, 2 	&	22	&	6	&	16	&	2	&	2	&	2, 8 	\\
20	&	4	&	16	&	2	&	16	&	2, 8 	&	24	&	4	&	12	&	2	&	12	&	4, 6, 4 	&	21	&	4	&	18	&	4	&	18	&	2, 8 	\\
\hline
 %\hline
\end{tabular}
%\caption{Backend table {\fontfamily{ptm}\selectfont hash stem} (extract)}
\caption{List of failing corridor sequences, out of $10^6$ sampled ones; all but one contain a forbidden pattern}
\label{tablt5rib5}
\end{table}

\section{A deeper version of Gilbreath's conjecture with partial proof}

In this section, I iteratively compute the bottom element in the triangle attached to a sequence $q_1,\dots, q_n$ by working on the right diagonal only, updating
 it incrementally as $n$ increases.  The sequence succeeds if and only if the last element in the diagonal is equal to $1$. 
I also compute the maximum interval $I_n = [q_n^-, \, q_n^+]$ such that the augmented sequence $q_1,\dots,q_n, q_{n+1}$ continues to
succeed if and only if $q_{n+1}\in I_n$. Along the way, I build a theoretical framework to solve the conjecture, leading
 to a partial proof and deeper claims. The approach is 
%novel, and the 
innovative. 
%methodology is interesting in its own right. 
For instance, in section~\ref{reduxor}, I reduce valid sequences to a canonical form:
starting at some level in the triangle and going downwards, the original and reduced sequences have the exact same rows. They differ only
 in the upper levels. 
%that has the same rows as the original one, 
%### that contains only 3 and 5 values,  yet with the same triangle rows as the original one, from level 3 down to the bottom.  
It allows you to focus on a simpler, equivalent sequence, to check success status.

\subsection{Distance between a successful sequence and its closest failing sister}\label{pordel}

The simulations in section~\ref{gsimuls} did not produce sequences that fail for the first time only after a large number of terms.
This was my intended purpose. Despite generating $10^6$ sequences, only 81 failed, and all failed early on.
While it started an interesting discussion about banned prime gap constellations in section~\ref{forbidden} with more on this topic here, I have yet to
find first failure points occurring much later in a corridor sequence. In this section, this goal is I achieved, while setting the foundations with
detailed steps, some completed, to prove Gilbreath's conjecture for the prime number sequence in particular. 
I chose the primes not because the original conjecture focuses on them only, but because they are an ideal
candidate, and in some ways easier to handle compared to
other types of chaotic sequences.

%-----------------------------vince/riemann2and3.mp4
\begin{figure}[H]
\centering
\captionsetup{justification=centering}
%\includegraphics[width=1.0\textwidth]{case1a.png}
%\vspace{0.25ex}
\includegraphics[width=0.95\textwidth]{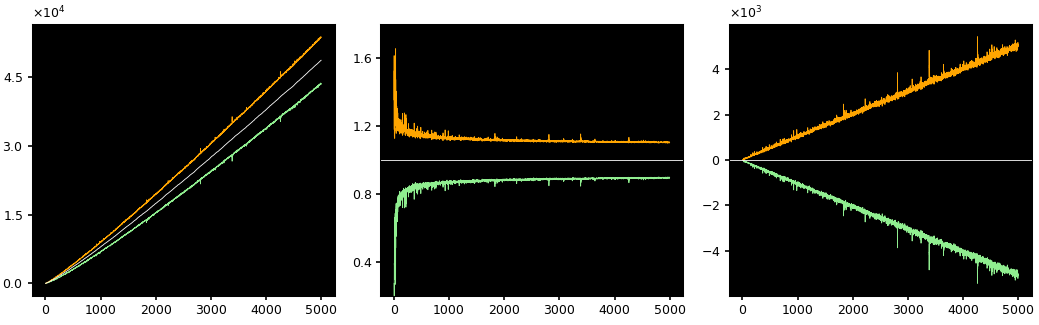}   % gb1sb.png / gbs.png
%\vspace{-1ex}
\caption{{\small Prime  sequence: $q^+_n,\,  q_n, \, q^-_n$ resp. in orange, white, green  (left). Ratio (center) and delta (right).}}
\label{fig:aajjjjh6y}
\end{figure}
%imgpy9979_2and3.PNG
%-------------------------

%-----------------------------vince/riemann2and3.mp4
\begin{figure}[H]
\centering
\captionsetup{justification=centering}
%\includegraphics[width=1.0\textwidth]{case1a.png}
%\vspace{0.25ex}
\includegraphics[width=0.95\textwidth]{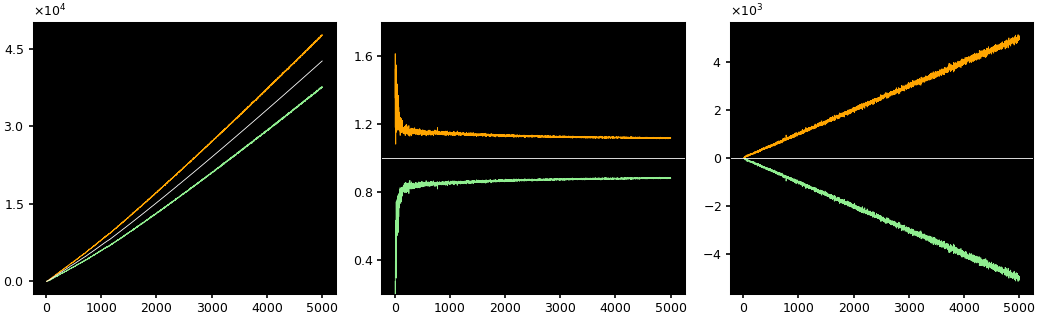}   % gb1sb.png / gbs.png
%\vspace{-1ex}
\caption{{\small  Random  sequence: $q^+_n,\,  q_n, \, q^-_n$ resp. in orange, white, green  (left). Ratio (center) and delta (right).}}  % >>>>>>>>> exponent is 0.35 
\label{fig:aajjjd3uyvo}
\end{figure}
%imgpy9979_2and3.PNG
%-------------------------

%-----------------------------vince/riemann2and3.mp4
\begin{figure}[H]
\centering
\captionsetup{justification=centering}
%\includegraphics[width=1.0\textwidth]{case1a.png}
%\vspace{0.25ex}
\includegraphics[width=0.95\textwidth]{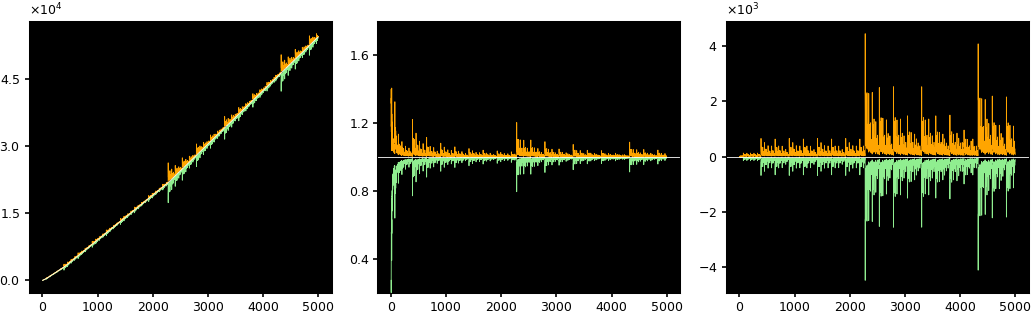}   % gb1sb.png / gbs.png
%\vspace{-1ex}
\caption{{\small  Log sequence: $q^+_n,\,  q_n, \, q^-_n$ resp. in orange, white, green (left). Ratios (center) and deltas (right).}}
\label{fig:aajjtt5xc}
\end{figure}
%imgpy9979_2and3.PNG
%-------------------------

\noindent The left plot in figures~\ref{fig:aajjjjh6y}--\ref{fig:aajjjd3uyvo} shows the increasing width of the \textcolor{index}{success interval} $I_n = [q_n^-, \, q_n^+]$ as $n$ increases, 
with $n$ on the X-axis. The sequence $q_1,q_2$ and so on, colored in white, has  $q_{n} \in I_{n-1}$. The orange and green curves
are (resp.) the admissible lower and upper bounds for $q_n$. The meaning is as follows: if the sequence is successful up to~$q_n$, 
then it will be successful at $q_{n+1}$ if and only if $q_{n+1}\in I_n$. The middle plot shows the ratios $q^+_{n-1}/q_n$ in orange, and
$q^-_{n-1}/q_n$ in green; the right plot shows the differences $q^+_{n-1}-q_n$ in orange, and
$q^-_{n-1}-q_n$ in green. Here, negative gaps are allowed.  I tested 4 sequences: 
\vspace{1ex}
\begin{itemize}
\item Figure~\ref{fig:aajjjjh6y}: The prime number sequence.
\item Figure~\ref{fig:aajjjd3uyvo}: Random sequence with Poisson increments of mean $\lambda=2.5$ (multiplied by 2 to make them even like prime gaps)
and rejection sampling to stay in the logarithmic corridor. See section~\ref{synthetia}.  
\item Figure~\ref{fig:aajjtt5xc}: Log model. Here $q_n = \lfloor \log q^+_{n-1}\rfloor$ rounded to the smallest odd integer $>q_{n-1}$. The brackets
$\lfloor \cdot \rfloor$ stand for the integer (floor) function.
\item Figure~\ref{fig:aajjjd3uyvffo}: Power model. Here $q_n = \lfloor (q^+_{n-1})^{0.35}\rfloor$ rounded to the smallest odd integer $>q_{n-1}$. 
\end{itemize}
\vspace{1ex}

\noindent The first two (prime numbers and the random sequence with similar growth) have the largest delta between~the lower and upper bounds. It means that at any given $q_{n}$, there is a lot of leeway to choose the next $q_{n+1}$ so that the augmented sequence continue to succeed iteratively, and thus indefinitely. Also, $q_n$ stays in the middle between the two bounds, never getting close to either one. 
It is impossible that these sequences could ever fail, if this behavior persists indefinitely.  
By contrast, the last two ones (the log and power sequences) progress in a narrow band, getting close to failing at regular intervals. 
Unless proved otherwise, success for all $q_n$ is not guaranteed. Some sequences could hit a wall with no possible escape from failure.  
Or for the prime number sequence, for some unknown insanely large $n$, $q_n$ could jump above the upper bound, maybe even just one single time in its entire history, 
returning within the band at $n+1$ and staying in it thereafter,  forever.

%-----------------------------vince/riemann2and3.mp4
\begin{figure}[H]
\centering
\captionsetup{justification=centering}
%\includegraphics[width=1.0\textwidth]{case1a.png}
%\vspace{0.25ex}
\includegraphics[width=0.95\textwidth]{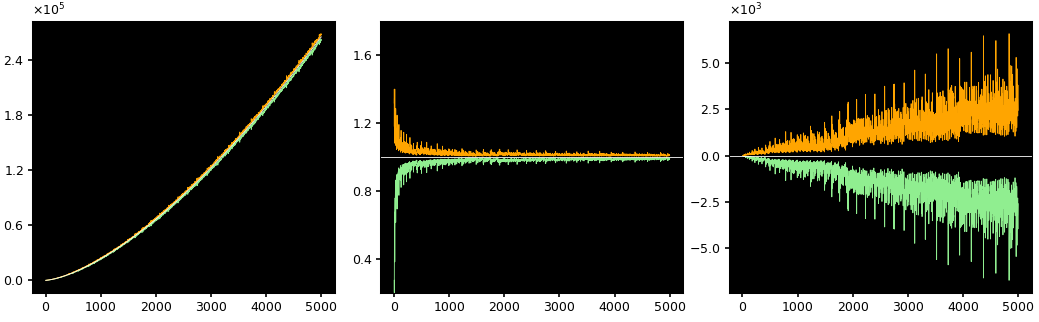}   % gb1sb.png / gbs.png
%\vspace{-1ex}
\caption{{\small  Power  sequence: $q^+_n,\,  q_n, \, q^-_n$ resp. in orange, white, green  (left). Ratio (center) and delta (right).}}  % >>>>>>>>> exponent is 0.35 
\label{fig:aajjjd3uyvffo}
\end{figure}
%imgpy9979_2and3.PNG
%-------------------------

\noindent By design, the random and the prime number sequences share a lot in common. However, there are major differences. 
The primes exhibit regular spikes in the upper and lower bounds.  
Also, for the primes, on occasions, the right diagonal consists only of 0 and 2, besides the first and last elements. These primes, always the largest in a pair
 of twin primes, are called \textcolor{index}{magic primes}: see theorem~\ref{easyn}. In the random sequence, this happens~much less frequently regardless of $\lambda$,
 and the pattern seems to stop at some point. One may argue that my simulations do not mimic the primes very well.  
The reality is the opposite: the primes lack randomness. It shows on multiple occasions including here,
making them more risky than most people think in cryptographic applications.

\subsection{General framework, proof for magic primes and other cases}

So far, I  discussed several situations that could lead to failure. For the prime numbers, the record gaps are~not the main issue; they are not large enough
 to be a concern. A combination of a moderately large gap surrounded by several tiny ones and/or some halfway between, especially if this pattern repeats itself in close
 succession,  creating plateaus, basins, dips and spikes in the first or higher order differences, causes problems not immediately but later on in the right diagonal.
I know formalize these findings, leading to theorem~\ref{ti4r} and conjecture~\ref{conjhk}. 
\vspace{1ex}

\begin{table}[H]
\centering
\renewcommand{\arraystretch}{1.23}
\small
\begin{tabular}{
>{\centering\arraybackslash}p{0.5cm}
>{\centering\arraybackslash}p{0.8cm}
>{\centering\arraybackslash}p{0.35cm}
>{\centering\arraybackslash}p{0.35cm}
>{\centering\arraybackslash}p{0.35cm}
>{\centering\arraybackslash}p{0.35cm}
>{\centering\arraybackslash}p{0.35cm}
>{\centering\arraybackslash}p{0.35cm}
>{\centering\arraybackslash}p{0.35cm}
>{\centering\arraybackslash}p{0.35cm}
>{\centering\arraybackslash}p{0.35cm}
>{\centering\arraybackslash}p{0.35cm}
>{\centering\arraybackslash}p{0.35cm}
>{\centering\arraybackslash}p{0.35cm}
}
%\toprule[1pt] 
%\hline
%\hline
%\midrule
%\hline
$n$	&	$q_n$	&	1	&	2	&	3	&	4	&	5	&	6	&	7	&	8	&	9	&	10	&	11	&	12	\\
\hline
%\toprule[1pt]\\[-0.4ex]
205	&	1277	&	\cellcolor{green2} \bf 18	&	8	&	6	&	2	&	0	&	0	&	2	&	0	&	2	&	0	&	0	&	2	\\
206	&	1279	&	\cellcolor{green2}  \bf \textcolor{red}{2}	&\bf	16	&	8	&	2	&	0	&	0	&	0	&	2	&	2	&	0	&	0	&	0	\\
207	&	1283	&	\cellcolor{green2} 4	&	2	&\bf	14	&	6	&	4	&	4	&	4	&	4	&	2	&	0	&	0	&	0	\\
208	&	1289	&	\cellcolor{green2}  6	&	2	&	0	&\bf 	14	&	8	&	4	&	0	&	4	&	0	&	2	&	2	&	2	\\
209	&	1291	&\cellcolor{green2} 	2	&	4	&	2	&	2	&	\rectangled[green2]{\bf \textcolor{black}{12}}	&	4	&	0	&	0	&	4	&	4	&	2	&	0	\\
210	&	1297	&\cellcolor{green2} 	6	&	4	&	0	&	2	&	0	&	\rectangled[green2]{\bf \textcolor{black}{12}}	&	8	&	8	&	8	&	4	&	0	&	2	\\
211	&	1301	&\cellcolor{green2} 	4	&	2	&	2	&	2	&	0	&	0	&	\rectangled[green2]{\bf \textcolor{black}{12}}	&	4	&	4	&	4	&	0	&	0	\\
212	&	1303	&\cellcolor{green2} 	2	&	2	&	0	&	2	&	0	&	0	&	0	&	\rectangled[green2]{\bf \textcolor{black}{12}}	&	8	
&	4	&	0	&	0	\\[0.1ex]
\hline
213	&	1309	&	\cellcolor{yellow} 6	&	4	&	2	&	2	&	0	&	0	&	0	&	0	&	\rectangled[yellow]{\bf \textcolor{black}{12}}	&	4	&	0	&	0	\\
214	&	1319	&	\cellcolor{yellow} 10	&	4	&	0	&	2	&	0	&	0	&	0	&	0	&	0	&	\rectangled[yellow]{\bf \textcolor{black}{12}}	&	8	&	8	\\
215	&	1327	&	\cellcolor{yellow} 8	&	2	&	2	&	2	&	0	&	0	&	0	&	0	&	0	&	0	&	\rectangled[yellow]{\bf \textcolor{black}{12}}	&	4	\\
%\hline
\end{tabular}
\caption{$\delta_k(q_n)$, with $1\leq k\leq 12$ and $205\leq n \leq 215$}
\label{table:rassel}
\end{table}

\begin{theorem}\label{easyn}
In the prime number sequence, there are twin primes $(p-2, \,p)$ called \textcolor{index}{magic twins}, presumably infinitely many of them,  for which the right diagonal  
$\delta(p)$
consists only of 0 and 2 except for the first and last elements. If the sequence succeeds up to $p-2$, adding $p$
guarantees success at $p$. % and also at the prime after $p$ requires balanced].
\end{theorem}
\vspace{1ex}
{\bf Proof}:  The fact that the sequence succeeds at $p$ if it succeeds at $p-2$ is trivial. Let $\delta_{-1}(p), \delta_{-2}(p)$
denote the last 2 elements in the right diagonal $\delta(p)$. Success means $\delta_{-1}(p)=1$. However if
we succeed at $p-2$, we must have $\delta_{-1}(p-2)=1$. Since $\delta_{-2}(p)\in\{0, 2\}$
and $\delta_{-1}(p) = |\delta_{-2}(p) - \delta_{-1}(p-2)|$, we must have $\delta_{-1}(p) = 1$.  $\square$
\vspace{2ex}

\noindent Theorem~\ref{easyn} features the easiest case where adding the next prime guarantees success in the augmented sequence if it was successful up to the previous one. And the regular resets to a right diagonal consisting only of 0 and 2 contributes
to making the full, infinite prime sequence successful.  The list of magic twin pairs $(p-2, p)$ is not in the OEIS encyclopedia. Below are the first few $p$, that I also call \textcolor{index}{magic primes}: 
\vspace{1ex}
\begin{quote}
13,
19,
43,
73,
103,
109,
283,
463,
619,
2689,
2713,
2803,
3469,
5659,
13693,
14563,
24109,
28663,
36013,
36793,
42409,
42703,
56533,
65719,
74203,
80683,
92383,
97003
\end{quote}
\vspace{1ex}
%when entering 0 2 cycle in right diag, success is guarantee [assuming success before]
%can be infinite, in which case you must change at least the two rightmost values, not just the last one

\noindent Theorem~\ref{easyn} leads to a generalization. First, let me define the \textcolor{index}{0-2 cycle} in a right diagonal in any valid, finite sequence.
It is the longest section at the bottom, on the right in the corresponding triangle, just above the last element (itself necessarily an odd integer), and consisting 
 of 0 or 2 only. It may or may not exist. The length of the 0-2 cycle and when it starts are what matter the most; it ranges from 1 element (a single 0 or 2) to the whole diagonal as in theorem~\ref{easyn}. If and only if not present, the sequence fails. The 0-2 cycle may contain only 0, only 2, or any combination of both. 
\vspace{2ex}
\begin{theorem}\label{easy2}
For a finite valid sequence to continue to succeed after adding a term $q$, the  right diagonal $\delta(q)$ in the triangle attached to the augmented sequence, must have
 a 0-2 cycle. You use this fact to recursively~augment~a sequence ad infinitum while keeping success at all times.  
\end{theorem}
\vspace{1ex}
{\bf Proof}: It uses the same argument as in the proof of theorem~\ref{easyn}. $\square$

\vspace{2ex}

\noindent 
I now discuss related concepts. In any sequence $q_1,\dots q_n$, the right diagonal $\delta(q_n)$ has $n$ elements by
construction. Let $\rho_n(k)$ be the proportion of $k$, for $k=0,2, 4$ and so on. The sequence is said to have a
\textcolor{index}{balanced right diagonal} if the following is satisfied:
\begin{equation}
\max\Big(\,\Big|\rho_n(0) -\frac{1}{2}\Big|, \, \Big|\rho_n(2) -\frac{1}{2}\Big| \,\Big) \leq 
\frac{\displaystyle{\sqrt{2n\log\log n}}}{n}.\label{serout}
\end{equation}
If $n>928$, the diagonal $\delta(q_n)$ is conjectured to contain fewer than  $\sqrt{n}\log n$ elements greater than $2$.
Thus, as $n$ increases, $\rho_n(k)$ tends to $\frac{1}{2}$  if $k\in\{0, 2\}$, and to 0 otherwise. 
Formula (\ref{serout}) is a direct application of the \textcolor{index}{law of the iterated logarithm}.  The bound in~(\ref{serout}) is the
strongest possible asymptotically under perfect randomness. At a magic prime where 0 and 2 alternate (thus non random), $2n\log\log n$ can be lowered by orders of magnitude. In my conjecture~\ref{conjhk}, part 4 claims that~(\ref{serout}) is always satisfied for the prime number sequence if $n>770$. But a weak version, say with $2 n\log\log n$ replaced by $n^{3/2}$, is enough for our needs. 
When only the weak version is satisfied, the right diagonal is said to be \textcolor{index}{weakly balanced}.  
If $\delta(q_n)$ is balanced, you need additional conditions
on $\delta(q_n)$ to guarantee that it remains true for $\delta(q_{n+1})$. These conditions are distinct from randomness; they~pertain
 to the lengths of the successive runs of 0 and 2 in the 0-2 cycle of the right diagonal $\delta(q_n)$.  

\noindent Now I can state and prove the main theorem~\ref{ti4r}, starting with lemmas~\ref{lemshi} and~\ref{berlot}.
\vspace{2ex}
\begin{lemma}\label{lemshi}
Let $q_1,\dots,q_n$ be a valid sequence ($q_k\geq q_{k-1}+2$) with right diagonal $\delta(q_n) = [\,\delta_0(q_n),\dots,\delta_{n-1}(q_n)\,]$. Thus, $\delta_k(q_n) = |\delta_{k-1}(q_n) - \delta_{k-1}(q_{n-1})|$ with $\delta_0(q_n)=q_n$.
Let $g_k = q_k - q_{k-1}$ be the first order differences. Then for $k \geq 2$, we have:
\begin{align}
%\delta_k(q_n) & \leq \max\big(g_n, g_{n-1},\dots,  g_{n-k+1}\big) - 2.\label{minal}\\
\delta_k(q_n) & \leq \sup_{1\leq j\leq k} (g_{n-j+1} - 2),\label{minal}\\
%\delta_k(q_n) & \leq \max\big(g_n; \delta_1(q_{n-1}),\dots,\delta_{k-1}(q_{n-1})\big).\label{minal2} \\
\delta_k(q_n) & \leq \max\Big[\delta_1(q_n), \sup_{1\leq j < k}\delta_j(q_{n-1})\Big].\label{minal2}
\end{align}
The maximum in~(\ref{minal}) cannot be lowered without restrictions on the sequence. Thus, 
$\delta(q_n)$ may have reversals;~its second largest value after $q_n=\delta_0(q_n)$ is not always $g_n=\delta_1(q_n)$. However
all values besides $q_n$ are bounded by the maximum gap in $g_2,\dots, g_n$. 
For prime numbers ($q_n=p_n$), the max gap in question is $< p_n^{0.525}$. 
\end{lemma}   
\vspace{1ex}
\begin{minipage}{\textwidth}
{\bf Proof}: By recursion. For instance, $\delta_3(q_n) =|\delta_2(q_n) - \delta_2(q_{n-1})|$. 
Assuming $\delta_2(q_n)\in [0, \max(g_n, g_{n-1}) - 2]$ and 
$\delta_2(q_{n-1})\in[0, \max(g_{n-1}, g_{n-2}) - 2]$, we have
\end{minipage}
\begin{align*}
\delta_3(q_n) & \leq \max\Big[\max(g_n, g_{n-1})-2, \, \max(g_{n-1}, g_{n-2})-2\Big] \\
      & = \max\big(g_n, g_{n-1}, g_{n-2}\big)-2.
\end{align*}
Thus the result is also true for $\delta_3(q_n)$. The inequality comes from this: if $x\in[a, b]$, $y\in[c, d]$ and
$z = |x-y|$, then $z\leq \max(|c-b|, |d-a|)$. 
Also, $\delta_1(q_n) = g_n \in [2, \, g_n]$. That's where the 2 comes from. 
This proves~(\ref{minal}). 
\vspace{1ex}\\
\noindent The proof for (\ref{minal2}) is based on a simple generalization of the following argument:
Using twice the recursive formula that defines $\delta_k(n)$, we have 
$\delta_k(n) = | | | \delta_{k-3}(q_n)  - \delta_{k-3}(q_{n-1})  | - \delta_{k-2}(q_{n-1})  | - \delta_{k-1}(q_{n-1}) |$. 
Since for any $a, b, c, d$, $|||a-b|-c|-d|\leq \max(a, b, c, d)$, we conclude that
$\delta_k(n) \leq\max\big(\delta_{k-3}(q_n);  \delta_{k-3}(q_{n-1}), \delta_{k-2}(q_{n-1}), \delta_{k-1}(q_{n-1})\big) $.
 Iteratively apply the same recursion to the leftmost argument in the maximum to complete the proof.  
\vspace{1ex}\\
\noindent The last part, stating that $\delta_k(q_n) < p_n^{0.525}$ for the prime number sequence, comes from a famous
theorem proved in 2001 by Baker~\cite{pq111}. $\square$
\vspace{1ex}

\noindent In practice, for any fixed $n$, the values $\delta_k(q_n)$ decay very fast as $k$ increases, with ups and downs, and enter the 0-2 cycle in fewer
than $\sqrt{n}\log n$ steps. Proving this is the main challenge. According to~(\ref{minal2}), they can potentially bounce~back several times and stay elevated for very long.  Table~\ref{table:rassel} shows how this could happen. Each row features a right diagonal $\delta_n(q_n)$
in the prime number sequence. The third column represents the prime gap $g_n =\delta_1(q_n)$. 
Rows with $n\leq 212$ list the real gaps. Beyond $n=212$, I replaced prime gaps by other values 
(in yellow) chosen to keep pushing the large rectangled elements to the right, towards
the end of $\delta(q_n)$. The main findings~are as follows:
\vspace{1ex}
\begin{itemize}
\item There is a large gap at $n=205$. Not a problem, $\delta(q_{205})$ looks good. 
It is followed by a tiny gap at~$n=206$. Again not a problem yet. Between $n=205$ and $n=209$, the bold values 18, 16, 14, 14, 12
keep decreasing on average. So far, so good. 
\item Between $n=210$ and $n=215$, there is no decrease. The impact of the large gap followed by a small one at $n=205, \,206$ (resp.) is now
fully felt. The length of the long tail of $\delta(q_n)$, starting after the bold value, still remains the same as $n$ increases. But now that tail {\em always} starts with the large value 12, 
increasing the risk that at some future $n$, the sequence could fail if the streak continues.  
\item In reality, at $n=213$, the true value of $\delta_1(q_n)$ is 4, not 6, breaking the streak. I changed it to 6 to check how
long the bad streak could last in a random sequence.  
Instead of 6, I could have picked up 2 with same result. But it would have created 4, 2, 2 in successive gaps, which is a 
\textcolor{index}{forbidden prime gap constellation}, see section~\ref{forbidden}. Following 6, I had four options
for $\delta_1(q_{214})$: 2, 6, 10, 14. Only 10 was not forbidden, so I chose 10. I had more acceptable options for 
$\delta_1(q_{215})$. Forbidden patterns severely limit the progression of a streak, but not completely. 
What will kill the streak is that soon enough, your only options for the next prime will be composite numbers due
 to the rarity of primes. For instance, in the table, $q_{213}=1309$~is~not a prime.
\item The scenario described here may be more complex and involve multiple parallel streaks in table~\ref{table:rassel},
each caused upstream by a large gap followed by a small one. Also, the 12 could move to the right at a faster pace, and even increase,
further reducing the chance of success. Changing $\delta_1(q_{213})$ from 6 to 38 does both. But it kills the 12 streak,
starting a new one with even larger values, concerning because $\delta_1(q_{212})=2$ is a tiny gap. 
Incidentally, the real gap at $n=217$ is $\delta_1(q_{217})=34$, a record. And all the fear just described is vastly
exaggerated. But for a formal proof, no scenario can be a priori ruled out, even if insanely unlikely. 
\end{itemize}
\vspace{2ex}
\begin{lemma}\label{berlot}
Let $q_1,\dots, q_{n-1}$ be a valid, successful sequence and $\nu_2(q_{n-1})$ be the number of elements equal~to~2~in 
the 0-2 cycle of its right diagonal $\delta(q_{n-1})$. Then $q_1,\dots, q_{n}$ also succeeds if
$g^*_n \leq 2\nu_2(q_{n-1})+2$, where 
\begin{equation}
g^*_n = \max\big(g_2, g_3, \dots, g_n\big). % q_2-q_1, q_3-q_2, q_4-q_3, \dots, q_n - q_{n-1}) 
\end{equation}
As usual, $g_k = \delta_1(q_k) = q_k - q_{k-1}$ is the gap between $q_k$ and $q_{k-1}$. This lemma is not specific to prime numbers.
\end{lemma}
\vspace{1ex}
{\bf Proof}: I pictured the two right diagonal $\delta(q_{n-1})$ and $\delta(q_n)$ in table~\ref{table:rasselced}. Blocks of the same color have the~same number of elements.
The gray block in the top row is $\delta(q_{n-1})$'s 0-2 cycle. It starts at index $k=\tau_n$ and it is followed by a single
value in green, always an odd integer but in this case equal to 1 because I assumed that $\delta(q_{n-1})$ is successful. 
This 0-2 cycle must exist to guarantee success: see theorem~\ref{easy2}. 

\begin{table}[H]
\centering
\renewcommand{\arraystretch}{1.2}
%\small
\begin{tabular}{
%|>{\centering\arraybackslash}p{1.2cm}
|>{\centering\arraybackslash}p{0.8cm}|
>{\centering\arraybackslash}p{0.8cm}|
>{\centering\arraybackslash}p{2cm}|
>{\centering\arraybackslash}p{0.8cm}|
>{\centering\arraybackslash}p{8.15cm}|
>{\centering\arraybackslash}p{0.80cm}|
>{\centering\arraybackslash}p{0.80cm}|
}

%\toprule\\[-4.1ex]
	\cellcolor{black} \textcolor{white}{\bf 0}	&	\cellcolor{black} \textcolor{white}{\bf 1}	&	\cellcolor{black} \textcolor{white}{$\boldsymbol{\cdots}$}	&
	\cellcolor{black} \textcolor{white}{$\boldsymbol{\tau_n}$}	& \cellcolor{black} \textcolor{white}{$\boldsymbol{\cdots}$}&	\cellcolor{black} \textcolor{white}{\small $\boldsymbol{n{-}2}$}	&  \cellcolor{black} \textcolor{white}{\small $\boldsymbol{n{-}1}$} \\ 
\addlinespace[-2pt] 
%\hline
%\hline
%\toprule[1pt]\\[-4.3ex] 
\cellcolor{blue2}  $q_{n-1}$	& \cellcolor{blue2} $g_{n-1}$	 &	\cellcolor{cyan2}	 $< g^*_{n-1}$	&	\multicolumn{1}{c}{\cellcolor{black2}	}	&\cellcolor{black2}	\hspace{-0.9cm}	0-2 cycle (long tail, 0 and 2)	&\cellcolor{green3}	 1	& \cellcolor{white}	\\
\hline
\cellcolor{blue2}  $q_{n}$	& \cellcolor{blue2}  $g_{n}$	&	\cellcolor{cyan2}	 $< g^*_{n}$	&\cellcolor{yellow2}		$v_n$	&	\multicolumn{1}{c}{\cellcolor{black2}	\hspace{0.60cm} decreasing values $\leq v_n + 2$ (also, $v_n< g^*_n$)}	&	\cellcolor{black2}		&\cellcolor{green3}	\bf ?\\
%Cell 1 & Cell 2 & Cell 3 \\ \hline
%Cell 4 & \multicolumn{1}{c}{Cell 5} & Cell 6 \\ \hline
%Cell 7 & Cell 8 & Cell 9 \\ \hline
\hline
\end{tabular}
\caption{$\delta(q_{n-1})$ (top) and $\delta(q_n)$ (bottom): joint block structure.  Labels represent the index $k$.}
\label{table:rasselced}
\end{table}
\noindent Now let's focus on the bottom row, $\delta(q_n)$. By construction, the gray block has the same number of elements as the 0-2 cycle in $\delta(q_{n-1})$. We don't know what's in it.
By construction, it is followed by the last element in green, and preceded by the yellow element whose unknown value is denoted as $v_n$. In general the gray block 
in $\delta(q_n)$ is not the 0-2 cycle. The 0-2 cycle may or may not exist. 
What we know is that except for $\delta_0(q_n)=q_n$ and $\delta_1(q_n)=g_n$,
we have $\delta_k(q_n)< g^*_n$, by virtue of~(\ref{minal}) in lemma~\ref{lemshi}. In particular, this applies to $v_n$. Also,~all values are even integers, and
the total length of $\delta(q_n)$ is $n$, starting at index $k=0$ and ending at $k=n-1$.
\vspace{2ex}\\
\noindent If $v_n=0$, success is trivial. So, let's assume that $v_n\geq 2$. By definition, $\delta_k(q_n) = |\delta_{k-1}(q_n) -\delta_{k-1}(q_{n-1})|$.  
Also, at $k=\tau_n$, $\delta_k(q_n)=v_n$. 
If $k>\tau_n$, we are in the gray block in the bottom row, and the corresponding value $\delta_{k-1}(q_{n-1})$ in the top row is $\leq 2$. 
Thus, $\delta_{k}(q_n) \in \{ \delta_{k-1}(q_n) - 2, \delta_{k-1}(q_n) \}$ unless $\delta_{k-1}(q_n)=0$. We can ignore that exception: when it happens, success
is guaranteed. The challenge is when $v_n$ is large: do we have enough runway, that is, is the the gray block long enough, to end up with 
$\delta_{n-1}(q_n)=1$ in the rightmost cell? 
\vspace{2ex}\\
\noindent To answer that question, let's look at an example. Say $v_n=14$ at index $k=5$, with the 0-2 cycle in $\delta(q_{n-1})$ as featured
in the top row in table~\ref{table:rasselcedggg}. Then the sequence in the middle row succeeds: it ends with 1 in the green cell. 
\vspace{1ex}

\begin{table}[H]
\centering
\renewcommand{\arraystretch}{1.25}
\setlength{\tabcolsep}{0pt}
\small
\begin{tabular}{
%|>{\centering\arraybackslash}p{1.2cm}
|>{\centering\arraybackslash}p{0.81cm}|
>{\centering\arraybackslash}p{0.8cm}|
>{\centering\arraybackslash}p{2.3cm}|
>{\centering\arraybackslash}p{0.8cm}|
>{\centering\arraybackslash}p{0.80cm} %%
>{\centering\arraybackslash}p{0.80cm} %%
>{\centering\arraybackslash}p{0.80cm} %%%
>{\centering\arraybackslash}p{0.80cm} %%%
>{\centering\arraybackslash}p{0.80cm} %%%
>{\centering\arraybackslash}p{0.80cm} %%%
>{\centering\arraybackslash}p{0.80cm} %%%%
>{\centering\arraybackslash}p{0.80cm} %%%%
>{\centering\arraybackslash}p{0.80cm} %%%%
>{\centering\arraybackslash}p{0.80cm} %%%%
>{\centering\arraybackslash}p{0.80cm}| %%%%
>{\centering\arraybackslash}p{0.823cm}|
>{\centering\arraybackslash}p{0.80cm}|
}

	\cellcolor{black} \textcolor{white}{\bf 0}	&	\cellcolor{black} \textcolor{white}{\bf 1}	&	\cellcolor{black} \textcolor{white}{$\boldsymbol{\cdots}$}
&\cellcolor{black} \textcolor{white}{$\boldsymbol{5}$}	
%&\cellcolor{black} \textcolor{white}{xx} %%
&\multicolumn{11}{c|}{\cellcolor{black}	\textcolor{white}{$\boldsymbol{\cdots}$}}
%& \cellcolor{black} \textcolor{white}{$\boldsymbol{\cdots}$}
&\cellcolor{black} \textcolor{white}{\bf 17}	&  \cellcolor{black} \textcolor{white}{\bf 18} \\ 
\addlinespace[-2pt]

\cellcolor{blue2}  $q_{n-1}$	& \cellcolor{blue2} $g_{n-1}$	 &	\cellcolor{cyan2}	 $< g^*_{n-1}$ 	
&\cellcolor{black2}	0 %%
&\cellcolor{black2}	\textcolor{dkgreen}{2} %%	
&\cellcolor{black2}	\textcolor{dkgreen}{2} %%%	
&\cellcolor{black2}	0 %%%	
&\cellcolor{black2}	\textcolor{dkgreen}{2} %%%	
&\cellcolor{black2}	0 %%%	
&\cellcolor{black2}	\textcolor{dkgreen}{2} %%%
&\cellcolor{black2}	0 %%%
&\cellcolor{black2}	0 %%%
&\cellcolor{black2}	0 %%%
&\cellcolor{black2}	\textcolor{dkgreen}{2} %%%
&\cellcolor{black2}	\textcolor{dkgreen}{2} %%%
&\cellcolor{green3}	 \textcolor{dkgreen}{\bf 1}	& \cellcolor{white}	\\

\hline

\cellcolor{blue2}  $q_{n}$	& \cellcolor{blue2}  $g_{n}$	&	\cellcolor{cyan2}	 $< g^*_{n}$	
&\cellcolor{yellow2} 14	%%
&\cellcolor{black2}	14	%%
&\cellcolor{black2}	\textcolor{dkgreen}{12} %%%	
&\cellcolor{black2}	\textcolor{dkgreen}{10} %%%	
&\cellcolor{black2}	10 %%%	
&\cellcolor{black2}	\textcolor{dkgreen}{8} %%%	
&\cellcolor{black2}	8 %%%
&\cellcolor{black2}	\textcolor{dkgreen}{6} %%%	
&\cellcolor{black2}	6 %%%%	
&\cellcolor{black2}	6 %%%%	
&\cellcolor{black2}	6 %%%%	
&\cellcolor{black2}	\textcolor{dkgreen}{4} %%%%	
&\cellcolor{black2}	\textcolor{dkgreen}{2} %%%%	
&\cellcolor{green3}	 \textcolor{dkgreen}{\bf 1}\\

%\hline

\cellcolor{blue2}  $q_{n}$	& \cellcolor{blue2}  $g_{n}$	&	\cellcolor{cyan2}	 $< g^*_{n}$	
&\cellcolor{yellow2} 16	%%
&\cellcolor{black2}	16	%%
&\cellcolor{black2}	\textcolor{dkgreen}{14} %%%	
&\cellcolor{black2}	\textcolor{dkgreen}{12} %%%	
&\cellcolor{black2}	12 %%%	
&\cellcolor{black2}	\textcolor{dkgreen}{10} %%%	
&\cellcolor{black2}	10 %%%
&\cellcolor{black2}	\textcolor{dkgreen}{8} %%%	
&\cellcolor{black2}	8 %%%%	
&\cellcolor{black2}	8 %%%%	
&\cellcolor{black2}	8 %%%%	
&\cellcolor{black2}	\textcolor{dkgreen}{6} %%%%	
&\cellcolor{black2}	\textcolor{dkgreen}{4} %%%%	
&\cellcolor{red2}	 \textcolor{red}{\bf3}\\

\hline
\end{tabular}
\caption{Sample $\delta(q_{n-1})$ (top) and $\delta(q_n)$ (two bottom rows).  Labels represent the index $k$.}
\label{table:rasselcedggg}
\end{table}

\noindent Now increase $v_n$ to 16, as shown in the bottom row. Then the sequence fails, as it ends with 3.
Clearly, success at $q_n$ depends on how many 2 you have in the 0-2 cycle of $\delta(q_{n-1})$, and how large $v_n$ is.
More specifically, the sequence succeeds if and only if
$v_n \leq  2\nu_2(q_{n-1}) + 2$. In my example, $\nu_2(q_{n-1})=6$ and $n=19$. We have success if
$v_n\leq 14 = 2\times 6 + 2$ and failure if $v_n>14$
\vspace{1ex}

\noindent Since $v_n < g^*_n$, success is guaranteed if $g^*_n\leq 2\nu_2(q_{n-1}) + 2$. $\square$
\vspace{2ex}

\begin{theorem}\label{ti4r} 
 Let $q_1,\dots, q_{n-1}$ be a valid, successful sequence and $\nu_2(q_{n-1})$ be the number of elements equal to 2 in 
the 0-2 cycle of its right diagonal $\delta(q_{n-1})$. \!Assume that $g^*_n < n^\alpha$ where $g^*_n$ is the record gap
at $q_n$, that is, the largest gap ever encountered so far. If there is a constant $\beta$ such that 
\begin{equation}
\nu_2(q_{n-1}) > n^\beta, \quad \text{ with } \beta > \alpha, \label{riji}
\end{equation}
then for $n$ large enough, the sequence also succeeds at $q_n$ if it succeeds at $q_{n-1}$.
\end{theorem}
\vspace{1ex}
{\bf Proof}: It is a corollary of lemma~\ref{berlot}. For the prime number sequence, $\alpha=0.525$ works as proved by Baker in 2001, see~\cite{pq111}.  
Independently and unrelated to primes, $\beta=0.99$ works, see part 4 in conjecture~\ref{conjhk}. 
Heuristics suggest that $\nu_2(q_{n-1}) \sim n/2$, an even stronger claim
 tied to~(\ref{serout}) about balanced right diagonals. $\square$
\vspace{2ex}\\
\noindent Essentially, my strongest claim is that the 0 and 2 in the 0-2 cycle of $\delta(q_{n-1})$ behave like the binary digits 0 and 1 in a \textcolor{index}{normal number}. 
Yet no one knows for a number like $\pi$ if the proportion of 0 and 1 actually exists. But I proved in~\cite{0and1new}
that for most numbers $x$ included $x=\pi$, infinitely many times, the proportion of 1 is between 25\% and 75\%
either for $x$, $x+\frac{1}{3}$, or both. Likewise, for $\sqrt{2}$, it was proved that the number of 1 in the first $n$ digits must be above 
$1.41 n^\beta$ with $\beta=0.5$. Given the analogy between normal numbers in base 2 and balanced diagonals, this is encouraging in terms
of the extra work needed to prove Gilbreath's conjecture from where I stand now.

It is easy to prove that if $\delta(q_n)$ is strictly decreasing and $g_n^*<n$, then the sequence succeeds at $q_n$ as it~has enough runway to reach 1 at the end. Also, if $q_n$ is a magic prime, the sequence succeeds 
both at $q_n$ by virtue of theorem~\ref{easyn}, but also at $q_{n+1}$. This also applies to sequences other than prime numbers, where ``magic prime"~is 
replaced by $q_{n}-q_{n-1} = 2$ along with $\delta(q_{n})$ consisting of 0 and 2 only (besides the first and last elements). It also requires
 the condition $g_n^*<n^\beta$.
I now summarize the unproved insights discovered so far in the next conjecture~\ref{conjhk}.

\vspace{2ex}
\begin{conjecture}\label{conjhk} The following was prepared for the prime number sequence. Some items in the list below may apply
to other valid sequences even outside the corridor (with much faster growth), under specific conditions. 
\vspace{1ex}
\begin{enumerate}
\item[(1)] If the sequence $q_1,\dots, q_n$ succeeds,  adding $q_{n+1}= q_{n}$  
preserves success. 
If it fails and adding $q_{n+1}$ fails, then replacing $q_{n+1}$ by any larger integer, will also fail. 
%\item[(2)] 
%The above items 
Thus there is a maximum $q^+_n\geq q_n$ such that~for~any $q\in\{q_n,\dots, q^+_n\}$,
the augmented sequence with $q_{n+1} = q$, succeeds, but fails if you choose $q_{n+1}>q^+_n$. 
\item[(2)] Thanks to this, you can run a binary search to efficiently find $q^+_n$, using the bisection method with 
$q_n$~and $2q_n$ as the lower and upper bounds (resp.) for the initial search interval. 
Ignore the triangle, just work on the right diagonal: update it incrementally as needed. It ends with 1 if and only if the sequence succeeds. 
\item[(3)] The rightmost picture in figure~\ref{fig:aajjjjh6y} strongly suggests that for primes ($q_n=p_n$), we
have $q^+_n = q_n + n + o(n)$. Also $q^-_n = q_n -n +o(n)$ if you allow for negative gaps. With
 success in all cases when $q^-_n\leq q_{n+1} \leq q^+_n$.
%\item[(6)] The previous claim, 
If true, it trivially proves the Gilbreath conjecture as there is always a prime
between $p_n$ and $p_n+n$, if $p_n\geq 11$. But it needs some care: $o(n)$ could be as low as $-\log n$ or $-n/\log n$.
\item[(4)] In the right diagonal  $\delta(q_n)$,  the section before the 0-2 cycle has fewer than $\sqrt{n}\log n$ elements if $n>928$. This is  more than what we need to prove Gilbreath's conjecture. 
Finally, for the prime number sequence, $\delta(q_n)$ is \textcolor{index}{balanced} for all $n>770$. Also,~(\ref{riji}) is satisfied if $\beta=0.99$ and 
$n>2535$, or $\beta=0.55$ and $n>16$. To prove Gilbreath, $\beta>0.525$ is enough as it is tied to the best proven bound on prime gaps. 
\end{enumerate}
\end{conjecture}
\vspace{2ex}
I tested the conjecture on the prime numbers and a few other sequences, with success up to the last value~that~I tried ($n=10^4$). 
But there are short artificial sequences where (1) is not satisfied, see tables~\ref{table:rab1} and~\ref{table:rab2}. These cases do
not jeopardize my conclusions; chance of failure is always much higher at the beginning, and you can replace
the requirement $q^+_{n}\geq q_n$ by $q^+_{n}\geq q_n + 2$ to drastically reduce the risk, ruling out non-valid
sequences as~they are more prone to failures. 

The sequences covered in the conjecture not only have a 0-2 cycle at all $q_n$ (guaranteeing success by virtue of theorem~\ref{easy2}),
 but that cycle occupies almost the entire space in the right diagonal. Let $\omega_n$ denote the length of the section before the 0-2 cycle. 
For $n$ large enough, we have $\omega_n < \sqrt{n}\log n$ while the diagonal $\delta(q_n)$ has $n$ elements.
None of this is proved yet despite massive empirical evidence.  For $n\leq 2 \times 10^4$, we have the following statistics:
\vspace{1ex}
\begin{itemize}
\item Magic primes ($\omega_n = 0$) represent 0.17\% of the $q_n$'s. With guaranteed success. Since we also have success at the successor of a magic prime, 
we have ``proven success" for 0.34\%" of the $q_n$'s.
\item When the right diagonal is strictly decreasing, success is also proved: there is enough runway to reach 1 at the end. This covers an additional 
4.90\% of the $q_n$'s, with provable success. 
\item When $\omega_n\leq 10$ at $q_n$, we obviously have a 0-2 cycle if $n \geq 12$, and thus success at $q_n$, and also success at $q_{n+1}$ based on theory. This covers 
52.05\% of the $q_n$'s. 
\item When $\omega_n\leq 32$ at $q_n$, we obviously have a 0-2 cycle if $n \geq 34$, and thus success at $q_n$, and also success at $q_{n+1}$ based on theory. This covers 
96.29\% of the $q_n$'s. 
\item There is a little bias: $\omega_n$ is less likely to be odd than even. 
\end{itemize}
\vspace{1ex}
\noindent In short, $w_n$ is overwhelmingly very small, much smaller than $\sqrt{n}\log n$ for almost all $q_n$. The high proportions just cited are of course lower
if you look at a bigger window, say $n\leq 10^5$, but they decrease very slowly with~$n$. 
Also, if you replace $q_{n+1}$ by much larger $q_n^+$, you still succeed (for primes, $q^+_n \lesssim q_{n} +n$ while $q_{n+1} \lesssim q_n + n^{0.55}$ due to bounds on 
record prime gaps). Yet, replacing $q_{n+1}$ by $q^+_n + 2$ leads to failure. Understanding what makes the former succeed and the later fail, may bring valuable insights.  
Another worthwhile investigation consists in analyzing the spikes in figure~\ref{fig:aajjjjh6y}. 
They are uncorrelated to magic primes locations, and absent in random primes in figure~\ref{fig:aajjjd3uyvo}. 
Finally, I developed a lighter version of Gilbreath's problem, where the recursion to compute $\delta(q_n)$,
namely  $\delta_k(q_n) = |\delta_{k-1}(q_n) - \delta_{k-1}(q_{n-1})|$, is replaced by
\begin{equation}
\delta_k(q_n) = \max\Big(|\delta_{k-1}(q_n) - \delta_{k-1}(q_{n-1})|, 1\Big). 
\end{equation}
It does not allow for 0 in the triangle (the minimum value is now 1), thus eliminating some complexity. In
many cases, failure in the standard problem also implies failure in this lighter version, but not always.  
The sequence in table~\ref{table:gilb2} is an exception.  The prime number
sequence still succeeds.  I encourage the reader to play with this easier version,
establish the equivalent of conjecture~\ref{conjhk}  (let's call it the \textcolor{index}{baby conjecture}), and try to prove it, before attacking the standard version. 
In the python code, to use the light version, add \texttt{delta='light'} in the arguments when calling the the \texttt{check} function from the 
\texttt{gilbreath\_lib} library. 

\vspace{1ex}

%---------------now

\begin{table}[ht]
\centering
\renewcommand{\arraystretch}{1.2}
\small
%\begin{tabular}{c cc|c|c|c|   c|c|c|c|c|c|  c|c|c|c|c|c|c} % Total 7 underlying columns
\begin{tabular}{
>{\centering\arraybackslash}p{0.9cm}
>{\centering\arraybackslash}p{0.35cm}
>{\centering\arraybackslash}p{0.35cm}
>{\centering\arraybackslash}p{0.35cm}
>{\centering\arraybackslash}p{0.35cm}
>{\centering\arraybackslash}p{0.35cm}|
>{\centering\arraybackslash}p{0.35cm}
>{\centering\arraybackslash}p{0.35cm}
>{\centering\arraybackslash}p{0.35cm}
>{\centering\arraybackslash}p{0.35cm}
>{\centering\arraybackslash}p{0.35cm}
>{\centering\arraybackslash}p{0.35cm}|
>{\centering\arraybackslash}p{0.35cm}
>{\centering\arraybackslash}p{0.35cm}
>{\centering\arraybackslash}p{0.35cm}
>{\centering\arraybackslash}p{0.35cm}
>{\centering\arraybackslash}p{0.35cm}
>{\centering\arraybackslash}p{0.35cm}
>{\centering\arraybackslash}p{0.35cm}
}
%\toprule[1pt] 
%\hline
%\hline
\toprule[0.96pt]\\[-4.1ex]
  & \multicolumn{5}{c|}{ $n \leq 5$} &  \multicolumn{6}{c|}{$n=6$ and $q_6 \leq q^+_5 = 25$} & \multicolumn{7}{c}{$n=6$ and $q_6 > q^+_5 = 25$} \\
 % & \multicolumn{5}{c|}{ $\delta(q_n)$, $n \leq 5$} &  \multicolumn{13}{c}{Search for $q_5^+$} \\
\hline
%\midrule
$n$ & 1 & 2 & 3 & 4 & 5 & 6   & 6 & 6   & 6   & 6 & 6    & 6   & 6 & 6    & 6   & 6 & 6    &6  \\
%\hline
$q_n$ & \textcolor{dkgreen}{2} & \textcolor{dkgreen}{3} & \textcolor{dkgreen}{5} & \textcolor{dkgreen}{9} & \textcolor{dkgreen}{15} & \textcolor{red}{15} & \textcolor{dkgreen}{17} & \textcolor{dkgreen}{19} &	\textcolor{dkgreen}{21}	&	\textcolor{dkgreen}{23}	&	\rectangled[green2]{25}	&	\textcolor{red}{27}	&	\textcolor{red}{29}	&	\textcolor{red}{31}	&	\textcolor{red}{33}	&	\textcolor{red}{35}	&	\textcolor{red}{37}	&	\textcolor{red}{39}	\\
%\hline
\toprule[0.96pt]\\[-4.1ex]
$\delta_1(q_n)$  & & \textcolor{dkgreen}{1} & 2 & 4 & \cellcolor{yellow} 6 & \cellcolor{lightgray} 0 & \cellcolor{lightgray} 2 &  \cellcolor{lightgray} 4  &	\cellcolor{lightgray} 6	&\cellcolor{lightgray}	8	&\cellcolor{lightgray}	10	&	12	&	14	&	16	&	18	&	20	&	22	&	24	\\
$\delta_2(q_n)$  &  &    & \textcolor{dkgreen}{1}  &  2 & \cellcolor{yellow} 2 &\cellcolor{lightgray} \bf 6  &\cellcolor{lightgray} \bf 4 & \cellcolor{lightgray} \bf 2 &\cellcolor{lightgray}	0	&\cellcolor{lightgray}	2	&	\cellcolor{lightgray} 4	&	6	&	8	&	10	&	12	&	14	&	16	&	18	\\
$\delta_3(q_n)$  &  &   &   & \textcolor{dkgreen}{1} & \cellcolor{yellow} 0 & \cellcolor{lightgray} \bf 4   &\cellcolor{lightgray} \bf 2 &   \cellcolor{lightgray} \bf 0   	&\cellcolor{lightgray}	\bf 2	&\cellcolor{lightgray}	0	&\cellcolor{lightgray}	2	&	4	&	6	&	8	&	10	&	12	&	14	&	16	\\
$\delta_4(q_n)$  &   &  &   &  &\cellcolor{yellow} \textcolor{dkgreen}1 &\cellcolor{lightgray}  \bf 4    &\cellcolor{lightgray}\bf 2 &\cellcolor{lightgray}  \bf 0  &\cellcolor{lightgray}	\bf 2	&\cellcolor{lightgray}	0	&\cellcolor{lightgray}	2	&	4	&	6	&	8	&	10	&	12	&	14	&	16	\\
\hline
$\delta_5(q_n)$  &  &   &  &  &  &  \textcolor{red}{3}    & \textcolor{dkgreen}{1} &      \textcolor{dkgreen}{1}   &	\textcolor{dkgreen}{1}	&	\textcolor{dkgreen}{1}	&	\textcolor{dkgreen}{1}	&	\textcolor{red}{3}	&	\textcolor{red}{5}	&	\textcolor{red}{7}	&	\textcolor{red}{9}	&	\textcolor{red}{11}	&	\textcolor{red}{13}	&	\textcolor{red}{15}	 \\ 

%\bottomrule[1pt]
\end{tabular}
\caption{Example with failure at $q_{n+1}=q_n$, even though $q_n$ succeeds}
\label{table:rab1}
\end{table}

%---------------- new table below

\begin{table}[ht]
\centering
\renewcommand{\arraystretch}{1.2}
\small
%\begin{tabular}{c|ccccc | ccccccc| cccccc} % Total 7 underlying columns
\begin{tabular}{
>{\centering\arraybackslash}p{0.9cm}
>{\centering\arraybackslash}p{0.35cm}
>{\centering\arraybackslash}p{0.35cm}
>{\centering\arraybackslash}p{0.35cm}
>{\centering\arraybackslash}p{0.35cm}
>{\centering\arraybackslash}p{0.35cm}|
>{\centering\arraybackslash}p{0.35cm}
>{\centering\arraybackslash}p{0.35cm}
>{\centering\arraybackslash}p{0.35cm}
>{\centering\arraybackslash}p{0.35cm}
>{\centering\arraybackslash}p{0.35cm}
>{\centering\arraybackslash}p{0.35cm}
>{\centering\arraybackslash}p{0.35cm}|
>{\centering\arraybackslash}p{0.35cm}
>{\centering\arraybackslash}p{0.35cm}
>{\centering\arraybackslash}p{0.35cm}
>{\centering\arraybackslash}p{0.35cm}
>{\centering\arraybackslash}p{0.35cm}
>{\centering\arraybackslash}p{0.35cm}
}
%\toprule[1pt] 
%\hline
%\hline
\toprule[0.90pt]\\[-4.1ex]
  & \multicolumn{5}{c|}{ $n \leq 5$} &  \multicolumn{7}{c|}{$n=6$ and $q_6 \leq q^+_5 = 17$} & \multicolumn{6}{c}{$n=6$ and $q_6 > q^+_5=17$} \\
 % & \multicolumn{5}{c|}{ $\delta(q_n)$, $n \leq 5$} &  \multicolumn{13}{c}{Search for $q_5^+$} \\
\hline
%\midrule
$n$ & 1 & 2 & 3 & 4 & 5 & 6   & 6 & 6   & 6   & 6 & 6    & 6   & 6 & 6    & 6   & 6 & 6    &6  \\
%\hline
$q_n$ & \textcolor{dkgreen}{2} & \textcolor{dkgreen}{3} & \textcolor{dkgreen}{5} & \textcolor{red}{11} & \hspace{0.8ex}\textcolor{dkgreen}{5}\hspace{0.8ex} & \textcolor{dkgreen}{5} & \textcolor{dkgreen}{7} & \textcolor{dkgreen}{9} &	\textcolor{red}{11}	&	\textcolor{dkgreen}{13}	&	\textcolor{dkgreen}{15}	&	\rectangled[green2]{17}	&	\textcolor{red}{19}	&	\textcolor{red}{21}	&	\textcolor{red}{23}	&	\textcolor{red}{25}	&	\textcolor{red}{27}	&	\textcolor{red}{29}	\\
%\hline
\toprule[1pt]\\[-4.1ex]
$\delta_1(q_n)$  & & \textcolor{dkgreen}{1} & 2 & 6 & \cellcolor{yellow} 6 & \cellcolor{lightgray} 0 & \cellcolor{lightgray} 2 &  \cellcolor{lightgray} 4  &	\cellcolor{lightgray} 6	&\cellcolor{lightgray}	8	&\cellcolor{lightgray}	10	& \cellcolor{lightgray}	12	&	14	&	16	&	18	&	20	&	22	&	24	\\
$\delta_2(q_n)$  &  &    & \textcolor{dkgreen}{1}  &  4 & \cellcolor{yellow} 0 &\cellcolor{lightgray} \bf 6  &\cellcolor{lightgray} \bf 4 & \cellcolor{lightgray} \bf 2 &\cellcolor{lightgray}	0	&\cellcolor{lightgray}	2	&	\cellcolor{lightgray} 4	&	 \cellcolor{lightgray} 6	&	8	&	10	&	12	&	14	&	16	&	18	\\
$\delta_3(q_n)$  &  &   &   & \textcolor{red}{3} & \cellcolor{yellow} 4 & \cellcolor{lightgray} \bf 6   &\cellcolor{lightgray} \bf 4 &   \cellcolor{lightgray} \bf 2   	&\cellcolor{lightgray}	0	&\cellcolor{lightgray}	2	&\cellcolor{lightgray}	4	& \cellcolor{lightgray}	6	&	8	&	10	&	12	&	14	&	16	&	18	\\
$\delta_4(q_n)$  &   &  &   &  &\cellcolor{yellow} \textcolor{dkgreen}1 &\cellcolor{lightgray}  \bf 2    &\cellcolor{lightgray}\bf 0 &\cellcolor{lightgray}  \bf 2  &\cellcolor{lightgray}	\bf 4	&\bf \cellcolor{lightgray}	2	&\cellcolor{lightgray}	0	& \cellcolor{lightgray}	2	&	4	&	6	&	8	&	10	&	12	&	14	\\
\hline
$\delta_5(q_n)$  &  &   &  &  &  &    \textcolor{dkgreen}{1}    &  \textcolor{dkgreen}{1} &      \textcolor{dkgreen}{1}   &	\textcolor{red}{3}	&	\textcolor{dkgreen}{1}	&	\textcolor{dkgreen}{1}	&	\textcolor{dkgreen}{1}	&	\textcolor{red}{3}	&	\textcolor{red}{5}	&	\textcolor{red}{7}	&	\textcolor{red}{11}	&	\textcolor{red}{13}	&	\textcolor{red}{15}	 \\ 

%\bottomrule[1pt]
\end{tabular}
\caption{Example with failure at  $q_{n+1} <q^+_n$, despite success at $q_{n+1}=q_n$}
\label{table:rab2}
\end{table}

\noindent I conclude this section with computational complexity. To check if a sequence with $n$ elements succeeds, using the full table
 is $O(n^2)$. If you stop at the first row with 0 and 2 only, it is $O(n\log n)$. If you rely solely on the right diagonal, it is $O(n)$. 
Likewise, searching for $q^+_n$ takes $O(n)$ with brute force, and $O(\log n)$ with a binary search assuming part (2)
in conjecture~\ref{conjhk} is true. 
Thus, the total number of operations to compute $q^+_n$ for all $k\leq n$ ranges from $O(n^4)$ with the least efficient
implementation, down to $O(n^2 \log n)$ with the best one. 
As a side note, integer polynomial sequences of degree $d$ takes only $d+1$ levels in the triangle to end up in a row consisting of 0 only except for the first element.

\subsection{Canonical form reduction and equivalence classes for sequences}\label{reduxor}

If two sequences $S_1, S_2$ have the same $k$-th order differences for some $k>0$ (the $k$-th row in the triangle beneath the sequence itself), that remains true for any $k'>k$.
Thus, $S_1$ succeeds if and only if $S_2$ succeeds, if you ignore the first $k$ terms in both sequences for which success needs to be checked separately. 
The sequences $S_1, S_2$~are said to be \textcolor{index}{equivalent} at level $k$. If $S_1$ has large gaps and $S_2$ does not, it makes sense to work with $S_2$. 
Actually, you can even work with any sequence $S'_2$ built backwards one level above $k$. And neither $S_2$ nor $S'_2$ need to~be valid: in short, they do not need to be strictly increasing; they can be oscillating with tiny gaps up and down.

Tables~\ref{table:gilb1b} and~\ref{table:gilb1bx} illustrate the concept. Why work with the prime number sequence $S_1$ if you can work~with its equivalent at level 3, 
the sequence $S_2$? The triangles are the same, except for the first two rows at the~top. With $S_2$, success is already guaranteed at level 1, as the first row contains only 0 and 2. 
The numbers highlighted in yellow, on the left, represent the maximum value in the corresponding row. The next question is how~to~build levels backwards, going up starting from below, rather than the other way around. However, we have a problem: the recursion $\delta_k(q_n) = |\delta_k(q_{n-1}) - \delta_{k-1}(q_{n-1})|$ is not invertible. 
In fact, there is a very large number of ``inverses". They are all built with the following multivalued formula, with both options in~(\ref{case2}) equally acceptable:

\begin{numcases}{\delta_{k-1}(q_n)  =}
  \delta_{k}(q_n) + \delta_{k}(q_{n-1}) & \quad if $\delta_{k}(q_n) < \delta_{k}(q_{n-1})$ \label{case1} \vspace{1ex}\\
  \delta_{k}(q_n) \pm \delta_{k}(q_{n-1}) & \quad otherwise\label{case2}
\end{numcases}

\noindent That's how I built $S_2$. It is easy to check that the two tables satisfy both the inverse and standard recursions. The row at
level $k$ is denoted as $\delta_k$, with $k=0$ for the sequence itself. The leftmost element at level $k$ is denoted as $\delta_k(q_1)$, and equal to 1 if $k>0$ 
except in case of failure. When building $S_2$, first its level 2 based on level~3 in table~\ref{table:gilb1b}, then its level 1 recursively, 
then $S_2$ on top of level 1, I chose the minus sign in~(\ref{case2}) whenever possible. This
construction leads to the level 3 \textcolor{index}{canonical form} of $S_1$, in this case $S_2$ based on level 3 at $S_1$.

\begin{table}[H]
\centering
\renewcommand{\arraystretch}{1.0}
\small
\setlength{\tabcolsep}{1.pt}
\parbox{.45\linewidth}{
\centering

\begin{tabular}{*{23}{>{\centering\arraybackslash}p{0.25cm}}}
 %\hline
 % &  &  &        \\  [-2.5ex]
 % 1--grams & 2--grams & 3--grams & 4--grams \\%[0.5ex] 
% \hline 
\hline 
   &  &   &        \\  [-2.2ex]
%\hline
 & 2 & & 3 & & 5 & & 7 & & 11 & & 13 & & 17 & & 19 & & 23 & & 29 & & 31\\
\hline
  &  &  &  & & & & & & & & & &  & & & & & &   \\  [-2ex]
%10 & & 1&  & 2 & \\
\rectangled[yellow]{6} & &  1 & &  2 & &  2 & &  4 & &  2 & &  4 & &  2 & &  4 & & 6 &  & 2\\  %
\rectangled[yellow]{4} & & & 1 & & 0 & & 2 & & 2 & & 2 & & 2 & & 2 & & 2 & & 4 \\ %
2 & & & & 1 & & 2 & & 0 & & 0 & & 0 & & 0 & & 0 & &2 \\%
2 & & & & & 1 & & 2 & & 0 & & 0 & & 0 & & 0 & & 2\\
2 & & & & & & 1 & & 2 & & 0 & & 0 & & 0 & & 2\\
2 & & & & & & & 1 & & 2 & & 0 & & 0 & & 2\\
2 & & & & & & & & 1 & & 2 & & 0 & & 2\\
2 & & & & & & & & & 1 & & 2 & & 2\\ %
1 & & & & & & & & & & 1 & & 0\\
1 & & & & & & & & & & & \circled[green2]{1}\\
 %\hline
\end{tabular}
%\caption{Backend table {\fontfamily{ptm}\selectfont hash stem} (extract)}
\captionsetup{justification=centering}
\caption{The prime numbers sequence $S_1$}
\label{table:gilb1b}
}
%\hfill
\quad\quad
\parbox{.45\linewidth}{
\centering

\begin{tabular}{*{23}{>{\centering\arraybackslash}p{0.25cm}}}
 %\hline
 % &  &  &        \\  [-2.5ex]
 % 1--grams & 2--grams & 3--grams & 4--grams \\%[0.5ex] 
% \hline 
\hline 
   &  &   &        \\  [-2.2ex]
%\hline
 & 2 & & 1 & & 1 & & 1 & & 3 & & 3 & & 1 & & 1 & & 3 & & 3 & & 3\\
\hline
  &  &  &  & & & & & & & & & &  & & & & & &   \\  [-2ex]
%10 & & 1&  & 2 & \\
\rectangled[yellow]{2} & &  1 & &  0& &  0 & &  2 & &  0 & &  2 & &  0 & &  2 & & 0 &  & 0\\  %%%
\rectangled[yellow]{2} & & & 1 & & 0 & & 2 & & 2 & & 2 & & 2 & & 2 & & 2 & & 0 \\ %
2 & & & & 1 & & 2 & & 0 & & 0 & & 0 & & 0 & & 0 & &2 \\%
2 & & & & & 1 & & 2 & & 0 & & 0 & & 0 & & 0 & & 2\\
2 & & & & & & 1 & & 2 & & 0 & & 0 & & 0 & & 2\\
2 & & & & & & & 1 & & 2 & & 0 & & 0 & & 2\\
2 & & & & & & & & 1 & & 2 & & 0 & & 2\\
2 & & & & & & & & & 1 & & 2 & & 2\\ %
1 & & & & & & & & & & 1 & & 0\\
1 & & & & & & & & & & & \circled[green2]{1}\\
 %\hline
\end{tabular}
%\caption{Backend table {\fontfamily{ptm}\selectfont hash stem} (extract)}
\captionsetup{justification=centering}
\caption{Seq. $S_2$ equivalent to $S_1$ at level 3}
\label{table:gilb1bx}

}
\end{table}
\vspace{1ex}

\noindent Let $\lambda$ be the chosen number of levels, for instance, $\lambda=3$ in table~\ref{table:gilb1bx}. You would think that always choosing the minus sign in~(\ref{case2}) leads to a sequence $S_2$ with minimum growth
 as shown in table~\ref{table:gilb1bx}.
In general, this is not the case when $n$ or $\lambda$ is large: $S_2$ grows much faster than $S_1$ as the technique is a variant of the 
\textcolor{index}{greedy algorithm}.  Of course, selecting the plus sign would result in even much faster growth. 
The slowest growth is achieved with a combination of pluses and minuses. 

The benefit of the canonical form is that the largest gap in $S_2$ remains small if $\lambda$ is not too small compared to $n$. 
 We want $\lambda=\lambda_n$ large enough, yet with $\lambda_n = o(n)$.
Then we can increase both $\lambda_n$ and $n$ iteratively, with $n$ growing much faster than $\lambda_n$, without ever facing a gap larger than 2 except in case of failure.  
See table~\ref{table:k66vc0oh} for illustration: for maximum efficiency, $\lambda_n$ grows at a rate similar to that of the $n$-th record gap $g_n^*$ in prime numbers
(OEIS sequence \href{https://oeis.org/A005250}{A005250}), 
while $n$ follows the growth rate of the index of the record in question (OEIS sequence \href{https://oeis.org/A005669}{A005669}).
Under these conditions, the maximum gap in the first $n$ terms of $S_2$ is conjectured to always be 2. 
To avoid confusions, $S_2$ is denoted as $S_2(\lambda_n)$ to indicate that it is the
canonical form of $S_1$ at level $\lambda_n$.

\begin{table}[H]
\centering
\renewcommand{\arraystretch}{1.0}
\small
\begin{tabular}{c|cccccc} 
 \hline\\[-2ex]
  $n$ &  50 & 200 &   1000 & 8000 & 20000 & 40000     \\ 
\hline\\[-2ex]
$\lambda_n$ & 10 & 20 & 36 & 66 & 66 & 96\\
$g^*_n$  & 14 & 22 & 34 & 72 & 86 & 112\\[0.5ex]
 \hline 
 %  &  &  &        \\  [-2ex]
%5804 & $\num{46587}$ & $\num{66625}$ & $\num{74717}$ \\

\end{tabular}
%\caption{Backend table {\fontfamily{ptm}\selectfont hash stem} (extract)}
\caption{Level $\lambda_n$  guarantees max gap\\ $\leq 2$ in the first $n$ terms of $S_2(\lambda_n)$}
\label{table:k66vc0oh}
\end{table}

When building $S_2(\lambda_n)$, you iteratively create news rows in the triangle on top of each other, above level $\lambda_n$.  
In each of these rows, the first value must be initiated. I set the first value to 1 to keep success at all levels, assuming this was the case for 
the parent sequence $S_1$. But what would happen if you set one or two of these initial values to 3, 5, or any larger odd integer, thus causing failure
in $S_2(\lambda_n)$? In other words, getting some rows in the new triangle to start with a value $>1$. 
You can do it in the program by specifying a non-empty list of
failing levels (above $\lambda_n$), via the optional argument \texttt{fail\_level}
 in the \texttt{canonical} function located in the \texttt{Gilbreath\_lib} library. 
In my example, I use the prime number sequence with $n=1000$ for $S_1$ and $\lambda_n = 400$ for $S_2$, 
choosing 2 levels $\lambda_1{<} \lambda_2{<} \lambda_n$ starting with 3 rather than 1, with $\lambda_1=120$. The results are as follows: 
\vspace{1ex}
\begin{itemize}
\item Because $\lambda_n$ is much larger than needed given $n$, the successful sequence $S_2(\lambda_n)$ (when no failures are added) consists only of 2 at the beginning, followed
by either 1's or 3's as in table~\ref{table:gilb1bx}.  Actually, the first $\lambda_n$ values, besides the 2 at the start, are all 1. 
You can see it in figures~\ref{fig:aaj6b0gm8g} and~\ref{figgggb2cu8vaw}, where the green ``curve" represents the successful sequence on a log scale, taking on two distinct values only: 1 and 3, that is, 0 and $\log 3$ on a log scale. With the first 400 flat before oscillating randomly between the two.
\item The red curve represents $S_2(\lambda_n)$ when adding failure at 
$\lambda_1 = 120, \lambda_2 = 143$ (figure~\ref{fig:aaj6b0gm8g}) or $\lambda_1 = 120, \lambda_2 = 144$  (figure~\ref{figgggb2cu8vaw}).
Failures are magnified in $S_2(\lambda_n)$. When $\lambda_2=143$, the magnification is massive with values~$> e^{40}$ in the first 400 terms.   
\item In conclusion, adding two failures generates one of the following: (1) little blips detectable in the first 400 terms, (2) a strong explosion highly visible as in figure~\ref{figgggb2cu8vaw},
 or (3) a titanic blast that dwarfs everything else and can go on well beyond the first 400 terms
 as in figure~\ref{fig:aaj6b0gm8g}, possibly forever. Failures at levels such as case (3) can not exist in practice; it does not happen in natural
 sequences, even faulty ones. Yet it is one of the most common when reverse building a sequence by adding failures at random levels. 
In short, the failure levels can not be random. 
These insights are summarized in statement~\ref{pokomi}.  
\item What will cause a little blip, a strong explosion, or a titanic blast is impossible to predict. Failure levels well spaced out produce little blips. Other than that, the proximity
between failing levels, how close to the top level the failures are located, the parity of $\lambda_1$ and $\lambda_2$, and the left value  on a failing row (with ${>}3$ worse than 3) are the top factors influencing the outcome. Here I assume that $\lambda_n$ 
is large enough so that there are only 0 and 2 at level $\lambda_n$, besides the leftmost value.  
\end{itemize}
\vspace{2ex}
\begin{statement}\label{pokomi}
For a slow growth sequence such as the prime numbers or synthetic primes, if it has more than one failing point, there are strong constraints on the distance between
these failing points, and the parity of their indexes. Failures cannot happen at random locations, regardless of the gap sizes.  
Otherwise, the canonical sequence may explode in a way that is incompatible with acceptable bounds. 
\end{statement}\vspace{2ex}
\noindent I conclude this section with a probabilistic argument. 
I am interested in a probability central to the triangle:  the one that dictates how frequently the plus and minus signs occur in~(\ref{case2}). 
More specifically, let
\begin{equation}
\pi_k = P\big[\delta_{k}(q_n) < \delta_{k}(q_{n-1})\big].
\end{equation}
The \texttt{canonical} function in the code returns these probabilities  as the output vector \texttt{s1\_proba}, estimated on $S_1$.
For each $k$, it is computed by averaging on $n=1, 2$ and so on. 
 As usual, $k$ specifies the level,
with $k=1$ for the first order differences. 
Instead of always choosing the minus sign in~(\ref{case2}) to produce $S_2$, one may choose to
randomly alternate +/- at each level $k$ ($1\leq k < \lambda_n$) to slow the growth. And decide  how frequently to accept the minus sign depending on $\pi_k$.
My first attempt at that was inconclusive, but it revealed interesting facts:
\vspace{1ex}
\begin{itemize} 
\item Of course $\pi_1=1$ since the sequence $S_1$ is strictly increasing. Based on $n=20,000$, the next values for~the prime number sequence  are 
$\pi_2 = 0.69,  \pi_3 = 0.65,  \pi_4 = 0.49$. But for my synthetic primes, these values are $\pi_2 = 0.91,  \pi_3 = 0.51, \pi_4 = 0.40$.
The large differences confirm that the prime numbers do not abide well with the laws of randomness. 
\item Eventually, both for real and synthetic primes, $\pi_k$ decays to about 0.25 as $k$ increases. 
This may sound~like a mystery, but it has a simple explanation:  when $k$ is large and the rows contain just 0's and 2's, 
there are 4 possible combinations for $\big(\delta_k(q_n), \delta_k(q_{n-1})\big)$, each equally likely: (0, 0), (0, 2), (2, 0), and (2, 2).~Only
(2, 0)  satisfies~(\ref{case2}). This leads to my next and last conjecture.
\end{itemize}
\vspace{2ex}
\begin{conjecture}
For many sequences with no failing points including prime numbers and synthetic primes, as $k$ increases, the $k$-th row in the triangle consists almost exclusively of 0 and 2 in equal proportions. 
Also, rows $k$ and $k+1$ are uncorrelated. Of course, this is not true for integer polynomial sequences ending with 0 only, or the
geometric progression $1, 2, 4, 8, 16$ and so on, staying the same for all rows. 
\end{conjecture}

\section{Applications: cybersecurity, fraud detection and Fintech}  % quantifying chaos

I now discuss applications. In particular, I extend the concept of successful sequence to
continuous time series such as Brownian motions, and how success or failure translates into real-world data, and how to interpret it. I focus on modeling, data synthesis, random number generation and testing, liabilities in cryptography heavily relying on prime numbers, as well as error and 
pattern detection---especially hidden or hard to detect patterns. 

\subsection{Hidden patterns, error detection, checksum, and data quality audit}\label{erros}

Adding failure points at different levels when building the triangle backwards (from a lower level to the top) creates irregularities
in the reconstructed sequence. There is another way to look at it, with practical applications to detect errors in real-life data: compute the
differences in absolute value (first, second and higher orders). If some rows in the triangle start with a value other than 1, it is an
indicator of potential error or data quality issue. Especially if it impacts multiple rows and the starting value is higher than 3.

%-----------------------------vince/riemann2and3.mp4
\begin{figure}[H]
\begin{minipage}{.50\textwidth}
\centering
\captionsetup{justification=centering}
%\includegraphics[width=1.0\textwidth]{case1a.png}
%\vspace{0.25ex}
\includegraphics[height=6cm, width=8cm]{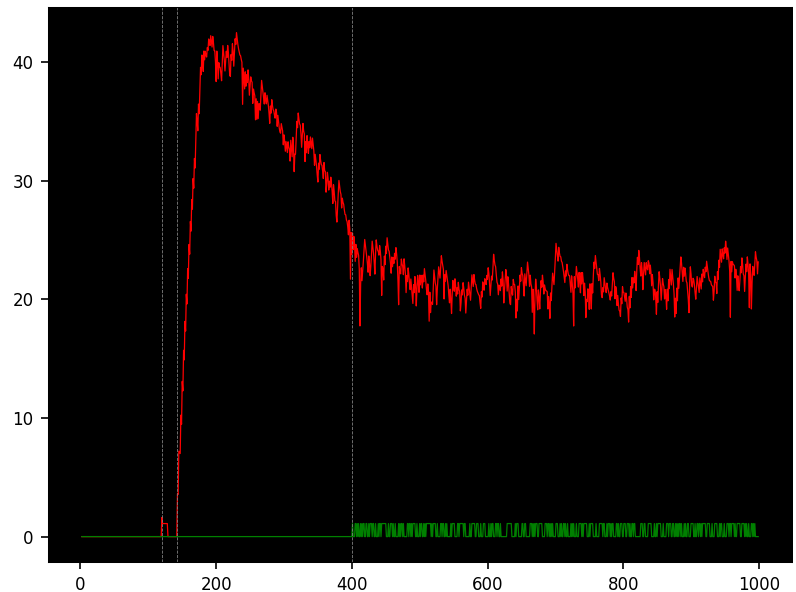}   % gb1sb.png / gbs.png
%\vspace{-1ex}
\caption{{\small Data issue detected via the canonical form: green signal when no problem is found, red otherwise.}}
\label{fig:aaj6b0gm8g}
\end{minipage}%\quad\quad%
\begin{minipage}{.50\textwidth}
\centering
\captionsetup{justification=centering}
%\vspace{-0.5ex}
%\includegraphics[width=1.0\textwidth]{case1b.png}
\includegraphics[height=6cm, width=8cm]{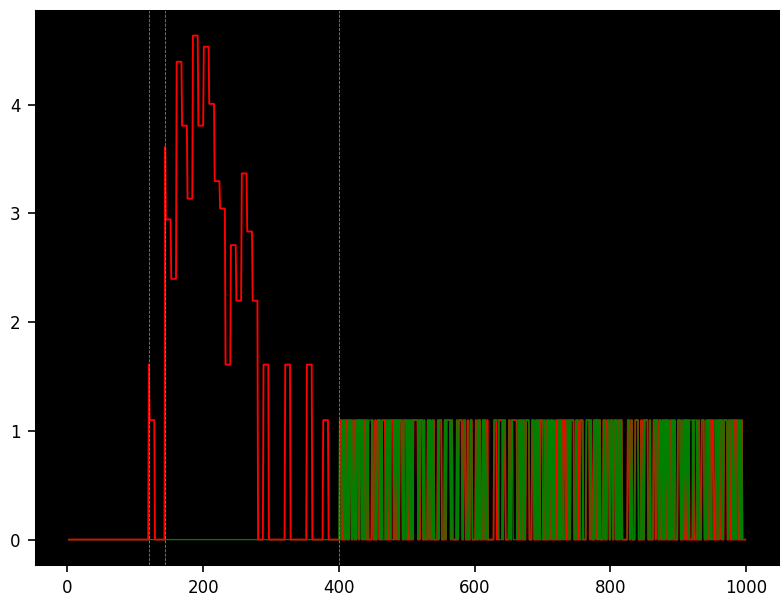}  % gb2b.png / gb2s.pnf
%\captionsetup{justification=centering}
%\vspace{0.1ex}
\caption{{\small Same as figure~\ref{fig:aaj6b0gm8g} but with less severe problem. The green signal takes on 2 values only, and starts flat.}}
\label{figgggb2cu8vaw}
\end{minipage}
\end{figure}
%imgpy9979_2and3.PNG
%-------------------------

\noindent Reconstructing the sequence backwards to find its canonical form, magnifies the errors
as seen in figures~\ref{fig:aaj6b0gm8g}~and~\ref{figgggb2cu8vaw}. It offers another way to reveal hidden patterns, errors and other types of issues in your dataset. 
In the~figures in question, the curves represent the reconstructed sequence in its canonical form, denoted as $S_2$ in section~\ref{reduxor}. 
The red curves correspond to corrupted data, that is, a row in the triangle not starting with 1. The left
element in each row plays the role of a \textcolor{index}{checksum} or error flag. When there is no issue (all rows start with 1), you get the green signal instead. 
This method can be combined with the standard technique for error detection, explained in~\cite{bubu66}.

\subsection{Fraud detection and transaction score synthesis}

There are many similarities between the admissible, valid sequences in this article, and credit card history~and transaction scoring, if you replace failure/success by fraud/no fraud or 
accepted/declined. The sequence corridor in section~\ref{effco} is the equivalent of the credit card limit over time based on customer history, slowly increasing, or much faster or even decreasing or blocked based on circumstances, while $q_n$ plays the role of the transaction amount. Each sequence represents a customer. 

The proportion of failures is low, comparable to that of fraudulent transactions. Like credit card issuers, I had to generate synthetic data to understand the 
various mechanisms that cause them. Analogy to rule-based systems is provided by the forbidden prime constellations. And finally, $q^+_n$ in 
section~\ref{pordel} is the critical threshold (based on the whole history) that tells you if the new value $q_{n+1}$ (the new credit card transaction) is flagged as fraudulent or not, or in my case, as success or failure.  A too high value is a typical cause of failure (or declined transaction), but it depends on history with some sequences (credit card users) allowed much
larger values. The subtle cause of failure in non-trivial cases, is due to complex patterns not involving large values. 

To conclude, the large sequence dataset \texttt{gil\_output.txt} and its many features, available upon request, is a great sandbox to test fraud detection and transaction
scoring algorithms. Perhaps the next step, both in the credit card industry and in my sequence system, is to generate all the possible patterns (trillions of them), then cluster
 and categorize them to finally attach a label to each of them: fraud or no fraud. This would lead to more efficient detection, especially in real time.
I started to gather all the possible patterns in my system, see section~\ref{beurred}, but I have yet to cluster them. Labeling as success/failure is trivial in my case, as that flag
is in the dataset by construction and readily available for any future ``transaction".  This last part is what makes my data particularly attractive in the context of testing fraud detection systems.

\subsection{Cellular automata random generator and new tests of randomness}\label{tnuuy7}

As noted in the last paragraphs in section~\ref{reduxor}, when you start with a large sequence and build the triangle from top to bottom, after a number of levels, you end up with 
uncorrelated rows consisting of 0 and 2 only in about the same proportions. This is the case when you start with the prime number sequence or synthetic primes,~and indeed with most non-failing sequences. The mechanism  behind the scenes is well understood: each element in the triangle is the difference in absolute value of the two elements above it. This operation is responsible for~the scrambling. Also, it prevents you from retrieving the elements in a row if you only know the elements one level below or further down. Exceptions include biased sequences such as prime numbers
where reverse engineering is somewhat less difficult due to forbidden patterns.

Thus, you can use the bottom rows of a triangle as a random number generator (PRNG). The code below does that efficiently, starting with a 
sequence consisting of $n$ bits, and then computing the successive  differences in absolute value. To keep the same number of terms in each row, I add one bit
 to the right at each level. The bit in question comes from a precomputed bitstream, here the array \texttt{last\_bit}.

\begin{quote}
\begin{lstlisting}[frame=none, escapechar=@,basicstyle=\ttfamily\footnotesize]
import numpy as np
from PIL import Image

np.random.seed(40)
n = 1000
n_levels = 1000 
sequence = np.random.binomial(n=1, p=0.5, size=n)
last_bit = np.random.binomial(n=1, p=0.5, size=n_levels)

bitmap = []

for k in range(n_levels):
    new_sequence =  np.abs(sequence[:n-1] - sequence[-(n-1):])
    new_sequence = np.append(new_sequence, last_bit[k]) ##
    covar = np.corrcoef(sequence[:n-1], new_sequence[:n-1])
    count_0 = np.count_nonzero(sequence == 0)
    count_1 = np.count_nonzero(sequence == 1)
    print(k, sequence[:20], count_0, count_1, covar[0, 1])
    sequence = new_sequence
    bitmap.append(sequence)

scaled_array = np.array(bitmap, dtype=np.uint8) * 255
img = Image.fromarray(scaled_array, mode='L').convert('1')
img.save('bitmap_output.png')
\end{lstlisting}
\end{quote}
In chapter 8 in my new book~\cite{0and1new}, I discuss NPG, one of my  cryptographically secure PRNGs, faster, and more random than the most 
recent additions to the Numpy library:  \textcolor{index}{PCG64} and \textcolor{index}{PCG64DXSM}. See my article on the topic, \href{https://mltblog.com/npg}{here}, where  I also  introduce a new test
of randomness: predicting the next bits using \textcolor{index}{large language models}, similar to predicting the next tokens, and based on a \textcolor{index}{deep neural network} (DNN). Bitstreams  for which the correct prediction rate is $> 52\%$ or $< 48\%$ are flagged as 
unsecure. For a fast alternative to DNNs, where the global optimum of the loss function is obtained in one shot with a simple closed-form formula, thus without training, 
see my blog post \href{https://bondingai.io/blog/96-correct-next-token-prediction-with-no-dnn-no-training-auto-distilled-model}{here}, and my article on the topic~\cite{arvivi}.

\begin{figure}[H]
%\begin{minipage}{.50\textwidth}
\centering
\captionsetup{justification=centering}
\includegraphics[height=16cm, width=16cm]{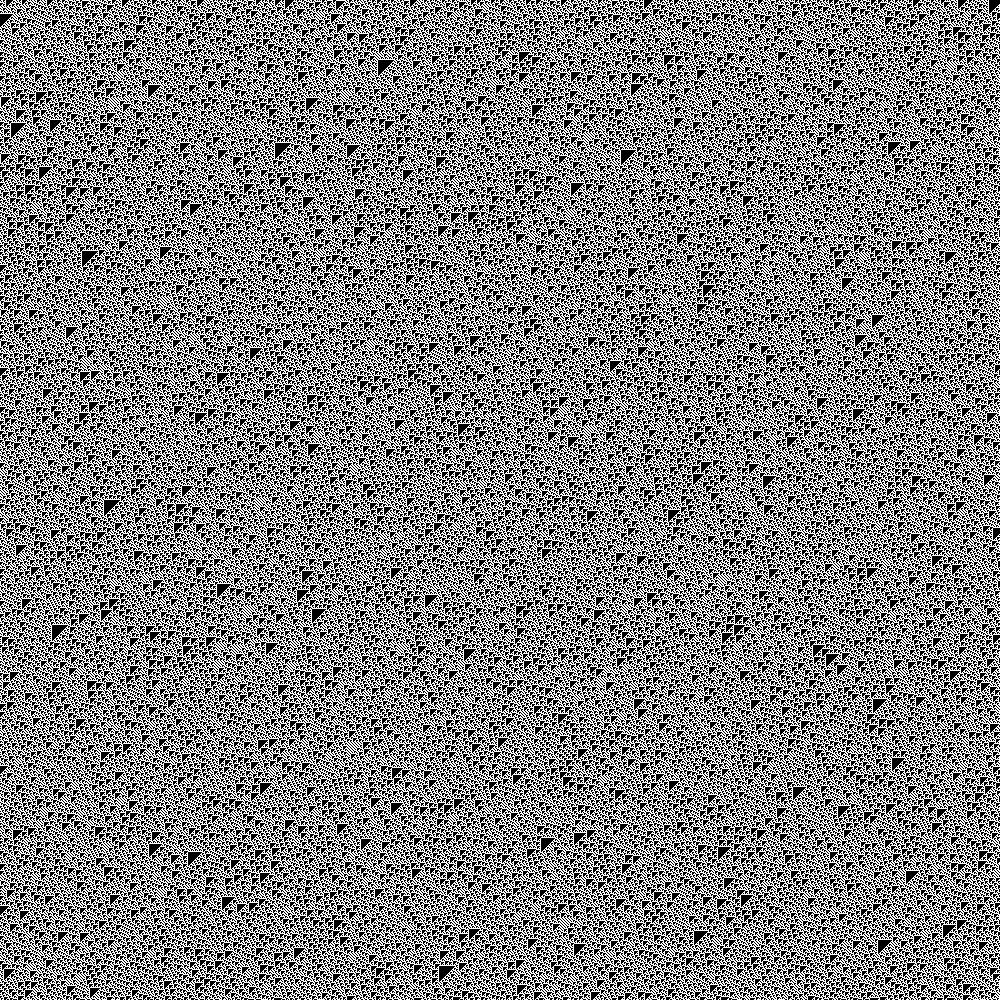} %%{gb1b.png}   % gb1sb.png / gbs.png
\caption{{\small Random numbers  generated with the code in section~\ref{tnuuy7}; each line is a level ($n=1000$, 1000 levels)}}
\label{fig:aajttf4eer55}
\end{figure}
%imgpy9979_2and3.PNG
%-------------------------

\noindent With binary numbers as in the above code, $|a-b| = \text{XOR}(a, b)$. The resulting triangle is identical to the~output of the
\textcolor{index}{rule 90 cellular automaton} [\href{https://en.wikipedia.org/wiki/Rule_90}{Wiki}], extensively studied~\cite{cell23} with numerous applications. 
Also, the associated
random number generator is fast, see \textcolor{index}{xoroshiro}~\cite{xoroshi}. By contrast to xoroshiro, my
system leads to a non-periodic PRNG thanks to the addition of the rightmost bit at each level. However, the raw version lacks some randomness as seen
in figure~\ref{fig:aajttf4eer55}: the patterns are similar to the \textcolor{index}{Sierpiński triangle}
and identical to those in the \textcolor{index}{Thomas matrix} introduced by James Thomas in 2004, see \href{https://observablehq.com/@nxrix/gilbreath-conjecture}{here}.
Zoom in on the picture to see more granular details. For other chaotic behavior
 linked to iterated differences in absolute value, see section 7.2.2 in~\cite{0and1new}.

%The quality of the random bitsream generated  // zoom in
 %is very high.  xxxxxxxxx   
%Thomas matrix was introduced by James Thomas in 2004 // image below is 1000 x 1000
% https://www.primepuzzles.net/puzzles/puzz_274.htm
% https://observablehq.com/@nxrix/gilbreath-conjecture

Finally, the material in this article led to yet another test, not found in most standard libraries: checking the indexes of the successive records in run lengths. 
Having more than $\log_2 n$ consecutive 0 in the first $n$ bits at any level in the triangle, violates the laws of randomness and increases the 
 risk of sequence failure. 
When running the above code snippet, the initial sequence fails the \textcolor{index}{record run test}. 
The \textcolor{index}{Mersenne Twister} in the Numpy \texttt{random} module causes the issue: there is a run of 8 ones in the first 14 bits; the maximum allowed this early is 4.
In practice, the initial sequence is the seed of my new PRNG. A good choice is a combination 
of interlaced binary digits coming from the square roots of various integers; the same applies to \texttt{last\_bit}.

\subsection{Unusual patterns in time series: modeling, detection and synthesis}\label{app45}

Gilbreath's conjecture and its traditional generalizations~\cite{zachch23} apply to discrete integer sequences with small~growth, strictly increasing 
and random in some ways. Here I show how it also works for discrete or time-continuous time series, smooth (no randomness) or highly chaotic,
with massive or no growth, or even decreasing or oscillating. I also discuss industrial applications, for instance in financial time series modeling.
Case studies are featured in figures~\ref{fig:aajj5gb17} and~\ref{figgggb2csku6}, where the red and green curves illustrate (resp.) sequence failure and success.  
In particular, figure~\ref{fig:aajj5gb17} shows time-continuous \textcolor{index}{Brownian motions} generated via a 
discrete random walk (the underlying sequence) with jumps up and down following a Poisson distribution of parameter $\lambda$. 
Success or failure is dictated by the value of $\lambda$, with higher $\lambda$ more prone to failures.

%-----------------------------vince/riemann2and3.mp4
\begin{figure}[H]
\begin{minipage}{.50\textwidth}
\centering
\captionsetup{justification=centering}
%\includegraphics[width=1.0\textwidth]{case1a.png}
%\vspace{0.25ex}
\includegraphics[height=6cm, width=8cm]{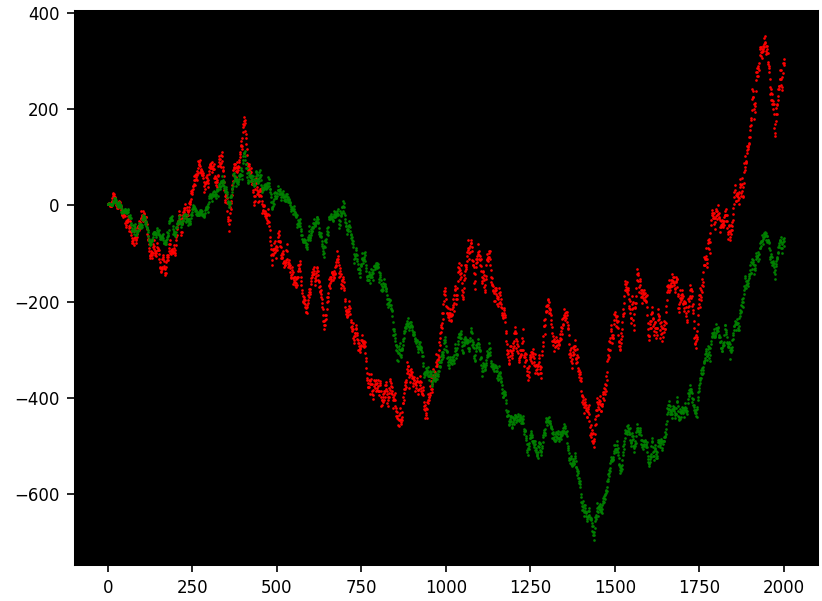} %%{gb1b.png}   % gb1sb.png / gbs.png
%\vspace{-1ex}
\caption{{\small Brownian motions moving in the same direction at each step;
red fails the test, but the green one passes it}}
\label{fig:aajj5gb17}
\end{minipage}%\quad\quad%
\begin{minipage}{.50\textwidth}
\centering
\captionsetup{justification=centering}
%\vspace{-0.5ex}
%\includegraphics[width=1.0\textwidth]{case1b.png}
\includegraphics[height=6cm, width=8cm]{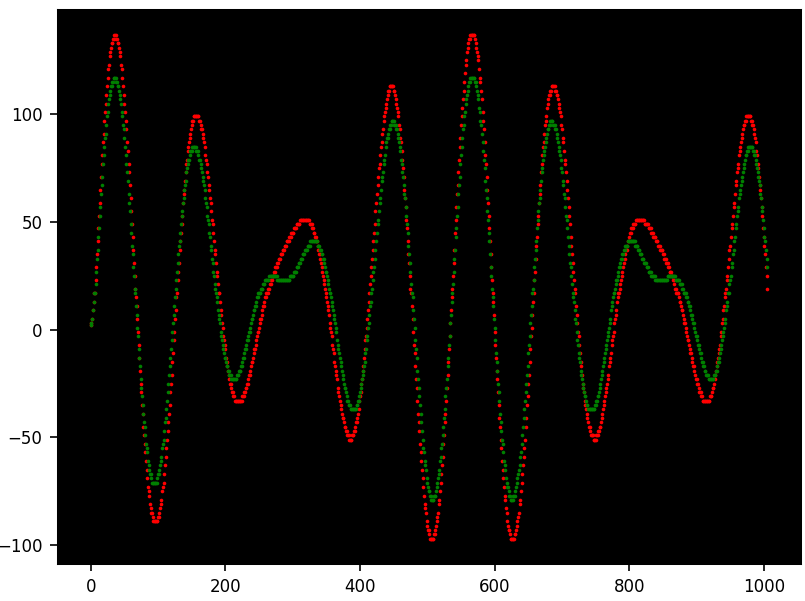} %%{gb2c.png}  % gb2b.png / gb2s.pnf
%\captionsetup{justification=centering}
%\vspace{0.1ex}
\caption{{\small Two smooth curves with nearly identical parameters;
red fails the test, green one passes it}}
\label{figgggb2csku6}
\end{minipage}
\end{figure}
%imgpy9979_2and3.PNG
%-------------------------

\noindent At a macro level, what makes the green and red curve different, resulting in a different success status, is visible~to the naked eye: it is the amount
of volatility. But in many cases, the chaos is well hidden and not~visible. A~good example is the sequence in table~\ref{table:gilb4}. 
Also, the two examples in figures~\ref{fig:aajj5gb17} and~\ref{figgggb2csku6} are borderline: green is close to failure~ and red close to success. 
What if we could have a metric that automatically quantify chaos and categorize time series at any given time based on invisible yet disruptive patterns, to
make appropriate decisions such as buying or selling assets? The \textcolor{index}{Lyapunov exponent} is one of several metrics but it does not reveal hidden patterns.
More are discussed in my book on chaotic dynamical systems~\cite{vgchaos}. A metric that takes into account more subtle patterns would be useful. 
And the material in this paper offers several alternatives for \textcolor{index}{quantized data}:
\vspace{1ex}
\begin{itemize}
\item How many levels it takes in the triangle to reach the first row consisting only of 0 and 2 (besides the odd value at the start of each row)?
The fewer you need, the less chaotic the sequence. Here chaos is defined not in a classical sense, but in connection to the triangle behavior. 
\item The number of failure points: how many times a row in the triangle starts with 3 instead of 1? Or worse, with 5 or above? How spaced out are these failure points?
It quantifies the level of chaos and the intensity of each failure point. 
\item In case of success, how close are we to a failure at any given time? This is linked to the distances
$|q_{n+1} - q^+_n|$ and $|q_{n+1} - q^-_n|$, and shown in figures~\ref{fig:aajjjjh6y}--\ref{fig:aajjjd3uyvffo}. 
 In case of failure, how close are we to success?
\item Is the canonical form of the sequence error-free, that is, do you get the green signal
in figure~\ref{figgggb2cu8vaw}, or the red one? And in case of errors, do we observe little blips, or 
a big explosion as in figure~\ref{figgggb2cu8vaw}? The case
in figure~\ref{fig:aaj6b0gm8g} (titanic, long lasting or permanent  blast) cannot happen in real-life data. 
\item For the prime number sequence and synthetic primes (both well behaved) at some point in the triangle, we end 
up with 0 and 2 only, evenly and uniformly distributed within each row. In addition, rows at this stage are independent from each other
(uncorrelated). Is this the case for your data? Even though real and fake primes are
well behaved,  the real ones exhibit little spikes
in figure~\ref{fig:aajjjjh6y}; fake ones in figure~\ref{fig:aajjjd3uyvo} don't. Those little spikes, if found in your data in a well behaved sequence, indicate some small irregularities and pinpoint their locations. The spikes point outward, which is good. If they were pointing inward, it~would increase the risk of failure.  
\item Are there missing patterns in your sequence, in the first order differences? That is, expected patterns that are absent, such as the
forbidden constellations in the prime number sequence. Missing patterns in the second or higher order differences is an indicator of a bigger problem. 
\end{itemize}
\vspace{1ex}
Finally, a time series can change from chaotic to smooth, or from fast to slow growth or even oscillating, even multiples times each with a different duration, yet succeed.
Another series, apparently better behaved when seen with the naked eye, may fail. Eventually, what my system flags as failure or success in a time series with structural changes  is not
how volatile it suddenly becomes, but whether the transition between the different phases is well controlled or not.
My set of metrics should be tested and fine-tuned on real data to see if it can catch a shift before other indicators notice it when it is too late for actions. Be it financial or Earthquake data.

%target VCs/angels in McKinney
%share my books on BDAI to get leads ***
%make the book with ch8 public

\subsection{Security issues in ciphers relying on prime numbers}

Prime numbers have patterns that significantly increase their odds to satisfy Gilbreath's conjecture, compared to other sequences with similar growth, but more random.
For instance, you cannot find 2, 2, let alone 2, 2, 2, 2, in successive prime gaps except at the very beginning. Yet 2, 2 is the most basic pair, one of the most frequent in random sequences, but it is banned in prime numbers, and referred to a \textcolor{index}{forbidden prime constellation}. At the same time, 2, 2, 2, 2 followed or preceded by a large gap, and then by another (say) 2, 2, 2 is among the most dangerous configurations: it can kill success if happening too soon in the first order differences. Another forbidden prime gap pattern is 2, 4, 2, 4, 2, and again, 
linked to higher failure risk. There are many others,~all with small numbers (thus abundant in other sequences but absent in primes), giving prime numbers some extra immunity against failure.

But what makes prime numbers strong in Gilbreath's context is exactly what makes them weak in cybersurity. The prime gaps are anything but random. 
And while there might be no forbidden patterns in higher order differences (the rows below prime gaps in the triangle), the pattern distribution remains
 biased downwards, with probabilities attached to patterns not matching those found in random data or synthetic primes. Put it differently,
given any row in the triangle, it is possible to tell whether or not it comes from prime numbers by looking at 
the probabilistic structure attached to the row in question. In short, \textcolor{index}{reverse engineering} is feasible if the original sequence consists of prime numbers, despite the 
obfuscation created by successive differentiating in absolute values. 
Ciphers or crypto-secure random numbers derived mostly from prime numbers but lacking strong obfuscation---and there are still many around---are subject to the same problem
caused by \textcolor{index}{congruential constraints}. 

\section{Conclusions}

Gilbreath's conjecture, first stated by François Proth in 1878 and unproved to this day, features a 
very simple pattern among prime numbers. It is well-known in mathematical circles as a very challenging 
unsolved mystery. In this paper, I discuss a path towards a resolution, as well as applications pertaining to Fintech,
chaotic systems, cybersecurity, high performance computing, data synthesis, random number generation, fraud detection and more, well beyond the original scope.

From a research perspective, Gilbreath's conjecture has little to do with prime numbers: a large class of sequences, including many with either slow or fast growth, with either chaotic or
smooth behavior, satisfy it. Yet the prime number sequence has something unique that makes it the ideal candidate. It's not its random character:
deterministic patterns such as forbidden prime number constellations, actually contribute to its success while the lack of randomness is a cybersecurity liability. Also, I proved that the Sieve of Eratosthenes leads to sequences that satisfy the conjecture.
It also works if you sieve in any order. Sieving in reverse order allows you to assess the long-range impact of small primes. This technique may be used
 to discover new results, perhaps about twin primes. For instance, the sum of the inverse of twin primes
converges; the proof is based on standard sieving.  

This paper turns what was initially recreational math into mainstream computer science, dynamical systems,  and number theory, 
with significant research and application potential across multiple disciplines, and many new opportunities 
yet to be explored.  Along the way, it leads to interesting integer sequences , some
listed in the OEIS encyclopedia, and some waiting to be added such as magic primes where success is always guaranteed.  Finally, it features high performance computing at its best,
with intense simulations to help progress and in some instances, even to prove statements.  Equivalent sequences and sequence reduction eliminate large gaps, 
facilitate theoretical research, and leads to applications such as hidden patterns discovery, error detection, and data quality audit. 
The connection to cellular automata leads to a very efficient random generator.

\bibliographystyle{plain} % We choose the "plain" reference style
\bibliography{refstats} % Entries are in the refs.bib file in same directory as the tex file

\appendix
\section{Appendix: Python code}\label{pythor}

The main program \texttt{gilbreath3.py} is listed in section~\ref{mcode}. It imports  \texttt{gilbreath\_lib.py}, a home-made, internal library. The latter is listed
in section~\ref{lcode}. Also, in section~\ref{mcode}, I included a summary, linking the various parts and functions in the Python code,
 to the sections and figures in this article.

\subsection{Main program}\label{mcode}

The code is also on GitHub, \href{https://github.com/VincentGranville/Experimental-Math-Number-Theory/blob/main/Source-Code/gilbreath3.py}{here}.
Below is a summary linking the different parts to the material in this article; \texttt{gil} stands for the \texttt{gilbreath\_lib.py} library. 
\vspace{1ex}
\begin{itemize}
\item {\bf Data synthesis} (synthetic primes): Function: \texttt{gil.generate\_random\_sequence}. Lines \textcolor{gray}{36--104} in~the code. Section~\ref{gsimuls} in the paper.
\item {\bf Forbidden prime constellations}: Function: \texttt{gil.banned\_found\_in\_seq}. Lines \textcolor{gray}{98} in the code.~Section \ref{forbidden} in the paper.
\item {\bf Sequence reduction} (canonical): Functions: \texttt{gil.canonical}, \texttt{gil.generate\_random\_sequence}. Lines \textcolor{gray}{105--141} in the code. Sections~\ref{reduxor} and~\ref{erros} in the paper.
\item {\bf Sequence augmentation} ($q^+_n, q^-_n$):  Functions: 
\texttt{gil.collect\_statistics}, \texttt{gil.max\_gap\_allowed}, \\
\texttt{gil.augment}, \texttt{gil.generate\_random\_sequence}, \texttt{gil.test\_conjecture}, \texttt{gil.arr\_plot}.
Lines \textcolor{gray}{142--255} in the code. Section~\ref{pordel} in the paper.
\item {\bf Sieving and reverse sieving}: Lines \textcolor{gray}{256--322} in the code. Section~\ref{siebel} in the paper.
\item {\bf Continuous time series}: Lines \textcolor{gray}{338--408} in the code. Section~\ref{app45} in the paper.
\item {\bf Full list of short sequences}: Functions: \texttt{gil.admissibility}, \texttt{gil.get\_all\_sequences}. Lines \textcolor{gray}{409--457} in the code. 
Section~\ref{beurred} in the paper.

\end{itemize}

%################### https://www.johndcook.com/blog/2009/09/09/gilbreath-conjecture/

%np.random.seed(87)  ############ my NPG allows for multiple bit streams built in parallel from different seeds
%#################### like    xxx.seed(56, 6)   for bit stream 6 ########## update doc of NPG ******

%contact john cook, the guy who reviewed my book, and all authors mentioned 

%prime gap records: https://t5k.org/notes/GapsTable.html
%https://arxiv.org/pdf/1802.07609
%https://www.perplexity.ai/search/distribution-of-prime-gap-reco-k6H2lNmJSsaLLis5E9qchw
%https://archive.lib.msu.edu/crcmath/math/math/p/p583.htm
%https://sweet.ua.pt/tos/gaps.html ****
%https://www.google.com/search?q=maximum+asymptotic+value+for+the+gap+between+two+primes
%https://www.ford126.web.illinois.edu/wwwpapers/primegaps.pdf

\vspace{1ex}
\noindent The code is based on Python 3.13.3. Besides my library \texttt{gilbreath\_lib.py} listed in section~\ref{lcode}, it requires Numpy 2.2.5, Matplotlib  3.10.1, PrimePy, and Time. It also produces the output file \texttt{gil\_output.txt}, which is not discussed in the article but linked to the computation of $q^+_n$. Throughout the code,
the main sequence is denoted as \texttt{arr\_a},  its last element $q_n$ as the variable \texttt{a}, and the first order differences as
\texttt{first\_diff}. 
\vspace{1ex}

\begin{lstlisting}[numbers=left,escapechar=@,basicstyle=\ttfamily\footnotesize]
import importlib
import gilbreath_lib as gil
importlib.reload(gil)

import time
import numpy as np
import matplotlib.pyplot as plt
import matplotlib as mpl

mpl.rcParams['axes.linewidth'] = 0.5
plt.rcParams['xtick.labelsize'] = 8
plt.rcParams['ytick.labelsize'] = 8
plt.rcParams['legend.fontsize'] = 'x-small'


#--- 1. Showing how central function gil.check works

nprimes = 30 
from primePy import primes
prime_list = primes.first(nprimes)
arr_a = prime_list[0:nprimes]
# arr_a = [2, 3, 5, 11, 13, 15, 31, 33, 35, 39, 51 ]
print("--- Part 1: gil.check\n")

flag, first_amax_idx, hash_stats = gil.check(arr_a, show_triangle = 'Full', mode = 'Full', delta = 'standard') 

right_diagonal = hash_stats['right_diagonal']
left_diagonal = hash_stats['left_diagonal']
failing_level = hash_stats['failing_level']
print("Right diagonal:", right_diagonal)
print("Left diagonal:", left_diagonal)
print("Failing_level:", failing_level)
print()


#--- 2. Detect failed sequences by simulation [Poisson increments]

parameters = { 
               'prime': [1.0, 1.3, 1.50, 5],
               'power': [0.9, 1.1, 1.25, 5, 1.50],
             }

model = 'prime'

n = 60  # number of values to add after arr_init
nsamples = 400 
llambda = 1.5 
rejection_sampling = True   # set to True to guarantee sequences are in the corridor
arr_init = [2, 3, 5, 7, 11, 13, 17, 19, 23, 29, 31, 37, 41, 43]  
rng = np.random.default_rng(seed=42) 

offset = len(arr_init)

from primePy import primes
prime_list = primes.first(n + offset)
p_n = prime_list[-1]
success = 0
in_corridor = 0
sum = 0
hash_fail = {}
print("--- Part 2: Simulations\n")

for sample in range(nsamples):

    arr_a = gil.generate_random_sequence(n, arr_init, llambda, parameters, model, rng)
    sum += arr_a[-1]
    flag, first_amax_idx, hash_stats = gil.check(arr_a, show_triangle = 'None')
    failing_level = hash_stats['failing_level']
    first_diff = hash_stats['first_diff']  
    left_diaginal = hash_stats['left_diagonal']
    valid, bounded, smooth, fvals, sequence = gil.admissibility(first_diff, parameters, model)  # fvals  empty if in corridor

    if flag == 'Fail' and (bounded or smooth):

        snapshot = arr_a[offset-1 : failing_level]
        left_diagonal = hash_stats['left_diagonal']
        sigma = 1 + left_diagonal[-1]
        failing_diff = tuple(first_diff[0:failing_level])
        gamma = np.max(failing_diff)
        hash_fail[failing_diff] = (failing_level, sigma, gamma, failing_diff[-2], failing_diff[-1])
        print("Failed:", sample, failing_level, sigma, bounded, smooth, arr_a[-1], 
                         p_n, failing_level, first_diff[0:failing_level]) 

    if  bounded and smooth:
        in_corridor += 1
    if flag == 'Success':
        success += 1
    if sample % 500 == 0:
        print("Sample", sample, success-1)

avg_last_value = sum/nsamples

for first_diff in hash_fail:

    value = gil.make_readable(hash_fail[first_diff])
    first_diff = gil.make_readable(first_diff)
    first_diff_seq = first_diff[offset-6:len(first_diff)]
    banned_found, banned_pattern = gil.banned_found_in_seq(first_diff_seq)
    banned_pattern = gil.make_readable(banned_pattern)

failed_deduped = len(hash_fail)
print("Simulations:", success, "/", nsamples, avg_last_value, p_n, in_corridor, failed_deduped)
print()

#--- 3. Canonical form of sequence 

n_levels = 400 
fail_levels = { 120:3, 281:3 }  # err1.png 
fail_levels = { 120:3, 143:3 }  # err2.png
fail_levels = { 120:3, 144:3 }  # err3.png

N = 1000 
primes_ = 'real'  # options: 'real' or 'fake'

if primes_ == 'real':
    prime_list = primes.first(N)
else:
    arr_init = [2, 3, 5, 7, 11, 13, 17, 19, 23, 29]  
    llambda = 2.5 
    prime_list = gil.generate_random_sequence(N, arr_init, llambda, parameters, model, rng)

s1 = np.copy(prime_list)
valid = False
print("--- Part 3: Canonical form\n")

s2_failed, s2_failed_gaps, s1_bottom, s1_max_gap, s1_proba = gil.canonical(s1, n_levels, fail_level = fail_levels, valid = valid)
s2_succes, s2_succes_gaps, s1_bottom, s1_max_gap, s1_proba = gil.canonical(s1, n_levels, valid = valid)

plt.gca().set_facecolor('black')
arr_k = np.arange(2, len(s1))
plt.plot(arr_k, np.log(s2_failed[2:]), c='red', lw=0.9)
plt.plot(arr_k, np.log(s2_succes[2:]), c='green', lw=0.6) 
for k in fail_levels:
    plt.axvline(x=k, color='gray', linestyle='--', linewidth=0.4)
plt.axvline(x=n_levels, color='gray', linestyle='--', linewidth=0.4)
plt.show()

print("Summary:", N, n_levels, np.max(s2_succes_gaps), np.max(s1_bottom), np.max(s2_succes), np.max(s1), s1_max_gap)
print("Probas:", s1_proba[0:5])

 
#--- 4. Find bounds for next prime, to keep the sequence successful moving forward

sequence_type = 'prime_gaps' # options: 'primes', 'log_gaps', 'power_gaps'
prime_type = 'real'  # 'real', 'simulated'
nprimes = 5000 
start = 2 # must be 2 or higher

test_mode = False  # very slow if True, will test conjecture
if test_mode:
    conjecture = 'satisfied' # satisfied until proven wrong
else:
    conjecture = 'untested'

if prime_type == 'real':
    from primePy import primes
    prime_list = primes.first(nprimes)
else:
    arr_init = [2, 3, 5, 7, 11, 13, 17, 19, 23, 29, 31, 37, 41, 43] 
    llambda = 2.5 # 5
    prime_list = gil.generate_random_sequence(nprimes, arr_init, llambda, parameters, model, rng)
    flag, acheck, hash_stats = gil.check(prime_list, show_triangle = 'None', mode = 'Fast')
    print("check status of generated sequence:", flag)

arr_k = []
arr_min = []
arr_max = []
arr_a = []
arr_old_a = []
resets_list = ()
a = prime_list[start-1]
print("\n--- Part 4: find q_plus, q_minus\n")

OUT = open("gil_output.txt", "wt") 
print("idx", "reset", "rvbool", "k", "a", "a-old_a", "max_a-a-k", "end", "counts", 
      "dmax", "dmax2", "dmax2_idx", "extract", file = OUT, sep="\t")

# mode = 'Full' to not stop as soon as success is detected and get full diagonal
flag, acheck, hash_stats = gil.check(prime_list[0:start], mode = 'Full')
right_diagonal = hash_stats['right_diagonal'] 
right_diagonal = gil.make_readable(right_diagonal) 
hash_output = gil.collect_statistics(right_diagonal)

start_time = time.perf_counter()

for k in range(start, nprimes): 

    # given previous primes, find possible range for next one to guarantee success
    # augmenting sequence with previous prime as starting point, seems to always succeeds 

    old_a = a  
    old_hash_output = hash_output
    
    max_a, trials1 = gil.max_gap_allowed(right_diagonal, old_a) 
    min_a, trials2 = gil.min_gap_allowed(right_diagonal, old_a)

    if sequence_type == 'prime_gaps': 
            a = prime_list[k]
    else:
        if sequence_type == 'log_gaps':
            # mimic prime growth 
            gap = int(np.log(max_a)) # max_a can be as large as 2a
        elif sequence_type == 'power_gaps':
            gap = int(max_a**0.35) # exponent must be < 1
        gap = 2 + (gap - gap % 2)
        a += gap

    arr_k.append(k)
    arr_min.append(min_a)
    arr_max.append(max_a)
    arr_a.append(a)
    arr_old_a.append(old_a)
        
    if test_mode:
        # check if all 'a' are in [old_a, max_a] succeed, and max_a + 2 fails
        satisfied, flag = gil.test_conjecture(old_a, max_a, right_diagonal)
        if not satisfied:
            conjecture = 'not satisfied'

    old_idx = old_hash_output['before_02_index']
    old_nu_2 = old_hash_output['nu_2']
    success, right_diagonal = gil.augment(right_diagonal, a)
    right_diagonal = gil.make_readable(right_diagonal) 
    hash_output = gil.collect_statistics(right_diagonal)

    reset     = hash_output['reset'] 
    counts    = hash_output['counts'] 
    dmax2     = hash_output['max_value'] 
    dmax2_idx = hash_output['max_index'] 
    end       = hash_output['before_02_value'] 
    idx       = hash_output['before_02_index'] 
    dmax      = hash_output['max_tail_value'] 
    # ri        = hash_output['reversal_idx'] 
    # rval      = hash_output['reversal_value']
    rvbool    = hash_output['reversal']

    if reset:
        resets_list = (*resets_list, k) 
    if k % 1000 == 0:
        print("progress:", k, "/", nprimes)

    extract = right_diagonal[1:20] 
    print(old_idx, idx, old_nu_2, reset, rvbool, k, a, a-old_a, max_a-a-k, end, counts,
           dmax, dmax2, dmax2_idx, extract, file = OUT, sep="\t")
    
end_time = time.perf_counter()
elapsed_ms = end_time - start_time
print(f"Elapsed time: {elapsed_ms:.2f} ms")
OUT.close()

gil.arr_plots(arr_k, arr_min, arr_max, arr_a) # , resets_list) 

# future project:  analyse p_n - 2 p_{n-1} + p_{n-2}


#--- 5. Produce trimmed sequences with Sieve of Eratosthenes

all_primes = False
max_period = 20000
if all_primes:
    from primePy import primes
    # use the first 100k primes
    moduli = primes.first(100000)
else:
    moduli = (2, 3, 5, 7, 11, 13, 17, 19)

period = 1
for p in moduli:
    period *= p - 1
print("period = ", period)
arr_a = [2, 3]
start = 4 
end =  min(period + len(moduli), max_period)

count = 0
k = start

while count < min(end, max_period): 
   keep = True
   index = 0
   while keep and index < len(moduli):
       mod = moduli[index]
       if k % mod == 0:
           keep = False
       index += 1
   if keep or k in moduli: 
       arr_a.append(k)
       count += 1
   k += 1
   if k % 1000000 == 0:
       print("/...", k, count)

print("\n--- Part 5: Sieve of Eratosthenes\n")
show_triangle = 'Compact'
arr_a_truncated = arr_a[0:max_period]
flag, acheck, hash_stats = gil.check(arr_a_truncated, show_triangle)

# gap combos 

print()
hash_gaps = {} 
nobs = len(arr_a)
for k in range(1,nobs):
    gap = arr_a[k] - arr_a[k-1]
    block = int(4*k/nobs)
    if (gap, block) in hash_gaps:
        hash_gaps[(gap, block)] += 1
    else:
        hash_gaps[(gap, block)] = 1

for gap in range(2, 16, 2):
    for block in range(4):
        key = (gap, block)
        if key in hash_gaps:
            count = hash_gaps[key]
        else:
            count = 0
        print("Gap combos:", key, count)
print("len arr_a", len(arr_a))
   

#--- 6. Impact of removing/adding one number

arr_list = [[2, 3, 5, 7, 11, 13],
            [2, 3, 5, 7, 13],
            [2, 3, 5, 7, 11, 13, 19],
            [2, 3, 5, 9, 13, 17, 25, 49, 81, 121, 169, 225, 313],
            [2, 3, 5, 9, 13, 17, 25, 49, 81, 121, 169, 225, 301, 455, 753, 1347, 2549, 2567, 2999, 3439],
            [2, 3, 5, 9, 13, 11, 15, 9, 7, 11, 21, 25, 27, 21, 13, 17]]  

print("\n--- Part 6: Removing/adding one term\n")

for arr_a in arr_list:
    flag, acheck, hash_stats = gil.check(arr_a, show_triangle)
    print("\n")


#--- 7. Continuous time series: Brownian motion

rng = np.random.default_rng(45)
np.random.seed(45)

# Generate a single Poisson random variable with lambda (lam) = 4
single_val = rng.poisson(lam=4)
arr_a = [2, 3]
arr_b = [2, 3]
a = arr_a[-1]
b = arr_a[-1]

for k in range(2000):
   sign = np.random.uniform(0, 1)
   if sign < 0.5:
       sign = -1
   else:
       sign = 1
       llambda_a = 0.8*np.log(1+k) 
       llambda_b = 0.5*np.log(1+k) 

   a = a + 2 * sign * rng.poisson(lam=llambda_a) 
   b = b + 2 * sign * rng.poisson(lam=llambda_b)
   arr_a.append(a)
   arr_b.append(b)

print("\n--- Part 7: Brownian motion\n")
show_triangle = 'Compact'
flag_a, acheck_a, hash_stats_a = gil.check(arr_a, show_triangle)
arr_stats_a = hash_stats_a['arr_amax']
flag_b, acheck_b, hash_stats_b = gil.check(arr_b, show_triangle)
arr_stats_b = hash_stats_b['arr_amax']

arr_k = np.arange(0, len(arr_a), 1)
plt.scatter(arr_k, arr_a, linewidth = 0.0, s = 1.6, c= 'red')  
plt.scatter(arr_k, arr_b, linewidth = 0.0, s = 1.6, c= 'green')

plt.gca().set_facecolor('black')
plt.show()
print(flag_a, len(arr_stats_a), flag_b, len(arr_stats_b)) 


#--- 8. Continuous time series: Smooth curve

arr_a = [2, 3, 5, 9, 13, 17]
arr_b = [2, 3, 5, 9, 13, 17]
a0 = arr_a[-1]
b0 = arr_a[-1]

for k in range(1000):
    a = a0 + 2*int(38*np.sin(k/21) + 23*np.sin(k/17)) 
    b = b0 + 2*int(28*np.sin(k/21) + 23*np.sin(k/17)) 
    arr_a.append(a)
    arr_b.append(b)

print("\n--- Part 8: Smooth curve\n")
show_triangle = 'Compact'
flag_a, acheck_a, hash_stats_a = gil.check(arr_a)  
arr_stats_a = hash_stats_a['arr_amax']
flag_b, acheck_b, hash_stats_b = gil.check(arr_b)
arr_stats_b = hash_stats_b['arr_amax']


arr_k = np.arange(0, len(arr_a), 1)
plt.scatter(arr_k, arr_a, linewidth = 0.0, s = 3.0, c= 'red') 
plt.scatter(arr_k, arr_b, linewidth = 0.0, s = 3.0, c= 'green')
plt.gca().set_facecolor('black')
plt.show()
print(flag_a, len(arr_stats_a), flag_b, len(arr_stats_b)) 


#--- 9. Build all short sequences with n <= nlevels, sigma in (0, 6, 2) 

left_start = 1
nlevels = 5

parameters = { 
               'prime': [1.0, 1.3, 1.50, 5],
               'power': [0.9, 1.1, 1.25, 5, 1.50],
             }

model = 'prime'

hash_count = {}
hash_count_valid = {}
hash_count_bounded = {}
hash_count_smooth = {}


print("\n--- Part 9: Build all short sequences starting from bottom\n")
gil.test_primes(parameters)
print()

for seed in range(0, 6, 2):
    # seed is denored as sigma in the paper

    full_hash = gil.get_all_sequences(seed, nlevels, left_start)

    for first_differences in full_hash:

        level = full_hash[first_differences]  
        key = (seed, level)
        gil.update_hash(hash_count, key, 1)
        valid, bounded, smooth, fvals, sequence = gil.admissibility(first_differences, parameters, model)
        if valid:
             gil.update_hash(hash_count_valid, key, 1)

        if valid and bounded:
            gil.update_hash(hash_count_bounded, key, 1)
            if smooth: 
                gil.update_hash(hash_count_smooth, key, 1)
                if level >= 2:
                    print("output part 9A", seed, level, sequence)

for key in hash_count:
    all = gil.get_hash(hash_count, key)
    valid = gil.get_hash(hash_count_valid, key)
    bounded = gil.get_hash(hash_count_bounded, key)
    smooth = gil.get_hash(hash_count_smooth, key)
    print("output part 9B:", key, all, valid, bounded, smooth)
\end{lstlisting}

\subsection{Gilbreath library}\label{lcode}

The code below features the \texttt{gilbreath\_lib.lib} library, used extensively in the main program in section~\ref{mcode}.
It is also on GitHub, \href{https://github.com/VincentGranville/Experimental-Math-Number-Theory/blob/main/Source-Code/gilbreath_lib.py}{here}. 
\vspace{1ex}

\begin{lstlisting}[numbers=left,escapechar=@,basicstyle=\ttfamily\footnotesize]
import numpy as np
import copy

#--- Utils

def make_readable(data, type = "int"):

    readable_data = []
    for elt in data:
        if type == "int":
            readable_data.append(int(elt))
        elif type == "float":
            readable_data.append(float(elt))
    return(readable_data)


def update_hash(hash, key, count):
    if key in hash:
        hash[key] += count
    else:
        hash[key] = count
    return(hash)


def get_hash(hash, key):
    if key in hash:
        count = hash[key]
    else:
        count = 0
    return(count)


def arr_plots(arr_k, arr_min, arr_max, arr_a, resets_list = ()):

    import matplotlib.pyplot as plt
    import matplotlib as mpl

    mpl.rcParams['axes.linewidth'] = 0.5
    plt.rcParams['xtick.labelsize'] = 8
    plt.rcParams['ytick.labelsize'] = 8
    plt.rcParams['legend.fontsize'] = 'x-small'

    arr_k = np.array(arr_k)
    arr_min = np.array(arr_min)
    arr_max = np.array(arr_max)
    arr_a = np.array(arr_a)

    from matplotlib.ticker import FormatStrFormatter  
    from matplotlib.ticker import MaxNLocator

    fig, (ax1, ax2, ax3) = plt.subplots(1, 3, figsize=(0.8*9, 0.8*3))

    ax1.yaxis.get_offset_text().set_fontsize(6)
    ax1.plot(arr_k, arr_min, lw=0.4, color='lightgreen')
    ax1.plot(arr_k, arr_max, lw=0.4, color='orange')
    ax1.plot(arr_k, arr_a, color='white', lw = 0.4)
    for k in resets_list:
        ax1.axvline(x=k, color='red', lw=0.4)
    ax1.ticklabel_format(style='sci', scilimits=(0, 0), axis='y', useMathText=True)
    ax1.yaxis.set_major_locator(MaxNLocator(nbins=5, prune='lower'))
    ax1.set_facecolor('black')
    ax1.tick_params(labelsize=6, pad=2, length=2)  

    ax2.plot(arr_k, arr_min/arr_a, lw=0.4, color='lightgreen') 
    ax2.plot(arr_k, arr_max/arr_a, lw=0.4, color='orange') 
    ax2.axhline(y=1.0, color='white', lw=0.4) 
    ax2.ticklabel_format(style='sci', scilimits=(0, 0), axis='y', useMathText=True)
    ax2.yaxis.set_major_formatter(FormatStrFormatter('%.1f'))
    ax2.yaxis.set_major_locator(MaxNLocator(nbins=5, prune='lower'))
    ax2.set_facecolor('black')
    ax2.tick_params(labelsize=6, pad=2, length=2) 
    ax2.set_ylim(0.2, 1.8)

    ax3.yaxis.get_offset_text().set_fontsize(6)
    ax3.plot(arr_k, arr_min-arr_a, lw=0.4, color='lightgreen')
    ax3.plot(arr_k, arr_max-arr_a, lw=0.4, color='orange')
    ax3.axhline(y=0.0, color='white', lw=0.4) 
    ax3.ticklabel_format(style='sci', scilimits=(0, 0), axis='y', useMathText=True)
    ax3.yaxis.set_major_locator(MaxNLocator(nbins=6, prune='lower'))
    ax3.set_facecolor('black')
    ax3.tick_params(labelsize=6, pad=2, length=2)

    plt.tight_layout()
    plt.subplots_adjust(wspace=0.15) 
    plt.show()
    return()


#--- banned prime gaps series and simulation

def banned_gaps_loop(gaps, start):

    val = start 
    seq = (val,)
    for k in range(0, len(gaps)):
        val += gaps[k]
        seq = (*seq, val) 

    banned = False

    for modulo in (3, 5, 7, 11, 13): 
        hash_residues = {}
        for val in seq:
            residue = val % modulo
            update_hash(hash_residues, residue, 1) 
        
        if len(hash_residues) == modulo:
            banned = True
            break

    return(banned)


def banned_gaps(gaps):

    banned = True
    for start in (2, 4, 6, 8, 10, 12):
        banned_start = banned_gaps_loop(gaps, start)
        if not banned_start:
            banned = False
            break
    return(banned)


def banned_found_in_seq(first_diff_seq):

    banned_found = False
    banned_pattern = ()

    for size in (2, 3, 4, 5, 6, 7):
        for k in range(len(first_diff_seq)-size):
            gaps = ()
            for j in range(size):
                gaps = (*gaps, first_diff_seq[k+j]) 
            banned = banned_gaps(gaps)
            if banned:
                banned_found = True
                banned_pattern = gaps
                break
        if banned_found:
            break
    return(banned_found, banned_pattern)


def generate_random_sequence(n, arr_init, llambda, parameters, model, rng, rejection_sampling = True): 

    # the generated sequence is guaranteed to be valid (no duplicates)

    q = arr_init[-1]  
    arr_a = np.copy(arr_init)
    sum = np.sum(arr_a)
    offset = len(arr_init)

    for k in range(n):  
      
        ccontinue = True

        while ccontinue: 
            # until in corridor
            
            u = 2 + 2*rng.poisson(lam=llambda)  # try increasing llambda
            new_q = q + u
            current_n = offset + k + 1

            admissible = admissibility_local(q, new_q, sum, current_n, parameters, model)
            if admissible or not rejection_sampling:
                ccontinue = False

        q += u
        arr_a = np.append(arr_a, q)

    return(arr_a)


#--- get canonical form of sequence arr_a

def canonical(s1, n_levels, fail_level = {}, valid = False, threshold = 1.0):

    # s1 is in the input sequence
    # s2 is the canonical form of s1 if nor failures added during construction
    # fail_level = {} for no failure
    # fail_level = {23:3, 41:5} to build fails at level 23, 41 resp. with 3, 5 instead of 1 (1 = no fail) 

    arr_a = np.copy(s1)
    np.random.seed(66)
    n = len(arr_a)
    arr_p =[-1]

    for level in range(1, n_levels+1, 1):

        arr1 = np.copy(arr_a)
        arr_a = []
        cnt = 0
        for k in range(len(arr1) - 1):
            a = abs(arr1[k+1] - arr1[k])
            arr_a.append(a)
            if arr_a[k] < arr1[k]:
                cnt += 1
        arr_a = make_readable(arr_a)
        p = cnt / len(arr_a)
        arr_p.append(p)
        if level == 1:
            s1_max_gap = np.max(arr_a)

    s1_bottom = np.copy(arr_a)

    for level in range(n_levels-1, -1, -1):

        # "reduced" primes is last arr_a produced (after adding 1)

        arr1 = np.copy(arr_a)
        if level in fail_level:
            fail_value = fail_level[level]
            arr_a = [fail_value]
        else:
            arr_a = [1]   

        for k in range(1, len(arr1)+1):

            a1 = arr_a[k-1] + arr1[k-1]
            a2 = arr_a[k-1] - arr1[k-1]   
            if a2 < 0:
                a = a1
            else:
                u = np.random.uniform(0,1)
                # random choice
                if u < threshold:  # try u < arr_p[level]
                    a = a2
                else:
                    a = a1
            if valid:
                if level == 0:
                    # make sure next prime is not less than previous prime
                    a = a1
                elif level == 1:
                    # make sure next prime is bigger than previous one
                    if a2 == 0:
                        a = a1
            arr_a.append(a)

        arr_a = make_readable(arr_a)
        if level == 1:
            s2_gaps = np.copy(arr_a)

    for k in range(len(arr_a)):
        arr_a[k] += 1
    s2 = np.copy(arr_a)
    s1_proba = make_readable(arr_p, type = "float")

    return(s2, s2_gaps, s1_bottom, s1_max_gap, s1_proba)


#--- find bounds for next prime

def augment(right_diagonal, a):

    new_diagonal = [a,]
    val = a
    success = False
    for j in range(len(right_diagonal)):  
        next = abs(val - right_diagonal[j])
        new_diagonal.append(next)
        val = next
    if val == 1:
        success = True
    return(success, new_diagonal)


def max_gap_allowed(right_diagonal, a):

    lower = a
    upper = 2*a + 1
    found = False
    trials = 0

    while not found:
        trials += 1
        b = (upper + lower) // 2 
        if b % 2 == 0:
            b += 1
        success, new_diagonal = augment(right_diagonal, b)
        if success:
            lower = b
        else:
            upper = b 
        if upper - lower == 2 or trials > 28:
            found = True

    return(b, trials)


def min_gap_allowed(right_diagonal, a):

    lower = 1
    upper = a
    found = False
    trials = 0

    while not found:
        trials += 1
        b = (upper + lower) // 2 
        if b % 2 == 0:
            b += 1
        success, new_diagonal = augment(right_diagonal, b)
        if success:
            upper = b
        else:
            lower = b 
        if upper - lower == 2 or trials > 280:
            found = True
        
    return(b, trials)


def test_conjecture(old_a, max_a, right_diagonal):

    satisfied = True
    flag = 'no issue'

    # at 'a = max + 2', we must fail, it's above the max limit
    success1, test_diagonal = augment(right_diagonal, max_a + 2)
    success1 = bool(not success1)

    if success1:
        # check values below max_a; they must all succeed
        for a in range(old_a, max_a, 2):  
            success2, test_diagonal = augment(right_diagonal, a)
            if not success2:
                flag = (old_a, a, max_a)
                satisfied = False
                break
    else: 
        flag = 'a_max + 2 succeeds, not supposed to'
        satisfied = False
    return(satisfied, flag)


def collect_statistics(right_diagonal):

    length = len(right_diagonal)
    idx = -1
    idx2 = -1
    end = -1
    reversal_idx = - 1
    reversal_value = -1
    reversal = False
    reset = False
    hash_output = {}

    for kx in range(length-1, 0, -1):
        if right_diagonal[kx] > 2:
            idx = kx
            end = right_diagonal[kx]
            break
    if idx == -1:
        reset = True 

    dmax2 = np.max(right_diagonal[1:])
    for kx in range(length-1, 0, -1):
        if right_diagonal[kx] == dmax2:
            dmax2_idx = kx
            break

    for kx in range(len(right_diagonal[1:])):
        d1 = right_diagonal[kx]
        d2 = right_diagonal[kx+1]
        if d2 >= d1 and d2 > 2:
            reversal_idx = kx + 1
            reversal_value = d2
            reversal = True
            break

    # dmax is maximum value after the first log(length)
    start = int(np.log(length))
    dmax = np.max(right_diagonal[start:]) 

    counts = ()
    for kx in (0, 2):
        counts = (*counts, right_diagonal.count(kx))
    
    arr_cycle02 =right_diagonal[idx+1:-2]
    nu_2 = arr_cycle02.count(2)

    hash_output['reset'] = reset
    hash_output['counts'] = counts
    hash_output['max_value'] = dmax2
    hash_output['max_index'] = dmax2_idx
    hash_output['before_02_value'] = end
    hash_output['before_02_index'] = idx
    hash_output['max_tail_value'] = dmax
    hash_output['reversal_idx'] = reversal_idx
    hash_output['reversal_value'] = reversal_value
    hash_output['reversal'] = reversal
    hash_output['nu_2'] = nu_2

    return(hash_output)


#--- Main function

def check(arr_a, show_triangle = 'None', mode = 'Fast', delta = 'standard'):

    hash_stats = {}
    right_diagonal = []
    left_diagonal = []
    arr_amax = []
    right_diagonal.append(arr_a[-1])
    arr1 = np.copy(arr_a)
    arr1 = np.array(arr1).astype(int)
    if show_triangle in ('Full',): 
        print("Sequence:", arr_a) 
    imax = len(arr_a) - 1
    flag = 'Success'
    level = 0
    failing_level = -1

    for k in range(imax, 0, -1):

        level += 1
        arr2 = np.zeros(k)
        for j in range(k):
            if delta == 'standard':
                arr2[j] = abs(arr1[j+1] - arr1[j])
            elif delta == 'light':
                arr2[j] = max(abs(arr1[j+1] - arr1[j]), 1)
        if k == imax:
            first_diff = np.copy(arr2) 
        right_diagonal.append(arr2[-1])
        left_diagonal.append(arr2[0])

        amax = np.max(arr2)  
        first_amax_idx = 1 + np.argmax(arr2 == amax) 
        first_amax_val = int(arr1[first_amax_idx - 1])
        amax = int(amax)
        arr1 = np.copy(arr2)
        arr1 = np.array(arr1).astype(int)

        count = np.count_nonzero(arr1 == amax)
        arr_amax.append([level, amax, count]) 
        if arr1[0] != 1:
            flag = 'Fail'
            failing_level = level 
            if not mode == 'Full':
                break
        if show_triangle == 'Compact':
            print("Level",level,"|amax=",amax,"|idx=",first_amax_idx,"|val=",first_amax_val,"|count=",count) 
        elif show_triangle == 'Full':
            print("Level",level,"|amax=",amax,"|idx=",first_amax_idx,"|val=",first_amax_val,arr1) 
        if amax == 2:
            flag = 'Success'
            if mode == 'Fast':
                break

    right_diagonal = np.array(right_diagonal).astype(int)
    left_diagonal = np.array(left_diagonal).astype(int)
    first_diff = np.array(first_diff).astype(int)

    hash_stats['arr_amax'] = arr_amax
    hash_stats['first_diff'] = first_diff
    hash_stats['right_diagonal'] = right_diagonal
    hash_stats['left_diagonal'] = left_diagonal
    hash_stats['failing_level'] = failing_level

    return(flag, first_amax_idx, hash_stats)


def admissibility_local(q, proposed_next_q, sum, n, parameters, model):

    # testing a new value; assumes it is admissible prior to that

    constants = parameters[model]
    lower_constant = constants[0]  
    upper_constant = constants[1]
    smooth_constant = constants[2]

    old_sum = sum
    sum += proposed_next_q

    if model == 'power':
        exponent = constants[4]
        growth = n**exponent 
    elif model == 'prime':
        growth = n*np.log(n)
    lower = lower_constant * growth
    upper = upper_constant * growth 

    pass_test1 = bool(lower < proposed_next_q < upper) 
    pass_test2 = bool(proposed_next_q < smooth_constant * q) 

    if pass_test1 and pass_test2:
        admissible = True
    else:
        admissible = False
    return(admissible)


def admissibility(first_differences, parameters, model): 

    # check if sequence is in corridor; test the whole sequence
    # n = index of first failure
    # fvals = (n lower bound, observed value, upper bound) if failing bound test
    # fvals = (n, observed value, max allowed given past value) if failing smooth test

    valid = True
    bounded = True
    smooth  = True
    fvals = ()

    if 0 in first_differences:

        valid = False
        sequence = ()

    else:

        constants = parameters[model]
        lower_constant = constants[0]  
        upper_constant = constants[1]
        smooth_constant = constants[2]
        offset = constants[3]

        sum = 2
        sequence = (sum,)
        n = 1

        for value in first_differences:

            n += 1
            old_sum = sum
            sum += value
            if model == 'power':
                exponent = constants[4]
                growth = n**exponent 
            elif model == 'prime':
                growth = n*np.log(n)
            lower = lower_constant * growth
            upper = upper_constant * growth 

            pass_test1 = bool(lower < sum < upper)
            if not pass_test1 and n > offset: 
                bounded = False
                if len(fvals) == 0:
                    fvals = (n, lower, sum, upper)
            pass_test2 = bool(sum < smooth_constant * old_sum)
            if not pass_test2 and n > offset:
                smooth = False
                if len(fvals) == 0:
                    fvals = (n, lower, smooth_constant * old_sum)
            sequence = (*sequence, sum)

    return(valid, bounded, smooth, fvals, sequence)


def test_primes(parameters):

    # check if prime number sequence (first 1000 primes) is in corridor
    from primePy import primes
    prime_list = primes.first(1000)
    # prime_list = (2, 3, 5, 7, 11, 13, 17, 19, 23, 29, 31) 
 
    model = 'prime'
    first_differences = ()
    for k in range(len(prime_list) - 1):
        delta = prime_list[k+1] - prime_list[k] 
        first_differences = (*first_differences, delta)
    valid, bounded, smooth, fvals, sequence = admissibility(first_differences, parameters, model)
    print("Testing prime sequence: bounded = %s | smooth = %s " %(bounded, smooth))
    return()


def integrate(x, left_start = 1):

    # find all sequences y one level above x such that diff(y) = x

    y = (left_start,)
    list2 = (y,)

    for k in range(1, 1+len(x)):

        list3 = ()  # used to be {}

        for y in list2:  

            # we must have if len(y) == k:
    
            y_k_plus = y[k-1] + x[k-1]
            new_y_plus = (*y, y_k_plus)
            list3 = (*list3, new_y_plus)

            if x[k-1] !=0: 
                y_k_minus = y[k-1] - x[k-1]
                if y_k_minus >= 0:
                    new_y_minus = (*y, y_k_minus)
                    list3 = (*list3, new_y_minus)

        list2 = copy.deepcopy(list3)

    return(list2)


def get_all_sequences(seed, nlevels, left_start = 1):

    hash_arr = { (left_start, seed):1 }
    full_hash = {}

    for level in range(nlevels):

        hash3 = {}
        count = len(hash_arr)
        sub_count = 0

        for arr1 in hash_arr:

            if sub_count % 100 == 0 and sub_count > 0:
                print("GBL Progress",seed, level, sub_count,"/", count) 
            sub_count += 1
            local_list = integrate(arr1, left_start)

            # Add the collected local_lists to list2 
            for arr2 in local_list:
                hash3[arr2] = level + 1 

        hash_arr = hash3.copy()
        full_hash.update(hash_arr)  

    return(full_hash)
\end{lstlisting}

%-----------------------------

%-------------

%\printindex
%\pagebreak
%\tableofcontents

\hypersetup{linkcolor=red} % red %
\hypersetup{linkcolor=red}

\end{document}